\documentclass[sigconf]{acmart}
\PassOptionsToPackage{noend}{algorithmic} \PassOptionsToPackage{svgnames,table}{xcolor}
\AtBeginDocument{}

\usepackage{algorithm}
\usepackage{algorithmic}

\usepackage{listings}
\usepackage{cleveref}

\usepackage{xspace}

\def\|#1|{\textit{#1}}
\def\<#1>{\texttt{#1}}

\usepackage[svgnames,table]{xcolor}
\usepackage{soul}
\usepackage[T1]{fontenc}
\usepackage{amsmath,amsthm,amsfonts,mathtools}
\usepackage{alltt}
\usepackage{multirow}
\usepackage{pdflscape}
\usepackage{pgf}
\usepackage{pifont}
\usepackage{qtree}
\usepackage[english,rounding]{rccol}
\usepackage{rotating}
\usepackage{setspace}
\usepackage{tabularx}
\usepackage{tipa}
\usepackage{wrapfig}
\usepackage{tasks}
\usepackage{xstring}
\usepackage[tikz]{mdframed}
\usepackage{environ}
\usepackage{etoolbox}
\usepackage{fourier-orns}
\usepackage{kvoptions}
\usepackage[]{units}
\usepackage{forest}
\usepackage{tikz}
\usetikzlibrary{automata, arrows.meta, positioning}
\usepackage[skip=1ex]{caption}
\usepackage{ragged2e}
\usepackage{listings}
\usepackage{xspace}
\usepackage{mathtools}
\usepackage{algorithm}
\usepackage{subcaption}

\rcDecimalSignOutput{.}

\newcommand{\pre}{\text{pre}}
\newcommand{\stored}{\text{stored}}
\newcommand{\suff}{\text{suff}}
\newcommand{\row}{\text{row}}

\newcommand{\PP}{\text{ProperPrefix}}

\newcommand{\obstable}{(S_{\pre}, S_{\suff}, S_{\stored}, S_{\deadpref} , T)}
\newcommand{\A}{\mathcal{A}}

\newcommand{\member}{\textbf{MEMBER}\xspace}
\newcommand{\livepref}{\textbf{LIVE-PREF}\xspace}
\newcommand{\deadpref}{\textbf{DEAD-PREF}\xspace}
\newcommand{\y}{\textbf{YES}\xspace}
\newcommand{\n}{\textbf{NO}\xspace}

\newcommand{\prefix}{\text{Prefix}}
\newcommand{\positive}{\textbf{POSITIVE}\xspace}
\newcommand{\negative}{\textbf{NEGATIVE}\xspace}
\newcommand{\extend}{\textbf{EXTEND}\xspace}
\newcommand{\restart}{\textbf{RESTART}\xspace}
\newcommand{\algorithmicbreak}{\textbf{break}}
\newcommand{\ner}{\mathcal{N}}
\newcommand{\cl}{\mathcal{C}}
\newcommand{\arith}{\texttt{arith}\xspace}
\newcommand{\json}{\texttt{json}\xspace}
\newcommand{\dyck}{\texttt{dyck}\xspace}
\newcommand{\dat}{\texttt{date}\xspace}

\renewcommand{\algorithmiccomment}[1]{\hfill$\triangleright$ {#1}}

\newcolumntype{P}{>{\raggedleft\arraybackslash}p{(\textwidth-1in)/5}}
\newcolumntype{Q}{>{\raggedleft\arraybackslash}p{(\textwidth)/4}}

\usepackage{filecontents}
\usepackage{arydshln}

\begin{document}
\newtheorem{assumption}[theorem]{Assumption}
\newcounter{example}[section]
\newtheorem{property}[theorem]{Property}

\title{
Example-Free Learning of Regular Languages with Prefix Queries}
\author{Eve Fernando}
\affiliation{\institution{University of Sydney}
  \country{Australia}
}
\author{Sasha Rubin}
\affiliation{\institution{University of Sydney}
  \country{Australia}
}
\author{Rahul Gopinath}
\affiliation{\institution{University of Sydney}
  \country{Australia}
}

\renewcommand{\shortauthors}{et al.}

\begin{abstract}
Language learning refers to the problem of inferring a mathematical model which accurately represents a formal language. Many language learning algorithms learn by asking certain types of queries about the language being modelled. Language learning is of practical interest in the field of cybersecurity, where it is used to model the language accepted by a program's input parser (also known as its input processor). In this setting, a learner can only query a string of its choice by executing the parser on it, which limits the language learning algorithms that can be used. Most practical parsers can indicate not only whether the string is valid or not, but also \emph{where} the parsing failed. This extra information can be leveraged into producing a type of query we call the \emph{prefix query}. Notably, no existing language learning algorithms make use of prefix queries, though some ask \emph{membership} queries i.e., they ask whether or not a given string is valid. When these approaches are used to learn the language of a parser, the prefix information provided by the parser remains unused. 

In this work, we present PL*, the first known language learning algorithm to make use of the prefix query, and a novel modification of the classical L* algorithm proposed by \citet{RN6}. We show both theoretically and empirically that PL* is able to learn more efficiently than L* due to its ability to exploit the additional information given by prefix queries over membership queries. Furthermore, we show how PL* can be used to learn the language of a parser, by adapting it to a more practical setting in which prefix queries are the only source of information available to it; that is, it does not have access to any labelled examples or any other types of queries. We demonstrate empirically that, even in this more constrained setting, PL* is still capable of accurately learning a range of languages of practical interest.
\end{abstract}

\maketitle

\section{Introduction}
\label{sec:intro}

Language learning refers to the task of automatically inferring a mathematical model, typically an automaton or grammar, which accurately represents a formal language. Learning formal languages, namely sets of strings, is of theoretical interest in computational learning theory, as well as of practical interest in the field of cybersecurity, where it is used to model the \emph{input language} of a program. A program's \emph{input language} refers to the set of all valid inputs to the program; in other words, the set of all strings accepted by the program's \emph{input parser} (also known as its \emph{input processor}). A formal model of the program's input language (synthesised using language learning techniques) can be used in a number of ways, for example to analyse and validate the input parser itself \citep{RN88}, or to automatically generate test cases for the program \citep{RN86, RN87}.

To help a language learning algorithm construct a model of the program's input language, it is provided with some information about this language. If the learning algorithm is given a set of labelled examples, it is said to be a \emph{passive learner}, or an algorithm which \emph{learns from examples}. Conversely, if it is permitted to directly ask queries from the input parser, and furthermore to adapt its learning (i.e., choose which queries to ask) depending on answers to previous queries, it is said to be an \emph{active learner}.

Let us provide a simple example of this. Suppose we wish to model the language of a program's input parser which claims to parse JSON strings (JSON is a text-based format primarily used for transmitting structured data across a network). Assume we only have \emph{blackbox} access to this parser; that is, we do not have access to its source code (for example, we might only have access to the binary file). We know the parser implements some variant of the official JSON specification \citep{RN85}, but we do not know exactly how it deviates from this specification. So, we decide to use language learning to model the language of the JSON parser; that is, we wish to construct an automaton or grammar which accepts strings the parser successfully parses, and rejects strings the parser fails to parse. By analysing this model, we should be able to gain insight into possible discrepancies between the parser and the official JSON language, to the extent that the model is an accurate reflection of the parser. Let's suppose the learning algorithm we have chosen to use is a passive learner. In that case, we might provide it with some example strings, with each string labelled as either valid (i.e., would be successfully parsed by the parser) or invalid. These labelled examples may be sourced from the program's documentation, for example. Alternatively, if the learning algorithm is an active learner, we may allow it to execute the parser on any string of its choice. See Figure~\ref{fig:json-example} for a schematic of this.

\begin{figure*}[t]
\centering
\includegraphics[width=\textwidth]{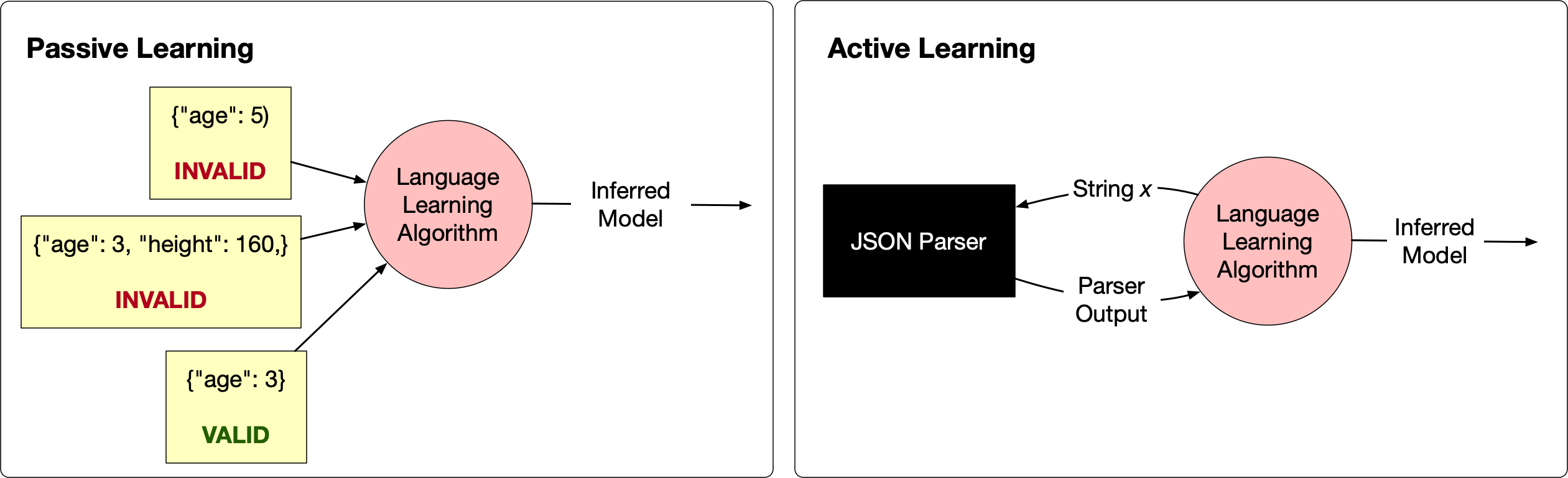}
\caption{Language learning for modelling a JSON parser}\label{fig:json-example}
\end{figure*}

In practice, many input parsers indicate not only whether an input is valid or not, but also whether it is a \emph{prefix} of some valid input or not. Hence, when a learning algorithm executes an input parser on a given string, we can say that it is asking a \emph{prefix query}. Notably, no existing language learning algorithms make use of prefix queries, though some make use of \emph{membership queries} which merely indicate whether or not an input is valid. Hence, when an algorithm that relies on membership queries is used to learn the language of a parser, the prefix information naturally provided by the parser is ignored. Thus, in this work, we seek to investigate how the information provided by prefix queries can be exploited in the field of language learning.

\noindent\textbf{Contributions.}
In this work, we present PL*, the first known language learning algorithm to make use of the prefix query and a novel modification of the classical L* algorithm. Specifically, whilst L* uses equivalence queries and membership queries, PL* uses equivalence queries and \emph{prefix} queries. We prove that PL* always outperforms L* in terms of the number of queries asked, when both algorithms are tasked with learning the same non-dense language (\emph{non-dense} is defined in Definition~\ref{definition-dense}; note that most practical languages are non-dense). Furthermore, we show empirically that PL* significantly outperforms L* both in terms of runtime and space usage. These results demonstrate that exploiting prefix queries can increase efficiency. We then show how PL* can be applied to the problem of learning the language of a parser, namely by adapting it to a more practical setting in which only prefix queries can be asked (and equivalence queries cannot be). Notably, in this setting, PL* requires only blackbox access to a parser (so that it can ask prefix queries); it does not require labelled examples or any other information about the language of the parser. We show that, even in this new, more constrained setting, PL* is still capable of efficiently and accurately learning a range of languages of practical interest.

The remaining paper is as follows:
In \Cref{chap:background}, we introduce some relevant terminology and preliminary notions. We begin by providing some useful definitions relating to sets and formal languages, before discussing finite automata and grammars. We then present a more detailed definition of language learning, before describing the main learning models used by the algorithms in this field. Finally, we conclude with a more formal definition of the prefix query, followed by a discussion of how it naturally arises from parsers.

In \Cref{chap:problem-statement}, we provide a clear statement of the problem we are aiming to address in this work.

In \Cref{chap:related-work}, we describe the key developments in the field of language learning, noting that no existing language learning algorithm makes use of the prefix query. This aptly sets the scene for \Cref{chap:method}, in which we introduce PL*, the first known language learning algorithm to make use of the prefix query and a novel modification of the L* algorithm proposed by \citet{RN6}. After providing a detailed description of the PL* algorithm, a proof of its correctness and termination, and a theoretical analysis of the improvements it gives over L*, we describe how it can be applied to the problem of learning the language of a parser.

In \Cref{chap:empirical-evaluation}, we empirically demonstrate the significant improvements given by PL* over L*, both in terms of runtime and space usage. We then proceed show how, when PL* is applied to the problem of modelling a parser, namely by adapting it to a more practical setting in which equivalence queries cannot be answered, it still is able to efficiently and accurately learn various languages of practical interest. 

In \Cref{chap:discussion}, we close this work with a final discussion of our results thus far and potential directions for future work.

\section{Background}
\label{chap:background}

In this section, we introduce some useful terminology and notation.

\subsection{Sets, Alphabets and Languages}\label{set-alphabet-language}

Let $A, B$ be arbitrary sets. Then, let $A \oplus B$ denote the \emph{symmetric difference} of $A$ and $B$; that is,

$$A \oplus B = (A \cup B) \setminus (A \cap B)$$

If $A$ is an arbitrary set, then $\mathcal{P}(A)$, referred to as the \emph{power set} of $A$, denotes the set of all subsets of $A$ (including $\varnothing$ and $A$ itself).

An \emph{alphabet} is a non-empty finite set whose elements are referred to as \emph{symbols}. A \emph{string} $s$ over alphabet $A$ refers to a finite sequence of symbols from $A$, where the number of symbols in a string, denoted by $|s|$, is referred to as its \emph{length}. For example, $0011$ is a string of length 4 over the alphabet $\{0, 1\}$. The set of all strings over the alphabet $A$ is denoted by $A^*$, and the set of all strings over $A$ of non-zero length is denoted by $A^+$. A \emph{language} $L$ refers to a (possibly-empty) subset of $A^*$. For example, $\{01, 00\}$ is a language over the alphabet $\{0, 1\}$.

\begin{definition}[Proper Prefix Language]\label{definition-pp}
    For some $L \subseteq A^*$, let $\PP(L)$, referred to as the \textup{proper prefix language} of $L$, be defined as follows:

    $$\PP(L) = \{x \in A^* \mid x \not\in L \text{ and } xy \in L \text{ for some $y \in A^+$}\}$$
\end{definition}

Intuitively, $\PP(L)$ is the set of all strings that are not in $L$, but can be extended to form a string in $L$.

Using this notation, we can now define several properties of strings in $A^*$ with respect to $L$. 

\begin{definition}[String Properties]\label{definition-string-prop}
    For arbitrary $x \in A^*$, 

    \begin{itemize}
        \item if $x \in L$, then $x$ is referred to as a \textup{member} of $L$;
        \item if $x \not\in L$ but $x \in \PP(L)$, then $x$ is a \textup{live prefix} of $L$; and
        \item if $x \not\in L$ \emph{and} $x \not\in \PP(L)$, then $x$ is called a \textup{dead prefix} of $L$.
    \end{itemize} 
\end{definition}

\begin{definition}[Prefix Language]\label{definition-pref-lang}
    For some $L \subseteq A^*$, let $\prefix(L)$, referred to as the \textup{prefix language of $L$}, be the set of all prefixes of strings in $L$; that is,

    $$\prefix(L) = L \cup \PP(L)$$
\end{definition}

\begin{definition}[Extension of a String]\label{definition-ext}
    Let $x \in A^*$ be arbitrary. If $y = xz$, for $z \in A^*$, then $y$ is said to be an \textup{extension} of $x$. If $z \neq \varepsilon$, then $y$ is said to be a \textup{proper extension} of $x$.
\end{definition}

\begin{definition}[Dense Language]\label{definition-dense}
    Let a \textup{dense language} be a language $L \subseteq A^*$ where, for all $x \not\in L$, $x \in \PP(L)$. All regular languages over $A^*$ which are \emph{not} dense will be referred to as \textup{non-dense languages}.
\end{definition}

Intuitively, a dense language $L \subseteq A^*$ is a language satisfying the following: every string not in $L$ can be extended to a string in $L$. Some examples of dense languages are the set of all strings of even length, and the set of all strings of odd length.

Many languages of practical interest are non-dense. For example, all the languages we will consider in \Cref{chap:empirical-evaluation} (Empirical Evaluation) are non-dense.

\subsection{Finite Automata}\label{finite-automata}

A \emph{finite automaton} can be thought of as a computer program which takes in an input string, processes it letter-by-letter (using variables with finite domains), and, after consuming it fully, either accepts or rejects it. There are several classes of finite automata, with the two most widely-used being the \emph{deterministic finite automata} and \emph{nondeterministic finite automata}.

\subsubsection{Deterministic Finite Automata}\label{finite-automata-dfa}

\begin{definition}[Deterministic Finite Automaton]\label{definition-dfa}
    A \textup{deterministic finite automaton} (DFA) $M$ is a 5-tuple $(Q, A, \delta, q_0, F)$, where $Q$ is a finite set of \textup{states}, $A$ is the alphabet, $\delta: Q \times A \rightarrow Q$ is the \textup{transition function}, $q_0 \in Q$ is the \textup{initial state}, and $F \subseteq Q$ is the set of \textup{final states}.
\end{definition}

Intuitively, when $M$ is in some state $q \in Q$ and reads some $a \in A$, it transitions to state $\delta(q, a)$.

For arbitrary $q \in Q$ and $w \in A^*$, let $\delta^*(q, w)$ denote the state which $M$ ends up in if it begins from state $q$ and reads the string $w$.

Suppose that $M$ reads in some $x \in A^*$. Beginning from $q_0$, $M$ moves through a sequence of states as it reads $x$ symbol-by-symbol, as determined by $\delta$. This sequence of states is referred to as a \emph{run of $M$ on $x$}, and can be denoted as $q_0 \xrightarrow{x_1} q_1 \xrightarrow{x_2} q_2 \xrightarrow{x_3} \dots \xrightarrow{x_n} q_n$, where $x = x_1x_2x_3 \dots x_n$ and each $x_i \in A$. If $q_n \in F$, we say that the run is \emph{accepting}. Due to a DFA's deterministic nature, for each $x \in A^*$, there is exactly \emph{one} run of $M$ on $x$. If the run of $M$ on $x$ is accepting, then we say that $M$ \emph{accepts} $x$. Equivalently, if $\delta^*(q_0, x) \in F$, we once again say that $M$ \emph{accepts} $x$. The set $L(M)$, called the \emph{language of $M$}, denotes the set of all strings accepted by $M$ i.e., $L(M) = \{x \in A^* \mid \delta^*(q_0, x) \in F\}$. This can be generalised to the notion of a \emph{state language} i.e., for any $q \in Q$, the \emph{state language of $q$}, denoted by $L(M; q)$, is defined as follows:

$$L(M; q) = \{w \in A^* \mid \delta^*(q, w) \in F\}$$

A language is called \emph{regular} if some DFA recognises it.

\begin{example}
    Figure~\ref{fig:dfa-example} is a simple example of a DFA $M = (Q, A, \delta, q_0, F)$, with $Q = \{q_0, q_1\}$, $A = \{0, 1\}$, start state $q_0$, $F = \{q_1\}$, and $\delta$ as defined in Table~\ref{tab:dfa-example-delta}.
        
    \begin{table}[ht]
    \begin{center}
    \begin{tabular}{  | c | c | c |  } 
        \hline
        $\delta$ & $0$ & $1$ \\
        \hline
        $q_0$ & $q_1$ & $q_0$ \\
        \hline
        $q_1$ & $q_1$ & $q_0$ \\
        \hline
    \end{tabular}
    \caption{\label{tab:dfa-example-delta}Transition table for example DFA $M$}
    \end{center}
    \end{table}
    
    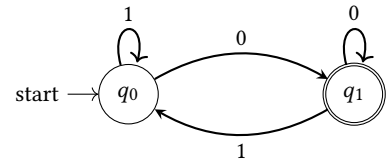
\begin{figure}[!htb]
    \centering
    \begin{tikzpicture} [node distance = 3cm, on grid, auto]
        \node (q0) [state, initial] {$q_0$};
        \node (q1) [state, right = of q0, double] {$q_1$};
        
        \path [-stealth, thick]
        (q0) edge [loop above]  node {$1$} ()
        (q0) edge [bend left] node {$0$} (q1)
        (q1) edge [loop above]  node {$0$} ()
        (q1) edge [bend left] node {$1$} (q0)
        ;
    \end{tikzpicture}
    \caption{A simple DFA example} \label{fig:dfa-example}
    \end{figure}
    
    If $M$ takes in the string $001$ as input, then, beginning from $q_0$, it transitions to $q_1$ upon reading the first $0$. Upon reading the second $0$, $M$ remains in $q_1$. Finally, upon reading the $1$, $M$ transitions back to $q_0$, which is not a final state. Thus, $001$ is rejected by $M$ i.e., $001 \not\in L(M)$.
    
    What if $M$ took in the string $010$ as input? Then, beginning from $q_0$, it would first transition to $q_1$ upon reading $0$. Upon reading the $1$, $M$ would then transition back to $q_0$. Finally, upon reading the second $0$, $M$ would transition to $q_1$ again, which \emph{is} a final state. Hence, $010$ is accepted by $M$ i.e., $010 \in L(M)$.
    
    After some thought, we can infer that $M$ accepts all strings ending with the character $0$.
\end{example}

\subsubsection{Minimal Deterministic Finite Automata}\label{finite-automata-minimal-dfa}

Let $L \subseteq A^*$ be an arbitrary regular language. Let $x, y \in A^*$ also be arbitrary. A string $z \in A^*$ is said to be a \emph{distinguishing string} of $x$ and $y$ if exactly one of $xz$ and $yz$ is in $L$. 

\begin{definition}[The Nerode Congruence]\label{definition-nerode-congruence}
    Let $\equiv_L^\ner$ denote the \textup{Nerode congruence} for $L$, where $x \equiv_L^\ner y$ if and only if there is no distinguishing string for $x$ and $y$.
\end{definition}

The congruence $\equiv_L^\ner$ is an equivalence relation on strings in $A^*$ -- and so, for arbitrary $x \in A^*$, we can use $[x]$ to denote the equivalence class of $x$ with respect to $\equiv_L^\ner$:

$$[x] = \{y \in A^* \mid x \equiv_L^\ner y\}$$

\begin{definition}[The Nerode Automaton]\label{definition-nerode-automaton}
    The \textup{Nerode automaton} recognising $L$, which we will denote by $M_L = (Q, A, \delta, q_0, F)$, is defined as follows:

    \begin{align*}
        Q &= \{[s] \mid s \in A^*\} \\
        q_0 &= [\varepsilon] \\
        F &= \{[s] \mid s \in L\} \\
        \delta([s], a) &= [sa]
    \end{align*}
\end{definition}

Since $L$ is regular, $Q$ is finite. In fact, if $L$ were non-regular, then $Q$ would be \emph{infinite}. \citet{RN72} prove that $M_L$ is the \emph{minimal DFA} recognising $L$ i.e., no other DFA recognising $L$ has fewer states than $M_L$.

We can also define the minimal DFA recognising $L$ in terms of \emph{residual languages}. This notion will be useful when we describe various automata learning algorithms in ~\Cref{chap:related-work} (Related Work). 

\begin{definition}[Residual Language]\label{definition-residual}
    The \textup{residual language} (or \textup{left derivative}) of $L$ with respect to some $x \in A^*$, denoted by $x \backslash L$, is defined as follows:

    $$x \backslash L = \{w \in A^* \mid xw \in L\}$$
\end{definition}

The connection between the Nerode congruence $\equiv_L^\ner$ and residual languages is as follows:

$$u \equiv_L^{\ner} v \hspace{3pt} \Longleftrightarrow \hspace{3pt} u \backslash L = v \backslash L$$

We can now use this notion of residual languages to define an automaton by setting

\begin{align*}
    Q &= \{s \backslash L \mid s \in A^*\} \\
    q_0 &= L \\
    F &= \{s \backslash L \mid s \in L\} \\
    \delta(X, a) &= a \backslash X
\end{align*}

This is again the minimal DFA recognising $L$.

\subsubsection{Non-Deterministic Finite Automata}\label{finite-automata-nfa}

\begin{definition}[Non-Deterministic Finite Automaton]\label{definition-nfa}
    A \textup{non-deterministic finite automaton} (NFA) $N$ is a 5-tuple $(Q, A, \delta, q_0, F)$. It is exactly the same as a DFA, except that $\delta: Q \times A_{\varepsilon} \rightarrow \mathcal{P}(Q)$ is not the transition \emph{function}, but instead the \textup{transition relation}.
\end{definition}

Intuitively, when $N$ is in some state $q \in Q$ and reads some $a \in A$, there is no \emph{single} state that $N$ must transition to (as is the case in a DFA); rather, there is a \emph{set} of states it could transition to i.e., we allow non-determinism. Furthermore, $A_{\varepsilon} = A \cup \{\varepsilon\}$, and so in NFAs we also allow \emph{epsilon transitions} i.e., transitions that do \emph{not} consume an input symbol. 

Suppose that $N$ reads in some $y \in A^*$. Beginning from $q_0$, $N$ moves through a sequence of states as it reads $y$ symbol-by-symbol, based on $\delta$. This sequence of states is referred to as a \emph{run of $N$ on $y$}, and can be denoted as $q_0 \xrightarrow{y_1} q_1 \xrightarrow{y_2} q_2 \xrightarrow{y_3} \dots \xrightarrow{y_m} q_m$, where $y = y_1y_2y_3 \dots y_m$ and each $y_i \in A_{\varepsilon}$. If $q_m \in F$, we say that the run is \emph{accepting}. Unlike with a DFA, there may be more than one run of $N$ on $y$. If $y$ has \emph{at least one} accepting run, then we say that $N$ \emph{accepts} $y$. The set $L(N)$, called the \emph{language of $N$}, denotes the set of all strings accepted by $N$ i.e., $L(N) = \{y \in A^* \mid \text{ $N$ accepts $y$}\}$.

A DFA is an NFA $N = (Q, A, \delta, q_0, F)$ that satisfies the following: for every state $q \in Q$ and symbol $a \in A$, $|\delta(q, a)| = 1$. Conversely, any NFA can be converted into an equivalent DFA (i.e., a DFA which accepts exactly the same language) using the \emph{subset construction} described by \citet{RN72}.

\begin{example}
    Figure~\ref{fig:nfa-example} is a simple example of an NFA $N = (Q, A, \delta, q_0, F)$, with $Q = \{q_0, q_1, q_2\}$, $A = \{0, 1\}$, start state $q_0$, $F = \{q_2\}$, and $\delta$ as defined in Table~\ref{tab:nfa-example-delta}.

    \begin{table}[ht]
    \begin{center}
    \begin{tabular}{  | c | c | c |  } 
        \hline
        $\delta$ & $0$ & $1$ \\
        \hline
        $q_0$ & $\{q_0, q_1\}$ & $\{q_0\}$ \\
        \hline
        $q_1$ & $\{q_2\}$ & $\{q_2\}$ \\
        \hline
        $q_2$ & $\varnothing$ & $\varnothing$ \\
        \hline
    \end{tabular}
    \caption{\label{tab:nfa-example-delta}Transition table for example NFA $N$}
    \end{center}
    \end{table}
    
    \begin{figure}[!htb]
    \centering
    \begin{tikzpicture} [node distance = 3cm, on grid, auto]
        \node (q0) [state, initial] {$q_0$};
        \node (q1) [state, right = of q0] {$q_1$};
        \node (q2) [state, right = of q1, double] {$q_2$};
        
        \path [-stealth, thick]
        (q0) edge [loop above]  node {$0, 1$} ()
        (q0) edge [above] node {$0$} (q1)
        (q1) edge [above] node {$0, 1$} (q2)
        ;
    \end{tikzpicture}
    \caption{A simple NFA example} \label{fig:nfa-example}
    \end{figure}
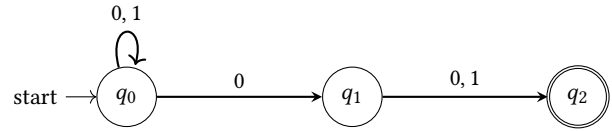
    
    If the cell in row $q \in Q$ and column $a \in A$ in Table~\ref{tab:nfa-example-delta} has value $\varnothing$, this means that, if $N$ is in state $q$ and then reads the symbol $a$, it crashes.
    
    If $N$ takes in the string $01$ as input, then, beginning from $q_0$, it can choose to transition to $q_1$ upon reading $0$. Upon reading the $1$, it will then transition to $q_2$, which \emph{is} a final state. Hence, $01$ is accepted by $N$ i.e., $01 \in L(N)$.
    
    What if $N$ takes in the string $110$ as input? After some thought, we can see that, in order for a run of $N$ on some $x \in A^*$ to be accepting, the second-last character of $x$ must be $0$. If the second-last character of $x$ is $1$, a run of $N$ on $x$ would either (1) not reach $q_2$ at all, or (2) reach $q_2$ but eventually crash. Thus, because the second-last character of $110$ is $1$, $110$ is rejected by $N$ i.e., $110 \not\in L(N)$.
\end{example}
\subsection{Grammars}\label{grammars}

A \emph{grammar} refers to a set of rules that can be used to generate strings. As with the finite automata, there are various classes of grammars, however we will focus on the class of \emph{context-free grammars} as these are the most widely-used in the field of language learning.

\begin{definition}[Context-Free Grammar]\label{definition-cfg}
    A \textup{context-free grammar} (CFG) $G$ is a 4-tuple $(V, A, R, S)$, where $V$ is a finite set of \textup{variables}, $A$ is a finite set of \textup{terminals}, $R$ is a finite set of \textup{rules}, and $S \in V$ is the \textup{start variable}. A rule in a CFG is of the form $B \rightarrow w$, where $B \in V$ and $w$ is a string of variables and/or terminals.
\end{definition}

To generate a string using $G$, use the following procedure:

\begin{enumerate}
    \item Begin with a string consisting of only $S$.
    \item Select any variable $B$ occurring in the current string, find a rule of the form $B \rightarrow w$, and substitute $w$ for $B$.
    \item Keep performing similar substitutions until the string consists of only terminal symbols.
\end{enumerate}

The above procedure computes what is called a \emph{derivation}; in particular a derivation from the start variable $S$ (because that's the variable we began with in Step (1)). By beginning with a different $B \in V$ in Step (1), we can compute a derivation from $B$. If some $x \in A^*$ can be derived from $B \in V$ (using the above procedure), we write $B \Rightarrow^* x$.

In Step (2) of the above procedure, if the leftmost variable is always chosen, then the derivation is called a \emph{leftmost derivation}. Conversely, if the rightmost variable is always chosen, then the derivation is called a \emph{rightmost derivation}.

The set $L(G, B)$ denotes the set of all strings over $A$ that can be derived from some non-terminal $B \in V$ i.e., $L(G, B) = \{u \in A^* \mid B \Rightarrow^* u\}$. The set $L(G, S)$ is called the \emph{language of $G$}, and is generally abbreviated to $L(G)$.

A language is called \emph{context-free} if some CFG generates it. CFGs are strictly more powerful than finite automata i.e., every language that can be recognised by a finite automaton can be generated by a CFG, but there are some languages that can be generated by a CFG but not recognised by a finite automaton. Given an arbitrary DFA $M = (Q, A, \delta, q_0, F)$ over $A$, there is a CFG $G$ over $A$ such that $L(G) = L(M)$. We now provide a definition of $G = (V, A, R, S)$, adapted from \citet{RN75}.

\begin{align*}
    V &= Q \\
    S &= q_0 \\
    R &= \{B \rightarrow aC \mid \delta(B, a) = C\} \cup \{B \rightarrow a \mid \delta(B, a) \in F\}
\end{align*}

$G$ as defined above is a special kind of CFG, known as a \emph{regular grammar}.

The language $\{0^n1^n \mid n \geq 1\}$ (over alphabet $\{0, 1\}$) is an example of a language that can be generated by a CFG, but not recognised by a finite automaton, thus demonstrating that CFGs are indeed strictly more expressive than finite automata.

\begin{example}
    Below is a simple example of a CFG $G = (V, A, R, S)$, with $V = \{S\}$, $A = \{(, )\}$, start symbol $S$, and the set of rules $R$ defined as follows:

    \begin{align*}
        S &\rightarrow (S) \\
        S &\rightarrow SS \\
        S &\rightarrow \varepsilon
    \end{align*}
    
    Because the above rules all have $S$ on the left-hand side, we may abbreviate them into a single line, using the symbol ``$\mid$'' as a separator:
    
    $$S \rightarrow (S) \mid SS \mid \varepsilon$$
    
    $G$ is able to generate the string $()()$:
    
    $$S \Rightarrow SS \Rightarrow (S)S \Rightarrow (S)(S) \Rightarrow ()(S) \Rightarrow ()()$$
    
    It is also able to generate the string $(())$:
    
    $$S \Rightarrow (S) \Rightarrow ((S)) \Rightarrow (())$$
    
    After some thought, we can see that $G$ generates the set of all strings of balanced parentheses.
\end{example}

\subsection{Language Learning}\label{language-learning}

The goal of language learning is to infer a mathematical model, typically either a finite automaton or grammar, which can accurately classify an arbitrary $s \in A^*$ as being either a member, or not a member, of an unknown \emph{target language} $X$. In general, a language learning algorithm is provided with the following inputs:

\begin{itemize}
    \item an alphabet $A$,
    \item a class $\mathcal{C}$ of languages over $A$, of which $X$ is known to be a member (i.e., $X \in \mathcal{C}$ is a fact that can be assumed by the language learning algorithm),
    \item a \emph{hypothesis space}, $\mathcal{R}$ (typically either a set of automata or a set of grammars), where each $L \in \mathcal{C}$ is recognised or generated by some hypothesis in $\mathcal{R}$, and
    \item some information about $X$ (this information can take various forms, as described in \Cref{learning-models}).
\end{itemize}

The language learning algorithm is then required to output a \emph{hypothesis} i.e., an element of $\mathcal{R}$, which can accurately classify strings in $A^*$ according to their membership in $X$. Typically, existing language learning algorithms output a hypothesis which is either a finite automaton or a CFG.

\subsection{Learning Models}\label{learning-models}

A \emph{learning model} refers to the framework within which a learning algorithm operates. In particular, it specifies exactly what information about the unknown target language $X$ is available to the learner. We now describe the main learning models used in language learning.

\subsubsection{Passive Learning}\label{learning-models-passive}

In the \emph{passive learning} model, a learner is given a finite set of examples classified according to their membership in $X$, and is required to produce an automaton which generalises these examples.

Let us now provide a more formal definition of this learning model, adapted from \citet{RN37}. It is useful to view our target language $X$ as a mapping $X: A^* \rightarrow \{0, 1\}$, where, for arbitrary $w \in A^*$, $X(w) = 1$ if $w \in X$, and $X(w) = 0$ if $w \not\in X$. A \emph{sample} of $X$ on $D \subseteq A^*$, denoted by ${X \mid}_{D}$, is a set of classified examples $\{(w, X(w)) \mid w \in D\}$. The set of all positive examples within this sample is equivalent to ${X \mid}^{-1}_{D}(1)$, and likewise the set of all negative examples within this sample is equivalent to ${X \mid}^{-1}_{D}(0)$. 

We can now formally define passive learning as follows: given a sample ${X \mid}_{D}$, for some finite $D \subseteq A^*$, we want to find a representation $\mathcal{A}$ in our hypothesis space $\mathcal{R}$ that is \emph{consistent} with the sample i.e., for each $w \in D$, $\mathcal{A}(w) = X(w)$. Ideally, this representation $\mathcal{A}$ should be small; in other words, it should generalise the given examples.

\subsubsection{Active Learning}\label{learning-models-active}

A passive learning algorithm has access to a finite set of strings $D \subseteq A^*$ classified according to their membership in $X$, without being able to influence which strings are placed into $D$. By contrast, an \emph{active learning} algorithm is afforded a greater degree of freedom, in that it is allowed to \emph{choose} certain questions to ask about $X$. Specifically, the active learning scenario involves two entities -- a \emph{learner} and a \emph{teacher}. The learner wants to learn an unknown target language $X$ known by the teacher. In order to do so, the learner is permitted to ask questions, or make queries, about $X$, which must be answered by the teacher. \citet{RN1} defines six key types of queries which can be made by the learner: 

\begin{itemize}
    \item \emph{Membership} -- the input is a string $w \in A^*$, and the output is \y if $w \in X$, and \n otherwise.

    \item \emph{Equivalence} -- the input is a representation $\mathcal{A'} \in \mathcal{R}$ of language $X' \in \mathcal{C}$, and the output is \y if $X' = X$ and \n otherwise. If the output is \n, an arbitrary $x \in X \oplus X'$ is returned.

    \item \emph{Subset} -- the input is a representation $\mathcal{A'} \in \mathcal{R}$ of language $X' \in \mathcal{C}$, and the output is \y if $X' \subseteq X$ and \n otherwise. If the output is \n, an arbitrary $x \in X' \setminus X$ is returned.

    \item \emph{Superset} -- the input is a representation $\mathcal{A'} \in \mathcal{R}$ of language $X' \in \mathcal{C}$, and the output is \y if $X' \supseteq X$ and \n otherwise. If the output is \n, an arbitrary $x \in X \setminus X'$ is returned.

    \item \emph{Disjointness} -- the input is a representation $\mathcal{A'} \in \mathcal{R}$ of language $X' \in \mathcal{C}$, and the output is \y if $X' \cap X = \varnothing$ and \n otherwise. If the output is \n, an arbitrary $x \in X' \cap X$ is returned.

    \item \emph{Exhaustiveness} -- the input is a representation $\mathcal{A'} \in \mathcal{R}$ of language $X' \in \mathcal{C}$, and the output is \y if $X' \cup X = A^*$ and \n otherwise. If the output is \n, an arbitrary $x \not\in X \cup X'$ is returned.
\end{itemize}

For all query types other than membership, the element $x$ returned with the output \n is known as a \emph{counterexample}. Counterexamples are assumed to be chosen arbitrarily i.e., the learning algorithm cannot make any assumptions about how the counterexample has been chosen.

A widely-used active learning framework is the \emph{minimally-adequate teacher} framework, introduced by \citet{RN6}. A \emph{minimally-adequate teacher} (MAT) is a teacher that can answer only two kinds of queries -- membership queries and equivalence queries. Many language learning algorithms have been formulated within this MAT framework, including the classical L* algorithm proposed by \citet{RN6} herself.

\noindent\textbf{Probably Approximately Correct Learning.}
\label{learning-models-active-pac}
Though many active learning algorithms make use of a minimally-adequate teacher, it is often not possible to simulate such a teacher in practice. Many practical systems of interest are capable of answering membership queries; however, answering equivalence queries is almost always infeasible. This however does not present a barrier to MAT algorithms being used in practice. Rather, the ability to answer equivalence queries can be replaced with a \emph{random sampling oracle}, allowing learning to take place within the \emph{probably approximately correct} (PAC) learning framework. We now describe this framework in more detail.

The \emph{probably approximately correct} (PAC) framework was originally introduced by \citet{RN42} to study the complexity of machine learning problems, and so his setting is much broader than just language learning. Thus, we use the formal definition presented by \citet{RN1}. Let $D$ be some (unknown but fixed) probability distribution over $A^*$. The learner has access to a \emph{random sampling oracle} which, when called, returns some $x \in A^*$ drawn according to $D$. Intuitively, instead of asking an equivalence query, the learner makes a certain number of calls to the random sampling oracle. If, for every $x \in A^*$ returned by the random sampling oracle, the learner's hypothesis $\mathcal{A}$ classifies $x$ correctly according to its membership in $X$, then $L(\mathcal{A})$ is deemed equivalent to $X$ and so $\mathcal{A}$ is returned. Otherwise, some sampled string is \emph{not} classified correctly by $\mathcal{A}$ -- which is a counterexample.

The PAC framework has two parameters -- the \emph{accuracy} parameter $\varepsilon$, and the \emph{confidence} parameter $\delta$. Once values have been chosen for $\varepsilon$ and $\delta$, the number of calls made to the random sampling oracle in place of each equivalence query can be set accordingly, as described by \citet{RN1}. Specifically, the number of calls made in place of the $i^{\text{th}}$ equivalence query is $q_i$, where

$$q_i = \left\lceil{\frac{1}{\varepsilon} \left(\ln{\frac{1}{\delta}} + i \ln{2}\right)}\right\rceil$$

Replacing each equivalence query with calls to the sampling oracle in the above-described manner gives rise to a theoretical guarantee on the classification accuracy of the inferred hypothesis. Before stating this guarantee, let us first define the notion of the \emph{difference} between arbitrary $L, L' \subseteq A^*$ with respect to $D$:

\begin{equation}\label{eqn:diff}
    d(L, L') = \sum_{x \in L \oplus L'} \Pr(x)
\end{equation}

Intuitively, $d(L, L')$ quantifies the probability that some element $x \in A^*$ drawn according to $D$ is in $L$ but not in $L'$, or in $L'$ but not in $L$. So, $\Pr$ in Equation~\ref{eqn:diff} is a probability over the sample space $A^*$.

If $X' \subseteq A^*$ denotes the language of the inferred hypothesis and $X \subseteq A$ the target language, then the guarantee satisfied by $X'$ can be expressed as follows:

\begin{equation}\label{eqn:pac-guarantee}
    \Pr(d(X', X) \leq \varepsilon) \geq 1 - \delta
\end{equation}

In words: the probability that the hypothesis returned by the learner is a ``good'' (quantified by $\varepsilon$) approximation of $X$ should be ``high'' (quantified by $\delta$). In Equation~\ref{eqn:pac-guarantee}, $\Pr$ is a probability over the samples drawn from $A^*$ and the internal randomness of the learning algorithm. Its sample space then becomes $(A^*)^m \times \{0, 1\}^r$, where $m$ is the number of samples drawn from $A^*$, and $r$ the number of random choices made by the algorithm (with each random choice being modelled by a single bit).

Note that this PAC guarantee is in contrast to the criterion of \emph{exact identification} which is expected of a MAT learning algorithm. Specifically, because a MAT learning algorithm can ask \emph{exact equivalence queries} (i.e., equivalence queries as described in ~\Cref{learning-models-active}), it is expected to output a hypothesis whose language is \emph{exactly} equivalent to the target language.

\subsection{The Prefix Query}\label{pref-query}

Recall the six types of queries described in \Cref{learning-models-active}. A \emph{prefix oracle} is able to answer a seventh kind of query, referred to as a \emph{prefix query}.

\begin{definition}[Prefix Query]\label{definition-pref-query}
    Given a string $x \in A^*$ as input, the output of a prefix query is one of the following:
    
    \begin{itemize}
        \item \member, if $x$ is a member of $X$,
        \item \livepref, if $x$ is a live prefix of $X$ i.e., $x$ is not in $X$, but it can be extended to form a string in $X$, or
        \item \deadpref, if $x$ is a dead prefix of $X$ i.e., $x$ is not in $X$, and furthermore it \emph{cannot} be extended to form a string in $X$.
    \end{itemize}
\end{definition}

\subsubsection{Parsers}\label{pref-query-parsers}

The prefix query arises quite naturally from the field of software engineering. Language learning is often used in the software engineering domain to learn the language accepted by a \emph{parser}. In general, a parser can be thought of as a program which, given a string $x \in A^*$ as input, does either one of the following:

\begin{itemize}
    \item constructs an internal representation of $x$ (for example, a syntax tree, as depicted in Figure~\ref{fig:parser-syntax-tree}), if $x$ conforms to the format expected by the parser, or
    \item returns an error message, if $x$ does \emph{not} conform to the format expected by the parser.
\end{itemize}

\begin{figure*}
\centering
\includegraphics[width=\columnwidth*1/2]{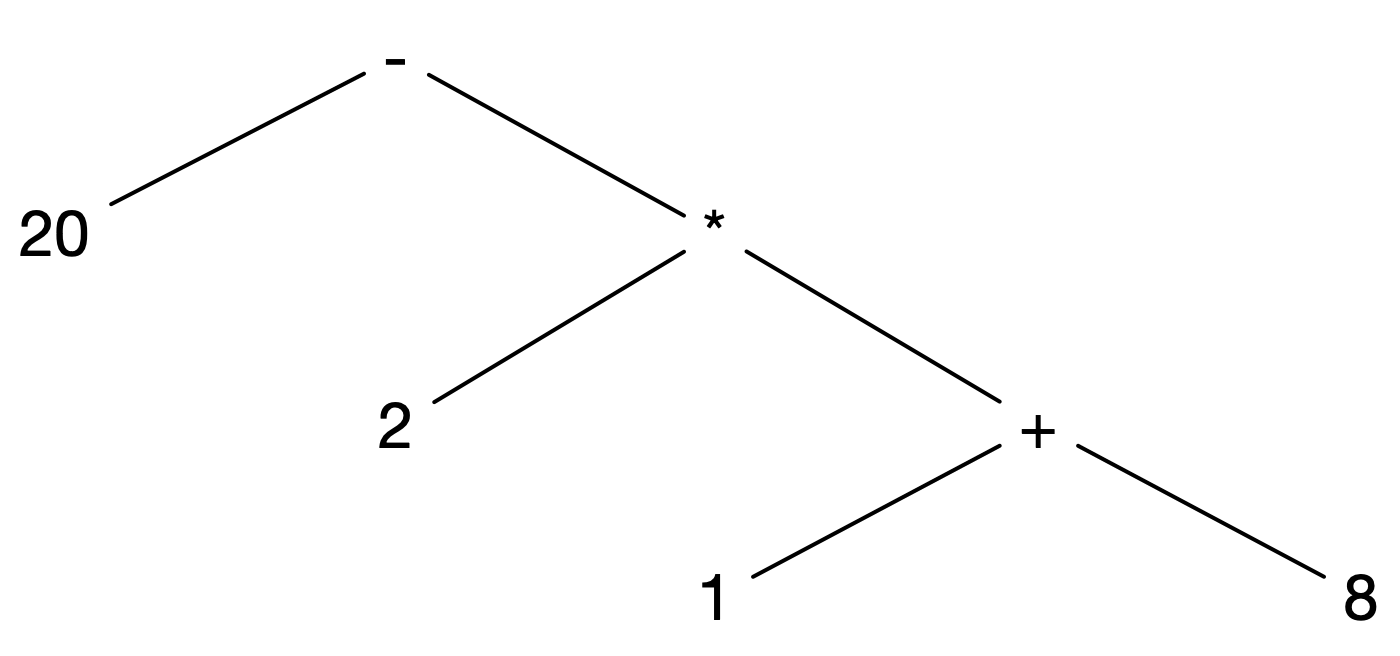}
\caption{Syntax tree for the arithmetic expression $20 - (2 * (1 + 8))$}\label{fig:parser-syntax-tree}
\end{figure*}

\citet{RN89} identify that many parsers prevalent in the field of software engineering are naturally capable of answering prefix queries. For example, two widely-used classes of parsers that can answer prefix queries are \emph{LL parsers} and \emph{LR parsers}. LL parsers are able to parse restricted context-free languages, known as \emph{LL languages}, whilst LR parsers are able to parse a class of context-free languages known as the \emph{deterministic context-free languages}. The LL languages are a strict subset of the deterministic context-free languages; i.e., every LL language is a deterministic context-free language, but some deterministic context-free languages are \emph{not} LL languages. Thus, LR parsers are strictly more powerful than LL parsers.

An \emph{LL grammar} is a context-free grammar (CFG) that can be parsed by an LL parser, whilst an \emph{LR grammar} is a CFG that can be parsed by an LR parser. Just as LR parsers are strictly more powerful than LL parsers, LR grammars are strictly more expressive than LL grammars.

Both LL and LR parsers read a given input string from left-to-right, symbol-by-symbol. An LL (LR) parser is called an \emph{LL(k) (LR(k)) parser} if it uses $k$ symbols of \emph{lookahead} when parsing an input string (typically $k = 1$). Specifically, this means that it ``peeks ahead'' at the $k$ unread symbols following the most recently-read symbol -- without actually reading them. Given a string $x \in A^*$ as input, an LL($k$) parser computes a leftmost derivation of $x$, whilst an LR($k$) parser computes a rightmost derivation of $x$ \emph{in reverse} (i.e., beginning with $x$ and ending with the start variable).

\subsubsection{The Viable-Prefix Property}\label{pref-query-viable-pref}

As previously stated, LL and LR parsers are naturally capable of answering prefix queries, due to possessing a property known as the \emph{viable-prefix property} \citep{RN74}:

\begin{definition}[Viable-Prefix Property]\label{definition-viable-pref-prop}
    Given an input string $x \in A^*$, a parser which has the \textup{viable-prefix property} throws an error as soon as it has read a prefix of $x$ which cannot be extended to form a string in the language. 
\end{definition}

We now describe how a parser which has the viable-prefix property can answer a prefix query.

Consider a parser that has the viable-prefix property and accepts the language $L \subseteq A^*$. Given an input string $x \in A^*$,

\begin{enumerate}
    \item if the parser accepts $x$, then return \member;
    \item if the parser fully consumes $x$ without throwing an error, but has not yet accepted (i.e., is waiting for the next input symbol), then return \livepref; and
    \item if the parser throws an error, then return \deadpref.
\end{enumerate}

We now argue that, using the above approach, prefix queries can be answered accurately:

\begin{enumerate}
    \item If the parser accepts $x$, then, by definition, $x \in L$, and so \member would indeed be the correct response.
    \item Since $x$ is not accepted by the parser, clearly $x \not\in L$. However, because $x$ could be consumed fully without an error being thrown, then by the viable-prefix property, $xy \in L$ for some $y \in A^+$ i.e., $x \in \PP(L)$. Thus, $x$ is a live prefix, and so \livepref would indeed be the correct response.
    \item Let's say that an error is thrown after the parser has read some prefix $x'$ of $x$. This means that $x' \not\in \PP(L)$ -- and so it is certainly the case that $x \not\in \PP(L)$. If an error is thrown, clearly $x$ is not accepted i.e., $x \not\in L$. Thus, $x$ is a dead prefix, and so \deadpref would indeed be the correct response.
\end{enumerate}

\section{Problem Statement} \label{chap:problem-statement}

In \Cref{pref-query-parsers}, we identified that many parsers are naturally capable of answering \emph{prefix queries}; i.e., intuitively, they indicate whether or not an input is valid, and also whether or not it is a \emph{prefix} of some valid input. However, as we shall see in \Cref{chap:related-work} (Related Work), no existing language learning algorithms make use of these prefix queries. Hence, broadly-speaking, the aim of this work is to investigate how prefix queries can be used in the field of language learning. Specifically, we aim to understand how existing language learning algorithms which use membership queries can be modified to exploit the additional information given by prefix queries, in order to either learn more efficiently (in terms of runtime or space usage, or both), or learn a broader class of languages. Then, for any approach we propose, we aim to demonstrate how it can be applied to the problem of learning the language of a blackbox parser.

\section{Related Work} \label{chap:related-work}

In this section, we describe a range of existing language learning algorithms, including both classical and state-of-the-art approaches. Notably, none of the algorithms we describe here make use of the prefix query introduced in Section~\ref{pref-query}; in fact, to our knowledge, this type of query has not been used by \emph{any} existing language learning algorithm. Hence, in ~\Cref{chap:method} (Method), we will describe how one existing language learning algorithm can be improved by modifying it to exploit prefix queries.

\subsection{Learning Deterministic Finite Automata}\label{learning-dfas}

The majority of language learning algorithms which learn a finite automaton recognising a regular target language learn a \emph{deterministic} finite automaton (DFA). Indeed, the novel algorithm we introduce in \Cref{chap:method} is one that learns a DFA representation of the target language. We now describe several algorithms that learn DFAs.

\subsubsection{Active Learning}\label{learning-dfas-active}

Recall the active learning model introduced in Section~\ref{learning-models-active}. In short, an active learning algorithm is one that learns by asking certain kinds of queries about the target language. Many active learning algorithms that infer a DFA representation of the target language make use of a minimally-adequate teacher (MAT), namely a teacher that can answer two kinds of queries -- membership queries and equivalence queries (see Section~\ref{learning-models-active} for further details about these two query types). There are a few algorithms however that use other query types, including a slight variant of the prefix query introduced in Section~\ref{pref-query}. In this section, we describe a range of active learning algorithms for inferring DFAs -- some of which rely on a MAT, whilst others use alternative query types.

\noindent\textbf{Learning from a Minimally-Adequate Teacher}\label{learning-dfas-active-mat}

The first, and arguably most well-known, minimally-adequate teacher (MAT) algorithm was introduced by \citet{RN6} as the L* algorithm. Angluin proves that L* learns the minimal DFA (Definition~\ref{definition-nerode-automaton}) recognising the unknown regular target language $X$ in time polynomial in the size of the minimal DFA recognising $X$ and the length of the longest counterexample returned by the teacher. Because the algorithm described in \Cref{chap:method} (Method) is a modification of L*, we provide a more detailed description of L* in Section~\ref{angluin-lstar-overview}. 

Another algorithm which efficiently learns a DFA representation of the target language is the discrimination tree (DT) algorithm proposed by \citet{RN71}. As observed by \citet{RN38}, an advantage of this approach is that, unlike L*, which performs a membership query for \emph{every} $x \in (S_{\pre} \cup S_{\pre} \cdot A) \cdot S_{\suff}$, the DT approach only performs membership queries which distinguish between states. Let $X \subseteq A^*$ be the unknown regular language which is to be learned, and $M = (Q, A, \delta, q_0, F)$ the minimal DFA recognising $X$. Intuitively, the algorithm maintains two sets of strings -- a set $S \subseteq A^*$ of \emph{access strings} (where, for each $s \in S$, $\delta^*(q_0, s)$ is a unique element of $Q$) and a set $D \subseteq A^*$ of \emph{distinguishing strings} (where, for every $s_1, s_2 \in S$, there is some $d \in D$ such that $M$ accepts exactly one of $s_1d$ and $s_2d$). The learning algorithm stores $S$ and $D$ in a binary classification tree, where each internal node is labelled by some $d \in D$, and each leaf node by some $s \in S$. The root of the tree is labelled by $\varepsilon$, and, for each internal node labelled $d \in D$, its left subtree will contain leaf nodes for each $s \in S$ such that $sd$ is rejected by $M$, and its right subtree will contain leaf nodes for each $s \in S$ such that $sd$ is accepted by $M$. Thus, any $s_1, s_2 \in S$ will be distinguished by the string $d \in D$ labelling the least common ancestor of their corresponding leaf nodes. The algorithm's main goal is to ``continually discover new states of $M$'' \citep{RN71} i.e., states which can be distinguished from every other state discovered so far. In doing so, it gradually enlarges $S$ till eventually $|S| = |Q|$ -- at which point the hypothesis produced by the learner is equivalent to $M$.

The state-of-the-art MAT algorithm is widely considered to be the TTT algorithm proposed by \citet{RN38}. TTT is an improvement upon the DT algorithm proposed by \citet{RN71}, which is proven to achieve lower, and in fact optimal, space complexity. It does so by analysing each counterexample received from the teacher and storing only the parts which offer new information about $X$. Specifically, the authors explain how the learner may receive an excessively-long counterexample, where a large portion of it presents \emph{redundant} information, or information that the learner had already ascertained from previous counterexamples. The TTT algorithm is able to identify and discard such redundancies, storing only the new information provided by each counterexample.

\noindent\textbf{Learning from Other Query Types}\label{learning-dfas-active-other-query-types}

Several active learning approaches use query types other than membership queries and equivalence queries, usually to reduce the number of queries asked (compared to if a MAT was used). For example, the algorithm proposed by \citet{RN10} is a modification of Angluin's L* which queries a teacher that, like a MAT, can answer both membership queries and equivalence queries. However, unlike a MAT, this teacher supplements an \n response to a membership query with some additional information, namely a set of ``similar'' strings which also are negative examples (i.e., not elements of the target language). Specifically, when asked a membership query with the string $w \in A^* \setminus X$, the teacher is required to supplement the response \n with one of the following:

\begin{enumerate}
    \item A set of strings that all start with the prefix $p$, for some $p \in A^*$ satisfying $w = px$, and are also negative examples.
    \item A set of strings that all end with the suffix $s$, for some $s \in A^*$ satisfying $w = xs$, and are also negative examples.
    \item A sequence of substrings $(w_1, w_2, \dots, w_n)$ of $w$ which, when present in a given string, in that same order, imply that the string is a negative example.
\end{enumerate}

A strength of this work is that the authors provide a clear practical application of their algorithm, namely to the field of ontology learning. A weakness however is that, though the authors show that their proposed algorithm performs better than L* \emph{in practice}, they do not analyse the algorithm's \emph{theoretical} complexity. If the proposed algorithm had been shown to outperform L* in terms of theoretical complexity as well, this would have been a stronger positive result.

The LCA algorithm, proposed by \citet{RN20}, is another approach that uses a more informative variant of the membership query described by \citet{RN1}. Specifically, LCA uses two kinds of queries -- equivalence queries and \emph{correction queries}. A \emph{correction query} is in fact identical to the prefix query described in Section~\ref{pref-query}, except that additional information is provided alongside the \livepref response. Specifically, if a correction query is asked for some $x \in A^*$, and $x$ is found to be a live prefix of the target language $X$, then, together with the \livepref response, a string $y \in A^*$ such that $y$ is one of the shortest strings satisfying $xy \in X$ is returned. This $y$ is termed a \emph{correcting string}, and the authors describe how, intuitively, it mimics the feedback or ``corrections'' given to a child who is learning a new language. Because a correction query is a slight variant of our prefix query (Definition~\ref{definition-pref-query}), we will henceforth refer to correction queries as \emph{extended prefix queries}.

LCA is arguably only a slight modification of Angluin's L*, with the only notable differences between them being the following:

\begin{enumerate}
    \item In both L* and LCA, each cell in the observation table corresponds to some $w \in A^*$. In L*, the cell corresponding to $w$ is filled in with $1$ if $w \in X$ and $0$ otherwise. In LCA, the cell corresponding to $w$ is filled in with
    
    \begin{itemize}
        \item $\varepsilon$, if $w \in X$,
        \item $y$, if $w \not\in X, w \in \PP(X)$ and $y \in A^*$ is the one of the shortest strings satisfying $wy \in X$, or
        \item $\varphi$, if $w \not\in X, w \not\in \PP(X)$ (where $\varphi \not\in A$).
    \end{itemize}

    \item LCA is able to use the results of previous extended prefix queries to reduce the number of queries that need to be made. For example, if the cell corresponding to $w \in A^*$ is filled in with $\varphi$, then it immediately deduces that the cell corresponding to $wx \in A^*$ (for some $x \in A$) should also be filled in with $\varphi$.
\end{enumerate}

Though \citet{RN20} do not \emph{theoretically} quantify the improvement given by LCA over L* (they only demonstrate this improvement \emph{empirically}), \citet{RN68} prove theoretical upper bounds on the number of equivalence queries and extended prefix queries performed by LCA. Intuitively, they define a property of languages called \emph{injectivity}, and argue that, if the target language has small injectivity, the number of extended prefix queries made by LCA is substantially smaller than the number of membership queries made by L*. Despite this positive result, a key shortcoming of the LCA algorithm is that it uses an extended prefix query which does not appear to be answerable in a practical context -- even by a parser that can answer (non-extended) prefix queries (as discussed in Section~\ref{pref-query-parsers}). This therefore makes the practical applicability of this algorithm unclear.

\subsubsection{Passive Learning}\label{learning-dfas-passive}

Recall the passive learning model described in Section~\ref{learning-models-passive}. In short, a passive learning algorithm is provided with a set of example strings labelled according to their membership in the target language. Passive learning algorithms that learn a DFA representation of a regular language typically use a technique known as \emph{state-merging}. In general, a state-merging algorithm learning a DFA that recognises the regular language $X \subseteq A^*$ begins by constructing a \emph{prefix tree acceptor} (PTA) $(Q, A, \delta, q_0, F)$ for the given positive examples ${X \mid}_{D}^{-1}(1)$, defined as follows:

\begin{align*}
    Q &= \prefix({X \mid}_{D}^{-1}(1)) \cup \{\varnothing\} \\
    \delta(x, a) &= 
        \begin{cases}
            xa, &\text{if both $x$ and $xa$ are elements of $\prefix({X \mid}_{D}^{-1}(1))$}\\
            \varnothing, &\text{otherwise}\\
        \end{cases} \\
    q_0 &= \varepsilon \\
    F &= {X \mid}_{D}^{-1}(1)
\end{align*}

Intuitively, a PTA accepts precisely the positive examples from the sample. After constructing this PTA, the state-merging algorithm then tries to generalise it to accept a larger subset of $X$ by merging as many states as possible, checking before each merge that the resulting automaton would not accept any of the negative examples from the sample.

One of the most well-known state-merging algorithms is the regular positive and negative inference (RPNI) algorithm proposed by \citet{RN40}, which controls the order in which state merges are performed to ensure the automaton remains consistent with the given examples and also does not have any non-determinism. The algorithm first constructs a PTA from the given examples, then assigns colours to each of the PTA's states. Specifically, it colours the PTA's root RED, the root's successors BLUE, and every other state WHITE. The algorithm then attempts to merge a BLUE state with a RED state. If performing such a merge would produce an automaton consistent with the provided examples, then this is done, after which the new combined state is coloured RED and all its (currently-WHITE) successors are coloured BLUE. The merge is followed by a ``folding'' operation, to eliminate any nondeterminism in the automaton produced. On the other hand, if a BLUE state cannot be merged with any RED state, it is coloured RED, and all its successors BLUE. The algorithm keeps performing these state merges until, eventually, there are no BLUE-coloured nodes remaining.

Though RPNI does provide \emph{some} direction on how state merges should be performed, one characteristic of the algorithm, which could be viewed as a limitation, is that it does not \emph{fully} specify the order in which the merges should be performed. This is because a few steps in the algorithm require choices to be made but do not specify how exactly this should be done. These choices could significantly affect the performance of the algorithm, specifically how well the final automaton captures the target language $X$. To address this, several variants of the RPNI algorithm have emerged, with one such algorithm being the evidence driven state merge (EDSM) algorithm proposed by \citet{RN41}. This algorithm defines a heuristic to score possible merges, and then performs the merge with the highest score. The EDSM algorithm is considered to be the state-of-art passive learning algorithm for inferring DFAs.

\subsection{Learning Nondeterministic Finite Automata}\label{learning-nfas}

All the learning algorithms we have considered so far learn DFA representations of regular languages. In many cases, however, it would be preferable to learn an NFA representation, particularly because the NFA recognising a given regular language may be exponentially smaller than the corresponding (minimal) DFA. Many algorithms that learn DFAs rely on the existence of a \emph{unique, minimal} DFA for every regular language. This property however does not always hold for NFAs; that is, not every regular language has a unique, minimal NFA accepting it. Because of this, most algorithms which learn NFA representations of regular languages concentrate on subclasses of NFAs which do satisfy this property. In particular, the class of \emph{residual finite state automata} (RFSAs) is a subclass of the NFAs which satisfies the following property: every regular language has a unique minimal \emph{canonical} RFSA accepting it.

Let $L \subseteq A^*$ be an arbitrary regular language. In the minimal DFA recognising $L$, each state accepts a unique residual language of $L$, and furthermore every residual language of $L$ is accepted by some state. Likewise, in an RFSA recognising $L$, each state accepts a residual language of $L$. However, not \emph{every} residual language of $L$ is accepted by a \emph{single} state; some residual languages are instead the \emph{union} of languages of other states. This therefore implies that an RFSA recognising $L$ may have fewer states than the minimal DFA recognising $L$. The \emph{canonical} RFSA recognising $L$ is the unique \emph{reduced} and \emph{saturated} RFSA recognising $L$. A \emph{reduced} RFSA has no state corresponding to a language which is the union of languages of some other states. A \emph{saturated} RFSA is one for which it is not possible to make any additional state initial, or add any edge, without modifying its language.

Two algorithms that learn RFSA representations of regular languages are the NL* \citep{RN12} and DeLeTe2 \citep{RN43} algorithms, which we will now briefly describe.

\subsubsection{Active Learning}\label{learning-nfas-active}

As previously described, RFSAs share several properties with DFAs, notably the existence of a unique minimal DFA as well as a unique minimal canonical RFSA for every regular language. \citet{RN12} recognise this similarity, and so make a few changes to the L* algorithm to allow it to learn the minimal canonical RFSA, rather than the minimal DFA, recognising the target language. In particular, the notions of closedness and consistency described by \citet{RN6} are modified slightly, giving rise to the notions of \emph{RFSA-closedness} and \emph{RFSA-consistency}, respectively. Aside from this, their algorithm, which they refer to as NL*, is essentially identical to L*, in terms of the general principles of maintaining an (RFSA-)closed and (RFSA-)consistent observation table by extending it as needed (using membership queries), and subsequently using this table to produce an automaton which is then presented as a conjecture. One shortcoming of this algorithm is that its query complexity is found to be poorer than that of the L* algorithm. However, the authors perform experiments to show that, \emph{in practice}, NL* tends to outperform L* in terms of the number of queries required.

\subsubsection{Passive Learning}\label{learning-nfas-passive}

Like NL*, the DeLeTe2 algorithm proposed by \citet{RN43} also learns a RFSA representation of a regular language. However, instead of using a MAT-based approach, DeLeTe2 constructs an RFSA consistent with a sample containing both positive and negative examples within polynomial time. The authors compare their algorithm with RPNI, which learns a DFA from a given sample (as described in Section~\ref{learning-dfas-passive}). In short, RPNI first constructs a prefix tree acceptor for the given positive examples and then tries to merge states that accept equivalent residual languages. Likewise, DeLeTe also considers prefixes of positive examples as candidate states; however, instead of considering \emph{equivalences} between state languages, it considers \emph{inclusions} (i.e., subset relations) between them. One limitation of this algorithm compared to NL* is that, whilst NL* learns the unique minimal canonical RFSA for the target language $X \subseteq A^*$, DeLeTe2 learns an RFSA whose size is between the sizes of the canonical RFSA and the minimal DFA for $X$. In other words, the representation learned by DeLeTe2 might not be as concise as that which is learned by NL*, however it would still be more concise than the representation learned by L*.

\subsection{Learning Context-Free Grammars}\label{learning-grammars}

All the algorithms we have considered so far only learned finite automata modelling an unknown target language. However, in some cases, it would be preferable to learn a more expressive context-free grammar (CFG) representation instead, as this representation may be able to better-capture the properties of the target language (in particular, recursion) compared to a finite automaton. We now describe some algorithms which learn CFGs.

\subsubsection{Active Learning}\label{learning-grammars-active}

\citet{RN48} presents an algorithm which learns a CFG from a teacher that can answer two kinds of queries -- \emph{structural membership queries} and \emph{structural equivalence queries}. Before defining these query types, we first define the notion of \emph{structural data}. Given a CFG $G$ over alphabet $A$, \emph{structural data} of $G$ refers to unlabelled derivation trees of $G$ i.e., derivation trees which only show terminal symbols, without showing any non-terminal symbols (non-terminals are all simply replaced with a special label $\sigma$). An unlabelled derivation tree of $G$ is referred to as a \emph{structural description of $G$}. Let $K(G)$ denote the set of all structural descriptions of $G$. A \emph{structural membership query} takes as input a structural description $s$ and asks if it is indeed a structural description of the target grammar $G$ i.e., if $s \in K(G)$. A \emph{structural equivalence query} takes as input a grammar $G'$ and asks if it is \emph{structurally equivalent} to $G$ i.e., if $K(G) = K(G')$. If not, a counterexample is provided -- a structural description of $G$ but not of $G'$; or of $G$ but not of $G'$.

Sakakibara reduces the problem of learning a CFG from structural membership and structural equivalence queries to the problem of learning a \emph{tree automaton}, namely an automaton which takes in a tree, rather than a string, as input. Sakakibara then presents a modification of Angluin's L* algorithm which learns a tree automaton, rather than a DFA. In particular, instead of using membership and equivalence queries, the algorithm uses \emph{structural} membership and equivalence queries. Furthermore, rather than constructing a DFA from a closed and consistent observation table, the algorithm constructs a tree automaton, and then converts it to a CFG which can then be used to make a structural equivalence query. Merits of this algorithm include its polynomial running time, as well as its ability to learn the full class of context-free languages. One limitation however is that it is unclear how structural membership and structural equivalence queries could be answered in a practical setting.

\citet{RN49} takes a different approach than \citet{RN48}; rather than trying to learn the \emph{full} class of context-free languages using more informative variants of membership and equivalence queries, he instead presents an algorithm which learns only a \emph{subclass} of the context-free languages, albeit using the traditional membership and equivalence queries introduced by \citet{RN6}. Specifically, the algorithm learns the class of all context-free languages definable by a \emph{congruential CFG}.

Let us define the notion of a congruential CFG. Before doing so, we introduce some notation. Let $L \subseteq A^*$ be arbitrary. A \emph{context} refers to a pair of strings $(l, r) \in A^* \times A^*$. The \emph{distribution} of a string $w \in A^*$ with respect to $L$, denoted by $C_L(w)$, is defined as follows:

$$C_L(w) = \{(l, r) \in A^* \times A^* \mid lwr \in L\}$$

The strings $u, v \in A^*$ are \emph{congruent with respect to $L$}, denoted by $u \equiv_L^\cl v$, if and only if $C_L(u) = C_L(v)$. This is an equivalence relation, with $[u]_L = \{v \mid u \equiv_L^\cl v\}$ denoting the equivalence class of $u$ under $\equiv_L^\cl$. Let $G = (V, A, R, S)$ be an arbitrary CFG over $A$. We say that $G$ is \emph{congruential} if, for every non-terminal $B \in V$, the set $L(G, B)$ is a subset of some congruence class of $L(G)$.

\citet{RN49} notes that, for arbitrary CFG $G$ and DFA $M$ over $A$ (where $M$ is minimal), the correspondence between congruence classes under $\equiv_{L(G)}^{\mathcal{C}}$ and non-terminals of $G$ is analogous to the correspondence between congruence classes under $\equiv_{L(M)}^{\mathcal{N}}$ (Definition~\ref{definition-nerode-congruence}) and states of $M$. Likewise, there is a resemblance between Clark's algorithm and Angluin's L* (which learns the minimal DFA recognising a regular language) -- it maintains a consistent observation table (though it consists of substrings and contexts in place of prefixes and suffixes), fills it in using membership queries, and uses it to construct a grammar which is then presented to the equivalence oracle. If the current grammar is incorrect, and a counterexample is received, then the grammar is modified accordingly. One advantage of this algorithm is its polynomial time complexity. However, though the author does give some observations about the class of context-free languages learned by the algorithm, for example that it is a strict superset of the regular languages (as observed by Clark, the Dyck language is generated by a congruential CFG, but is non-regular), the practical applicability of this language class is not explored. This therefore makes it unclear how the algorithm could be used in practice.

\citet{RN63} presents an algorithm which can learn a \emph{simple deterministic CFG} in polynomial time using extended prefix queries (defined in Section~\ref{learning-dfas-active-other-query-types}) and \emph{derivative queries}. We now define the notion of a simple deterministic CFG. A CFG $G = (V, A, R, S)$ is in \emph{Greibach normal form} (GNF) if every production rule is of the form $B \rightarrow aE$, for $B \in V, a \in A$ and $E$ a sequence of zero or more non-terminals. If $G$ is in GNF, then it is a \emph{simple deterministic CFG} (or \emph{simple grammar}) if, for every $B \in V, a \in A$ and $\alpha, \beta$ (possibly-empty) sequences of non-terminals, the existence of the rules $B \rightarrow a\alpha$ and $B \rightarrow a\beta$ implies that $\alpha = \beta$. Intuitively, a simple grammar is one in which, if you have some non-terminal $B$ and want to generate from it a string beginning with a particular terminal symbol $a$, there is \emph{at most one} possible production rule you could apply. A \emph{simple deterministic context-free language} (or \emph{simple language}) is one that can be generated by a simple grammar. Interestingly, there is a connection between LL($k$) grammars (defined in Section~\ref{pref-query-parsers}) and simple grammars: the class of LL(1) grammars in GNF is equivalent to the class of simple grammars \citep{RN73}. It is important also to note that the class of simple languages is a proper superset of the class of regular languages; for example, as observed by Yokomori, the language $\{ab^nab^na \mid n \geq 0\}$ is a simple language, but is clearly non-regular.

For arbitrary $L \subseteq A^*$ and $x \in A^*$, the \emph{left derivative of $L$ with respect to $x$} is defined in Definition~\ref{definition-residual}. The \emph{right derivative of $L$ with respect to $x$} is the set $L / x = \{w \in A^* \mid wx \in L\}$. Combining these two definitions, we have the following (for arbitrary $u, v \in A^*$):

$$u \backslash L / v = \{x \in A^* \mid uxv \in L\}$$

Also, note that

$$x \in u \backslash L / v \Leftrightarrow (u, v) \in C_L(x)$$

where $C_L(x)$ is the distribution of $x$ as defined by \citet{RN49}.

A \emph{derivative oracle} answers a type of query known as a \emph{derivative query}, which takes as input two pairs of strings $(u_1, v_1), (u_2, v_2) \in A^* \times A ^*$, and returns the output
\begin{itemize}
    \item \y, if $u_1 \backslash L / v_1 = u_2 \backslash L / v_2$, or
    \item \n otherwise.
\end{itemize}

Let $X \subseteq A^*$ be the target simple language. The algorithm proposed by \citet{RN63} uses extended prefix and derivative queries to construct the \emph{characteristic cover graph} for some simple grammar $G$ generating $X$. Intuitively, in this characteristic cover graph, each node would correspond to a left derivative $a \backslash X$, for $a \in A$. An edge between two nodes $a \backslash X$ and $ab \backslash X$ ($a, b \in A$) would be labelled by $b$ -- essentially the terminal symbol that must be "cut off" from the beginning of a string $x \in a \backslash X$ to convert it to a string in $ab \backslash X$. In order to construct this graph, the algorithm constructs and maintains a tabular data structure similar to the observation table introduced by \citet{RN6}. This characteristic cover graph can then be converted into a simple grammar generating the target language $X$.

A key merit of this work is that it was the first to introduce and make use of extended prefix and derivative queries. However, the author does not adequately explain how either the extended prefix or derivative oracles could be simulated in practice, thus limiting the practical applicability of the algorithm.

\subsubsection{Learning from Examples and Membership Queries}\label{learning-grammars-examples-mqs}

\citet{RN51} and \citet{RN52} both use a set of positive examples and membership queries to learn a CFG representing a \emph{program input language}, namely the set of all valid inputs to a software program.

\citet{RN51} propose an algorithm known as GLADE, which uses an approach similar to traditional state-merging algorithms (see Section~\ref{learning-dfas-passive}). Specifically, GLADE begins with a representation which accepts \emph{exactly} the positive examples, and then progressively generalises it. The first phase of the algorithm generalises a given positive example into a \emph{regular expression}, which intuitively is an algebraic way of representing a regular language. It does so by applying two regular expression operators, namely \emph{Kleene star} and \emph{union}, to different substrings of the example string. The second phase then converts this regular expression into a CFG, which is able represent certain recursive properties of the target language that cannot be represented by the regular expression.

Conversely, \citet{RN52} present an algorithm called ARVADA which learns a parse tree for each positive example, from which the algorithm can then extract a CFG. Each parse tree is initially a ``flat'' tree of height 1, with each leaf node corresponding to a symbol in the example string. ARVADA then begins to add layers and structure to this ``flat'' tree by repeatedly replacing a sequence of sibling nodes with a new node, which then becomes the parent of those sibling nodes. Each time this operation occurs, the grammar produced is gradually generalised.

Both \citet{RN51} and \citet{RN52} evaluate their respective algorithms in terms of a metric known as \emph{F1-score} (Definition~\ref{definition-f1}). In terms of F1-score, GLADE is shown to outperform L* and RPNI on several benchmarks, whereas ARVADA is shown to outperform GLADE itself. Both GLADE and ARVADA are considered to be state-of-the-art algorithms for modelling program input languages.

\subsection{Summary of Algorithms}\label{summary}

In Table~\ref{tab:summary-table}, we summarise all the algorithms described in \Cref{chap:related-work}.

\renewcommand{\arraystretch}{1.25}

\begin{table*}[ht]
\centering
\begin{tabular}{ | p{0.3\textwidth}l|  p{0.1\textwidth} |  p{0.1\textwidth} | p{0.2\textwidth} | p{0.2\textwidth} | } 
    \hline
    \textbf{Citation} & \textbf{Name} & \textbf{Section} & \textbf{Learning Model} & \textbf{Representation} \\
    \hline
    \citet{RN6} & L* & \ref{learning-dfas-active-mat} & Active & DFA \\
    \hline
    \citet{RN71} & DT & \ref{learning-dfas-active-mat} & Active & DFA \\
    \hline
    \citet{RN38} & TTT & \ref{learning-dfas-active-mat} & Active & DFA \\
    \hline
    \citet{RN10} & -- & \ref{learning-dfas-active-other-query-types} & Active & DFA \\
    \hline
    \citet{RN20} & LCA & \ref{learning-dfas-active-other-query-types} & Active & DFA \\
    \hline
    \citet{RN40} & RPNI & \ref{learning-dfas-passive} & Passive & DFA \\
    \hline
    \citet{RN41} & EDSM & \ref{learning-dfas-passive} & Passive & DFA \\
    \hline
    \citet{RN12} & NL* & \ref{learning-nfas-active} & Active & NFA \\
    \hline
    \citet{RN43} & DeLeTe2 & \ref{learning-nfas-passive} & Passive & NFA \\
    \hline
    \citet{RN48} & -- & \ref{learning-grammars-active} & Active & CFG \\
    \hline
    \citet{RN49} & -- & \ref{learning-grammars-active} & Active & Congruential CFG \\
    \hline
    \citet{RN63} & -- & \ref{learning-grammars-active} & Active & Simple deterministic CFG \\
    \hline
    \citet{RN51} & GLADE & \ref{learning-grammars-examples-mqs} & Passive \& active & CFG \\
    \hline
    \citet{RN52} & ARVADA & \ref{learning-grammars-examples-mqs} & Passive \& active & CFG \\
    \hline
\end{tabular}
\caption{\label{tab:summary-table}Summary of language learning algorithms}
\end{table*}

\section{Method} \label{chap:method}

The PL* algorithm is a novel modification to the classical language learning algorithm L* \citep{RN6}. Whilst L* queries a minimally-adequate teacher (MAT), PL* queries a teacher that answers prefix and equivalence queries. In this section, we begin by providing an overview of the classical L* algorithm proposed by \citet{RN6} (Section~\ref{angluin-lstar-overview}). We then describe the teacher used by PL* (Section~\ref{teacher}), before providing a summary of the PL* learning algorithm itself (Section~\ref{learner}). Subsequently, we describe in more detail the observation table maintained by PL* -- the key difference between itself and Angluin's L* (Section~\ref{obs-table}). We then prove the correctness and termination of PL* (Section~\ref{correctness-termination}), and show that it always asks less queries than L* when both learners are tasked with learning the same non-dense language (Section~\ref{comparison-with-lstar}). Finally, we discuss how PL* can be adapted to allow it to be used in a more practical setting, namely one in which an equivalence oracle is not available (Section~\ref{random-sampling-oracle}).

\subsection{Overview of Angluin's L*}\label{angluin-lstar-overview}

L* is a classical language learning algorithm proposed by \citet{RN6}, which learns the minimal deterministic finite automaton (DFA) recognising an unknown regular target language $X \subseteq A^*$, for $A$ a finite alphabet given as input to L*. During the learning process, L* is given access to a \emph{minimally-adequate teacher} (MAT), a teacher which can answer only \emph{membership queries} and \emph{equivalence queries}. In short, a membership query indicates whether or not a given string is in the target language, whilst an equivalence query indicates whether or not a given hypothesis recognises exactly the target language. Both query types are defined more formally in Section~\ref{learning-models-active}.

During its execution, L* maintains an observation table $(S_{\pre}, S_{\suff}, T)$ which stores information about a finite set of strings over $A$ -- namely whether or not each string in the set is a member of $X$. Below we describe each of the components of an L* observation table:

\begin{itemize}
    \item $S_{\pre}$ is a non-empty, finite, prefix-closed set of strings. (A set is \emph{prefix-closed} if and only if, for every member of the set, every one of its prefixes is also a member of the set.)
    \item $S_{\suff}$ is a non-empty, finite, suffix-closed set of strings. (A set is \emph{suffix-closed} if and only if, for every member of the set, every one of its suffixes is also a member of the set.)
    \item $T$ is a function mapping $((S_{\pre} \cup S_{\pre} \cdot A) \cdot S_{\suff})$ to $\{0, 1\}$. The value of $T(u)$ is set to 1 if $u \in X$, and 0 otherwise.
\end{itemize}

Intuitively, rows labelled by elements of $S_{\pre}$ correspond to potential states in the hypothesis automaton which L* will construct. The elements of $S_{\suff}$ are used to distinguish the states of the hypothesis automaton. Rows labelled by elements of $(S_{\pre} \cdot A)$ are used when constructing the hypothesis automaton's transition function.

The L* observation table may be visualised as a two-dimensional array, wherein each row is labelled by a string in $(S_{\pre} \cup S_{\pre} \cdot A)$, and each column by a string in $S_{\suff}$. The cell in row $s$ and column $e$ has value $T(se)$. For arbitrary $s \in (S_{\pre} \cup S_{\pre} \cdot A)$, let $\row_s$ be a function mapping $S_{\suff}$ to $\{0,1\}$, where for each $e \in S_{\suff}$, $\row_s(e) = T(se)$. Figure~\ref{fig:obs-table-angluin} provides a visualisation of this.

\begin{figure*}
\centering
    \includegraphics[width=\columnwidth*7/8]{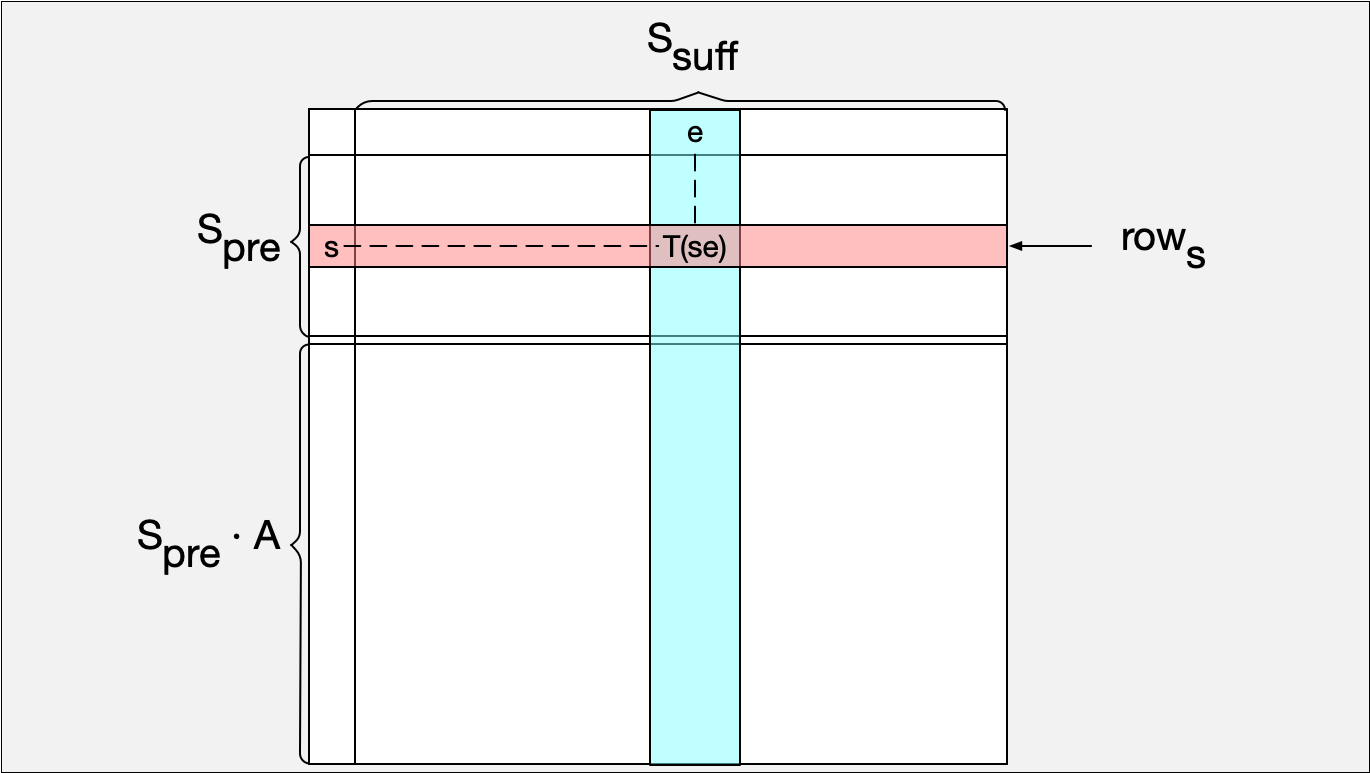}
\caption{Visualisation of an L* observation table}\label{fig:obs-table-angluin}
\end{figure*}

We say that two observation tables $(S_{\pre}, S_{\suff}, T)$ and $(S_{\pre}', S_{\suff}', T')$ are \emph{equivalent} if $S_{\pre} = S_{\pre}'$, $S_{\suff} = S_{\suff}'$ and $T = T'$.

Initially, $S_{\pre}$ and $S_{\suff}$ are both set to $\{\varepsilon\}$. To determine $T$, the learner asks membership queries for $\varepsilon$ and then every $a \in A$. Subsequently, the main loop of L* begins. The learner first checks whether the current observation table satisfies two properties -- \emph{closedness} (meaning that, for every $s \in S_{\pre}$ and $a \in A$, $\row_{sa} = \row_{s'}$ for some $s' \in S_{\pre}$) and \emph{consistency} (meaning that, for any $s_1, s_2 \in S_{\pre}$ such that $\row_{s_1} = \row_{s_2}$, $\row_{s_1a} = \row_{s_2a}$ for all $a \in A$). If the table is found to be not closed, L* identifies some $s \in S_{\pre}, a \in A$ such that $\row_{sa} \neq \row_{s'}$ for every $s' \in S_{\pre}$, and then adds $sa$ to $S_{\pre}$, extending $T$ accordingly. If the table is found to be not consistent, L* identifies some $s_1, s_2 \in S_{\pre}, a \in A, e \in S_{\suff}$ such that $\row_{s_1} = \row_{s_2}$ and $\row_{s_1a}(e) \neq \row_{s_2a}(e)$, and then adds $ae$ to $S_{\suff}$, extending $T$ accordingly. Once the table has been made both closed and consistent, the learner uses it to construct a hypothesis DFA $\mathcal{A}$. It then makes an equivalence query with $\mathcal{A}$. If the teacher responds with the output \y, L* terminates immediately returning $\mathcal{A}$. Otherwise, the learner processes the counterexample provided by the teacher by adding it and all its prefixes to $S_{\pre}$, extending $T$ accordingly. The main loop then recommences. \citet{RN6} shows that, in time polynomial in the size of the minimal DFA recognising $X$ and the length of the longest counterexample returned by the teacher, the learning algorithm will terminate returning the minimal DFA recognising $X$.

\subsection{The PL* Algorithm} \label{sect:pl*}

PL* is a modification of L* which uses prefix queries in place of membership queries. Intuitively, PL* exploits the additional information given by prefix queries over membership queries in two main ways: (1) by automatically filling in certain cells in its observation table without needing to make additional queries, and (2) by explicitly-storing fewer rows in its observation table. Hence, PL* in general uses less queries than L*, as well as less storage space.

Let $A$ be some finite alphabet given as input to PL*. Let the unknown regular target language which our learner is trying to identify be $X \subseteq A^*$.

\subsubsection{The Teacher} \label{teacher}

Instead of a MAT \citep{RN6} which answers membership queries and equivalence queries, our learning algorithm PL* relies on a teacher which answers \emph{prefix} queries and equivalence queries. In other words, our algorithm has access to two oracles: a \emph{prefix oracle} and an \emph{equivalence oracle}. A prefix oracle is able to answer \emph{prefix queries} (Definition~\ref{definition-pref-query}), and likewise an equivalence oracle is able to answer \emph{equivalence queries} (defined in Section~\ref{learning-models-active}).

\subsubsection{The Learner PL*}\label{learner}

The Learner PL* behaves in a very similar manner to L* \citep{RN6}, with the two main differences being that it maintains a different kind of observation table (which we will describe in Section~\ref{obs-table}) and uses prefix queries instead of membership queries.

The pseudocode for PL* is presented in Algorithm \ref{pl*} (in a similar style to the L* pseudocode presented by \citet{RN37}).

\begin{algorithm}
\caption{The PL* algorithm}\label{pl*}
\begin{flushleft}
    \hspace*{\algorithmicindent} \textbf{Input:} A finite alphabet $A$, prefix oracle, and equivalence oracle\\
    \hspace*{\algorithmicindent} \textbf{Output:} A deterministic finite automaton $\A$
\end{flushleft}
\begin{algorithmic}[1]
\STATE initialise $S_{\stored}$ and $S_{\deadpref}$ to $\{\}$
\STATE initialise $S_{\pre}$ and $S_{\suff}$ to $\{\varepsilon\}$, and update $S_{\stored}$, $S_{\deadpref}$ and $T$ accordingly
\REPEAT
    \WHILE{$\obstable$ is not closed or not consistent}
        \IF{$\obstable$ is not closed}
            \STATE find $s \in S_{\pre}$ and $a \in A$ such that $sa \in S_{\stored}$ and $\row_{sa}$ is different from $\row_{s'}$ for all ${s'} \in S_{\pre} \cap S_{\stored}$
            \STATE add $s \cdot a$ to $S_{\pre}$
            \STATE update $S_{\stored}$, $S_{\deadpref}$ and $T$
        \ENDIF
        \IF{$\obstable$ is not consistent}
            \STATE find $s_1$ and $s_2$ in $S_{\pre} \cap S_{\stored}$, $a \in A$ and $e \in S_{\suff}$ such that $\row_{s_1} = \row_{s_2}$ and $r_{s_1a}(e) \neq r_{s_2a}(e)$ \algorithmiccomment{The function $r$ is defined in Definition~\ref{definition-consistency}}
            \STATE add $a \cdot e$ to $S_{\suff}$
            \STATE update $S_{\stored}$, $S_{\deadpref}$ and $T$
        \ENDIF
    \ENDWHILE
    \STATE $\A \gets$ automaton constructed from $\obstable$
    \STATE make the conjecture $\A$
    \IF{the teacher replies with a counterexample $w$}
        \STATE add $w$ and all its prefixes to $S_{\pre}$, in ascending lexicographic order
        \STATE update $S_{\stored}$, $S_{\deadpref}$ and $T$
    \ENDIF
\UNTIL{the teacher replies \y to the conjecture $\A$}
\STATE \RETURN $\A$
\end{algorithmic}
\end{algorithm}

\subsubsection{The Observation Table}\label{obs-table}

One of the key differences between PL* (Algorithm~\ref{pl*}) and L* (Section~\ref{angluin-lstar-overview}) lies in the structure of their observation tables. Intuitively, PL* observation tables are like L* tables except they don't explicitly store rows for all the dead prefixes in $(S_{\pre} \cup S_{\pre} \cdot A)$. Instead, they introduce $S_{\stored}$ and $S_{\deadpref}$.

More formally, a PL* observation table is a 5-tuple \\
$(S_{\pre}, S_{\suff}, S_{\stored}, S_{\deadpref}, T)$. Below we describe each of these five components: 

\begin{itemize}
    \item $S_{\pre}$ is a non-empty, finite, prefix-closed set of strings. (A set is \emph{prefix-closed} if and only if, for every member of the set, every one of its prefixes is also a member of the set.)
    \item $S_{\suff}$ is a non-empty, finite, suffix-closed set of strings. (A set is \emph{suffix-closed} if and only if, for every member of the set, every one of its suffixes is also a member of the set.)
    \item $S_{\stored}$ is a non-empty subset of $(S_{\pre} \cup S_{\pre} \cdot A)$ containing all the strings that label explicitly-stored rows in the observation table. Intuitively, when a string $s$ is added to $S_{\pre}$, PL* considers adding each $x \in \{s\} \cup \{sa \mid a \in A\}$ to $S_{\stored}$. However, for certain $x \in \{s\} \cup \{sa \mid a \in A\}$, PL* may deem it unnecessary to add $x$ to $S_{\stored}$, based on the criteria given in Section~\ref{obs-table-update-s-pre-standard}.
    \item $S_{\deadpref}$ is a possibly-empty, finite set containing all the \emph{minimal} dead prefixes of $X$ that have been identified so far (through prefix queries).
    \begin{itemize}
        \item \emph{Minimal} means that no string in $S_{\deadpref}$ is a proper prefix of another string in \\ $S_{\deadpref}$. 
        \item Recall from Definition~\ref{definition-string-prop} that a \emph{dead prefix} is a string $x \in A^*$ which is not in $X$, and furthermore cannot be extended to a string in $X$.
    \end{itemize}
    \item $T$ is a function mapping $(S_{\stored} \cdot S_{\suff})$ to $\{0, 1\}$. The value of $T(u)$ is set to 1 if $u \in X$, and 0 otherwise.
\end{itemize}

To make this more concrete, we now give an example of a PL* observation table and the equivalent L* observation table.

\begin{example}\label{example:obs-table}
    In Figure~\ref{fig:obs-table-ex}, we see a PL* observation table on the left, and the equivalent L* observation table on the right. Though both tables have $S_{\pre} = \{\varepsilon, 0, 01, 001\}$ and $S_{\suff} = \{\varepsilon, 1\}$, the PL* table has fewer rows than the L* table. This is because $S_{\deadpref}$ in the PL* observation table contains $00$, implying that, for any extension $s$ of $00$, $\row_s$ would contain only 0s (since $s$, like $00$, would also be a dead prefix of $X$). Thus, the PL* observation table need not explicitly store a row for any such $s$. Conversely, the L* observation table explicitly stores a row for \emph{every} string in $(S_{\pre} \cup S_{\pre} \cdot A)$, including $001$, $0010$ and $0011$ (highlighted in red) which are all extensions of $00$.
    
\begin{figure*}[h]
\centering
    \includegraphics[width=0.75\textwidth]{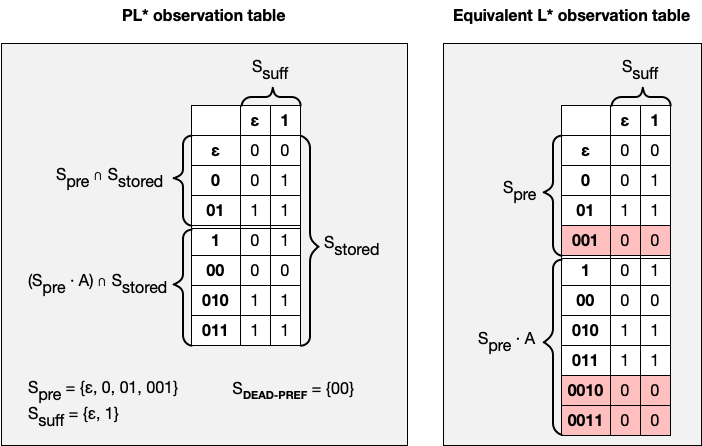}
    \caption{Example of a PL* observation table and equivalent L* observation table}\label{fig:obs-table-ex}
    \end{figure*}
\end{example}

Intuitively, in a PL* observation table, rows labelled by elements of $(S_{\pre} \cap S_{\stored})$ correspond to potential states in the hypothesis automaton which PL* will construct. The elements of $S_{\suff}$ are used to distinguish the states of the hypothesis automaton. Rows labelled by elements of $((S_{\pre} \cdot A) \cap S_{\stored})$ are used when constructing the hypothesis automaton's transition function.

A PL* observation table initially has $S_{\pre} = S_{\suff} = \{\varepsilon\}$ and $S_{\stored} = S_{\deadpref} = \{\}$. The sets $S_{\pre}$ and $S_{\suff}$ are augmented during the learning process, which in turn may trigger updates to $S_{\stored}$, $T$ and $S_{\deadpref}$.

\noindent\textbf{Outline}\label{obs-table-outline}

After we formally define the L* observation table equivalent to a given PL* observation table (Section~\ref{obs-table-equivalence-pl*-lstar*}), the remainder of Section~\ref{obs-table} will be dedicated to filling in some of the gaps in Algorithm~\ref{pl*}. These gaps, together with the corresponding line numbers in Algorithm~\ref{pl*}, are listed below:

\begin{itemize}
    \item \textbf{How does PL* update $S_{\deadpref}$ when a prefix query is asked?} (Section~\ref{obs-table-ask-pref-query}) Prefix queries may be asked when strings are added to $S_{\pre}$ or $S_{\suff}$ and the table is updated accordingly i.e., in lines 2, 8, 13 and 20.
    \item \textbf{How does PL* determine whether or not $x$ is a member of $X$?} (Section~\ref{obs-table-member}) This is done whenever strings are added to $S_{\pre}$ or $S_{\suff}$ and $T$ is updated accordingly i.e., in lines 2, 8, 13 and 20.
    \item \textbf{How does PL* determine whether $x$ is a member, live prefix, or dead prefix of $X$?} (Section~\ref{obs-table-member-live-dead-pref}) This is done whenever strings are added to $S_{\pre}$ and the table is updated accordingly i.e., in lines 2, 8 and 20.
    \item \textbf{How does PL* update $S_{\stored}$, $S_{\deadpref}$ and $T$ after $S_{\pre}$ has been augmented (i.e., after strings have been added to $S_{\pre}$)?} (Section~\ref{obs-table-update-s-pre-standard}) This is done in lines 2, 8 and 20. (Section~\ref{obs-table-update-s-pre-optimised} describes an second, potentially more optimised procedure for performing these updates.)
    \item \textbf{What does it mean for an observation table to be closed and consistent?} (Section~\ref{obs-table-closed-consistent}) These properties are referenced in lines 4, 5 and 10.
    \item \textbf{How does PL* construct an automaton from a closed and consistent observation table?} (Section~\ref{obs-table-automata-construction}) This needs to be done in line 16.
\end{itemize}

We conclude by defining the notion of \emph{size} for both PL* and L* observation tables (Section~\ref{obs-table-size}).

\noindent\textbf{Equivalence between L* and PL* Observation Tables}\label{obs-table-equivalence-pl*-lstar*}

Though Example~\ref{example:obs-table} gave an example of equivalent PL* and L* observation tables, we now formally define the L* table equivalent to a given PL* table.

There are two key differences between the L* observation table \\
$(S_{\pre}, S_{\suff}, T_{\text{L*}})$ and PL* observation table \\
$(S_{\pre}, S_{\suff}, S_{\stored}, S_{\deadpref}, T_{\text{PL*}})$ which share the same $S_{\pre}$ and $S_{\suff}$:

\begin{enumerate}
    \item The PL* observation table has two additional components not found in the L* observation table -- $S_{\stored}$ and $S_{\deadpref}$.

    \item $T_{\text{L*}}$ is a function mapping from $((S_{\pre} \cup S_{\pre} \cdot A) \cdot S_{\suff})$ to $\{0, 1\}$, whereas $T_{\text{PL*}}$ is a function mapping from $(S_{\stored} \cdot S_{\suff})$ to $\{0, 1\}$.
\end{enumerate}

Every time PL* updates its observation table (i.e., by augmenting $S_{\pre}$ or $S_{\suff}$, and updating $S_{\stored}$, $T$ and $S_{\deadpref}$ accordingly), it will do so in a manner which preserves the following property: 

\begin{property}[Dead Prefix Property]\label{property-dead-pref}
    For any $s \in (S_{\pre} \cup S_{\pre} \cdot A)$, if $s \not\in S_{\stored}$, then $s$ is an extension of a string in $S_{\deadpref}$.
\end{property}

\begin{definition}[L* Table Equivalent to a PL* Table]\label{definition-table-eq}
    Let \\
    $(S_{\pre}, S_{\suff}, S_{\stored}, S_{\deadpref}, T_{\text{PL*}})$ be a PL* observation table. Then the equivalent L* observation table \\
    $(S_{\pre}, S_{\suff}, T_{\text{L*}})$ has $T_{\text{L*}}$ defined as follows:

    \begin{align*}
        T_{\text{L*}}(se) &= \begin{cases}
            T_{\text{PL*}}(se), &\text{if $s \in S_{\stored}$}\\
            0, &\text{otherwise}\\
        \end{cases}
    \end{align*}
\end{definition}

The correctness of Definition~\ref{definition-table-eq} follows from the dead prefix property.

\noindent\textbf{Asking a Prefix Query}\label{obs-table-ask-pref-query}

Whenever PL* wishes to ask a prefix query, it calls Algorithm~\ref{ask-pref-query-algo}, which updates $S_{\deadpref}$ if necessary.

\begin{algorithm}
\caption{Ask a prefix query with $x$}\label{ask-pref-query-algo}
\begin{flushleft}
    \hspace*{\algorithmicindent} \textbf{Input:} $x \in A^*$\\
    \hspace*{\algorithmicindent} \textbf{Output:} \member, \livepref or \deadpref
\end{flushleft}
\begin{algorithmic}[1]
    \STATE ask the teacher a prefix query with $x$
    \IF{the answer received from the teacher is \deadpref}
        \STATE add $x$ to $S_{\deadpref}$
    \ENDIF
    \RETURN the answer received from the teacher
\end{algorithmic}
\end{algorithm}

\noindent\textbf{Determining a String's Membership in the Target Language}\label{obs-table-member}

To determine whether or not some $x \in A^*$ is a member of $X$, PL* will analyse the table itself, and make a prefix query if needed, as described in Algorithm~\ref{member-algo}.

\begin{algorithm}
\caption{Determine whether or not $x \in X$}\label{member-algo}
\begin{flushleft}
    \hspace*{\algorithmicindent} \textbf{Input:} $x \in A^*$\\
    \hspace*{\algorithmicindent} \textbf{Output:} $0$ or $1$ \algorithmiccomment{$0$ indicates that $x \not\in X$; $1$ indicates that $x \in X$}
\end{flushleft}
\begin{algorithmic}[1]
    \IF{$x = se$, where $s \in S_{\stored}, e \in S_{\suff}$}
        \RETURN $\row_s(e)$
    \ELSIF{$x = sy$, for $s \in S_{\deadpref}$ and $y \in A^*$}
        \RETURN $0$ \algorithmiccomment{$x$ is the extension of a dead prefix, so clearly not in $X$}
    \ELSE
        \STATE call Algorithm~\ref{ask-pref-query-algo} with input $x$ \algorithmiccomment{asking a prefix query}
        \IF{answer received was \member}
            \RETURN $1$
        \ELSE
            \RETURN $0$ \algorithmiccomment{answer received was either \livepref or \deadpref}
        \ENDIF
    \ENDIF
\end{algorithmic}
\end{algorithm}

\noindent\textbf{Member, Live Prefix, or Dead Prefix?}\label{obs-table-member-live-dead-pref}

To determine whether some $x \in A^*$ is a member, live prefix, or dead prefix of $X$, PL* will analyse the table itself, and make a prefix query if needed, as described in Algorithm~\ref{member-live-dead-pref-algo}.

\begin{algorithm}
\caption{Determine whether $x$ is a member, live prefix, or dead prefix of $X$}\label{member-live-dead-pref-algo}
\begin{flushleft}
    \hspace*{\algorithmicindent} \textbf{Input:} $x \in A^*$\\
    \hspace*{\algorithmicindent} \textbf{Output:} \member, \livepref or \deadpref
\end{flushleft}
\begin{algorithmic}[1]
    \IF{$x = sy$, for $s \in S_{\deadpref}$ and $y \in A^*$}
        \RETURN \deadpref
    \ELSIF{$x = se$, where $s \in S_{\stored}, e \in S_{\suff}$}
        \IF{$\row_s(e) = 0$}
            \RETURN \livepref \algorithmiccomment{can't be a dead prefix as it's not an extension of a string in $S_{\deadpref}$}
        \ELSE
            \RETURN \member
        \ENDIF
    \ELSE
        \STATE call Algorithm~\ref{ask-pref-query-algo} with input $x$ \algorithmiccomment{asking a prefix query}
        \RETURN the answer received
    \ENDIF
\end{algorithmic}
\end{algorithm}

We now prove a useful lemma about Algorithm~\ref{member-live-dead-pref-algo}.

\begin{lemma}\label{lemma-dead-pref-identified}
    For some $x \in A^*$, if Algorithm~\ref{member-live-dead-pref-algo} determines that $x$ is a dead prefix of $X$, then, by the end of its execution, $x$ will be an extension of a string in $S_{\deadpref}$.
\end{lemma}

\begin{proof}
    There are two places where Algorithm~\ref{member-live-dead-pref-algo} might identify $x$ as being a dead prefix; in either line 1 or line 10. If in line 1, then clearly we are done. If in line 10, then Algorithm~\ref{ask-pref-query-algo} had been called with input $x$, and so it must have added $x$ to $S_{\deadpref}$.
\end{proof}

\noindent\textbf{Updating $S_\stored$, $S_\deadpref$ and $T$ After Augmenting $S_\pre$}\label{obs-table-update-s-pre-standard}

Let $s$ be a string which has just been added to $S_{\pre}$. For first $x = s$, and then each $x \in \{sa \mid a \in A\}$, PL* runs Algorithm~\ref{update-s-pre-algo} with input $x$.

\begin{algorithm}
\caption{Update $S_{\stored}$, $S_{\deadpref}$ and $T$ after adding $x$ to $S_{\pre}$}\label{update-s-pre-algo}
\begin{flushleft}
    \hspace*{\algorithmicindent} \textbf{Input:} $x \in A^*$\\
\end{flushleft}
\begin{algorithmic}[1]
    \STATE determine whether $x$ is a member, live prefix or dead prefix of $X$ using Algorithm~\ref{member-live-dead-pref-algo}
    \IF{$x \in S_{\pre}$, $x$ is a dead prefix of $X$, and no $s \in S_{\stored}$ is a dead prefix of $X$}
        \STATE add $x$ to $S_{\stored}$ and update $S_{\deadpref}$ and $T$ accordingly using prefix queries (PQs)
    \ELSIF{$x \in S_{\pre}$, $x$ is a dead prefix of $X$, and some $s \in S_{\stored} \cap S_{\pre}$ is a dead prefix of $X$}
        \STATE do nothing
    \ELSIF{$x \in S_{\pre}$, $x$ is a dead prefix of $X$, and, though there is no $s \in S_{\stored} \cap S_{\pre}$ which is a dead prefix, there \emph{is} an $s' \in (S_{\stored} \cap (S_{\pre} \cdot A))$ which is a dead prefix}
        \STATE remove $s'$ from $S_{\stored}$
        \STATE add $x$ to $S_{\stored}$
        \STATE update $S_{\deadpref}$ and $T$ accordingly using PQs
    \ELSIF{$x \in (S_{\pre} \cdot A) \setminus S_{\pre}$, $x$ is a dead prefix of $X$, and no $s \in S_{\stored}$ is a dead prefix of $X$}
        \STATE add $x$ to $S_{\stored}$ and update $S_{\deadpref}$ and $T$ accordingly using PQs
    \ELSIF{$x \in (S_{\pre} \cdot A) \setminus S_{\pre}$, $x$ is a dead prefix of $X$, and some $s \in S_{\stored} $ is a dead prefix of $X$}
        \STATE do nothing
    \ELSE
        \STATE add $x$ to $S_{\stored}$ and update $S_{\deadpref}$ and $T$ accordingly using PQs\algorithmiccomment{$x \in \prefix(X)$}
    \ENDIF
\end{algorithmic}
\end{algorithm}

We now briefly explain how Algorithm~\ref{update-s-pre-algo} ensures the preservation of the dead prefix property (Property~\ref{property-dead-pref}). Recall that the dead prefix property refers to the following: For any $s \in (S_{\pre} \cup S_{\pre} \cdot A)$, if $s \not\in S_{\stored}$, then $s$ is an extension of a string in $S_{\deadpref}$. There are two ways a string $s \in (S_{\pre} \cup S_{\pre} \cdot A)$ ends up not being in $S_{\stored}$: it is never added in the first place (lines 5 and 13), or it is removed (line 7). If we reach either line 5 or 13, then we know that Algorithm~\ref{member-live-dead-pref-algo} had identified $s$ as being a dead prefix. From Lemma~\ref{lemma-dead-pref-identified}, we know that $s$ must be an extension of some string in $S_{\deadpref}$, as required. In line 7, the string $s'$ removed from $S_{\stored}$ is a dead prefix. Before $s'$ had been originally added to $S_{\stored}$, Algorithm~\ref{member-live-dead-pref-algo} would have identified it as being a dead prefix -- and so, from Lemma~\ref{lemma-dead-pref-identified}, $s'$ must be an extension of some string in $S_{\deadpref}$, as required.

\noindent\textbf{Updating $S_\stored$, $S_\deadpref$ and $T$ After Augmenting $S_\pre$ -- Optimised}\label{obs-table-update-s-pre-optimised}

Though the observation table behaviour described in Section~\ref{obs-table-update-s-pre-standard} is what we will analyse in Sections~\ref{correctness-termination} and~\ref{comparison-with-lstar}, in Section~\ref{chap:empirical-evaluation} we will also evaluate a second, more optimised version of PL*, which behaves slightly differently. This second version of PL* will be the subject of this section. We will refer to the version of PL* described in Section~\ref{obs-table-update-s-pre-standard} as ``standard PL*'', and the version described in this section as ``optimised PL*''.

Let $s$ be a string which has just been added to $S_{\pre}$. For first $x = s$, and then each $x \in \{sa \mid a \in A\}$, optimised PL* runs Algorithm~\ref{update-s-pre-optimised-algo} with input $x$.

\begin{algorithm}
\caption{Update $S_{\stored}$, $S_{\deadpref}$ and $T$ after adding $x$ to $S_{\pre}$ -- optimised}\label{update-s-pre-optimised-algo}
\begin{flushleft}
    \hspace*{\algorithmicindent} \textbf{Input:} $x \in A^*$\\
\end{flushleft}
\begin{algorithmic}[1]
    \STATE determine whether $x$ is a member, live prefix or dead prefix of $X$ using Algorithm~\ref{member-live-dead-pref-algo}
    \IF{$x \in S_{\pre}$, $x$ is a dead prefix of $X$, and no $s \in S_{\stored}$ is a dead prefix of $X$}
        \STATE add $x$ to $S_{\stored}$ and update $S_{\deadpref}$ and $T$ accordingly using prefix queries (PQs)
    \ELSIF{$x \in S_{\pre}$, $x$ is a dead prefix of $X$, and some $s \in S_{\stored} \cap S_{\pre}$ is a dead prefix of $X$}
        \STATE do nothing
    \ELSIF{$x \in (S_{\pre} \cdot A) \setminus S_{\pre}$, $x$ is a dead prefix of $X$, and no $s \in S_{\stored}$ is a dead prefix of $X$}
        \STATE add $x$ to $S_{\pre}$ and $S_{\stored}$ and update $S_{\deadpref}$ and $T$ accordingly using PQs
    \ELSIF{$x \in (S_{\pre} \cdot A) \setminus S_{\pre}$, $x$ is a dead prefix of $X$, and some $s \in S_{\stored} $ is a dead prefix of $X$}
        \STATE do nothing
    \ELSE
        \STATE add $x$ to $S_{\stored}$ and update $S_{\deadpref}$ and $T$ accordingly using PQs \algorithmiccomment{$x \in \prefix(X)$}
    \ENDIF
\end{algorithmic}
\end{algorithm}

In Algorithm~\ref{update-s-pre-optimised-algo}, we only consider five cases, whereas in Algorithm~\ref{update-s-pre-algo} we had six. This is because here we omit the following case: $x \in S_{\pre}$, $x$ is a dead prefix of $X$, and, though there is no $s \in S_{\stored} \cap S_{\pre}$ which is a dead prefix, there \emph{is} an $s' \in (S_{\stored} \cap (S_{\pre} \cdot A))$ which is a dead prefix. The reason why we are able to omit this case here is that, whenever we add some dead prefix $s$ to $S_{\stored}$ (lines 3 and 7), $s$ is always in $S_{\pre}$ by the time the addition step is complete. So, in subsequent addition steps, every $x \in S_{\stored}$ which is a dead prefix will also be in $S_{\pre}$. Thus, if the omitted case were to be included, it would never arise.

We hypothesise that this version of PL* is more optimal than the version described in Section~\ref{obs-table-update-s-pre-standard}. Intuitively, if we are in line 7, we have identified that some string $x \in A^*$ is a dead prefix of $X$. Thus, the target DFA must have some state whose language is $\varnothing$ -- since otherwise, for every dead prefix $s \in A^*$ of $X$, the target DFA would wrongly accept some extension of $s$. Due to the way that a hypothesis DFA is constructed (see Section~\ref{obs-table-automata-construction} for details), it would only have a state whose language is $\varnothing$ if there is some dead prefix in $S_{\pre}$ (the reasoning here is fairly informal; we provide a more formal statement and justification of this in Lemma~\ref{lemma-some-bad-prefix-row} and its proof). Thus, it seems sensible to \emph{immediately} add $x$ to $S_{\pre}$ (as optimised PL* does) rather than waiting for $x$ or some other dead prefix to be added to $S_{\pre}$ some time later (which is what would happen in standard PL*). In Section~\ref{empirical-evaluation-exact}, we compare the performance of standard and optimised PL* to determine if the latter is indeed more optimal in practice. However, for the remainder of Section~\ref{obs-table}, and in Sections~\ref{learner}, \ref{correctness-termination} and \ref{comparison-with-lstar}, we will be analysing only \emph{standard} (not optimised) PL*. 

\noindent\textbf{Closedness and Consistency}\label{obs-table-closed-consistent}

During its execution, PL* always tries to maintain a \emph{closed} and \emph{consistent} observation table. We now define these two notions for a PL* observation table, and show that they are equivalent to the definitions given by \citet{RN6} for an L* observation table.

Let $(S_{\pre}, S_{\suff}, S_{\stored}, S_{\deadpref}, T_{\text{PL*}})$ be a PL* observation table, and $(S_{\pre}, S_{\suff}, T_{\text{L*}})$ the equivalent L* observation table (Definition~\ref{definition-table-eq}). Let $\row_s$ denote a function mapping $S_{\suff}$ to $\{0, 1\}$, where $\row_s(e) = T_{\text{L*}}(se)$ (note that $T_{\text{L*}}(se) = T_{\text{PL*}}(se)$ for every $e \in S_{\suff}$ if $s \in S_{\stored}$). Figure~\ref{fig:obs-table-closed-consistent-setup} provides a visualisation of this setup.

\begin{figure*}[t]
\centering
\includegraphics[width=\textwidth]{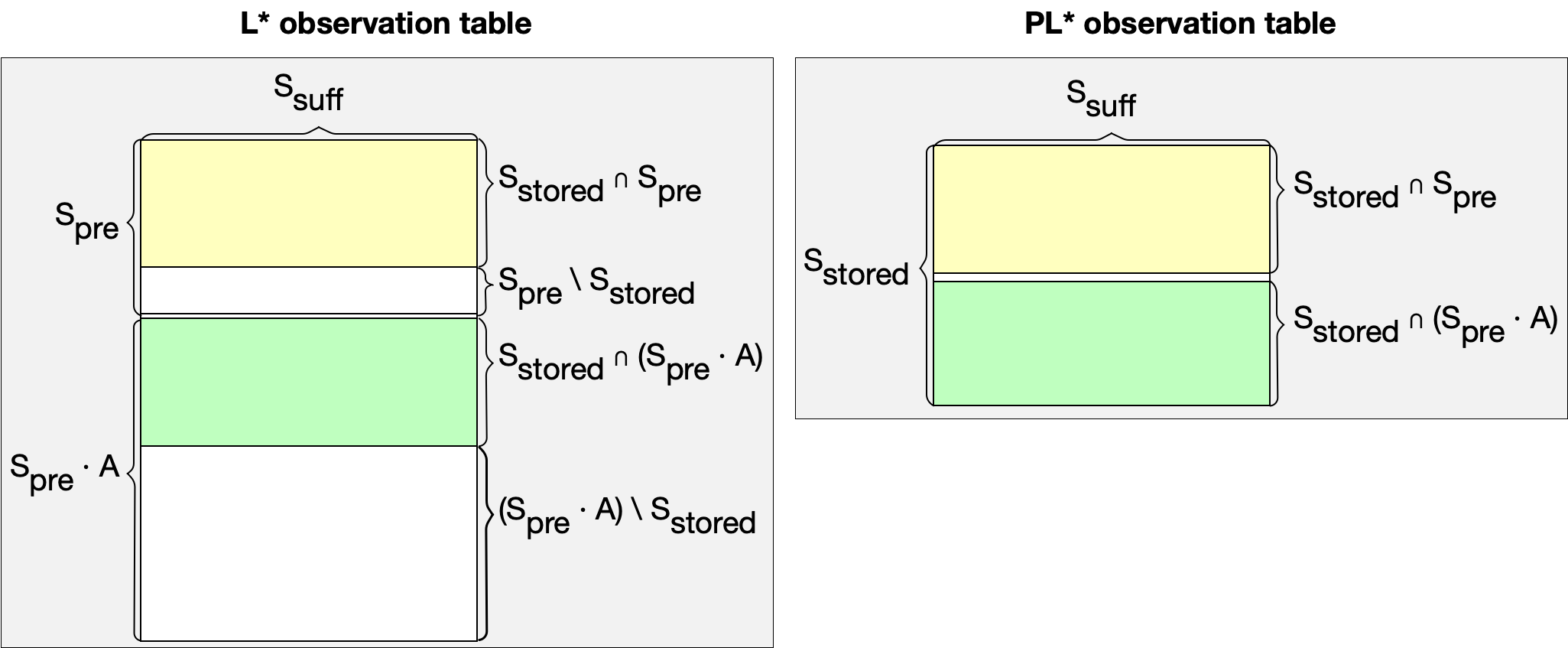}
\caption{Visualisation of our setup}\label{fig:obs-table-closed-consistent-setup}
\end{figure*}

\begin{definition}[Closedness]\label{definition-closedness}
    We say that a PL* observation table is \emph{closed} if, for all $s \in S_{\pre}, a \in A$ such that $sa \in S_{\stored}$, $\row_{sa} = \row_{s'}$ for some $s' \in S_{\pre} \cap S_{\stored}$.
\end{definition}

We now prove the following lemma:

\begin{lemma}\label{lemma-closedness}
    The PL* table \\
    $(S_{\pre}, S_{\suff}, S_{\stored}, S_{\deadpref}, T_{\text{PL*}})$ is closed if and only if the equivalent L* table $(S_{\pre}, S_{\suff}, T_{\text{L*}})$ is closed.
\end{lemma}

\begin{proof}
    We assume that the PL* observation table \\
    $(S_{\pre}, S_{\suff}, S_{\stored}, S_{\deadpref}, T_{\text{PL*}})$ is closed. If $(S_{\pre} \cdot A) \setminus S_{\stored}$ is empty, then for every $s \in S_{\pre}$ and $a \in A$, $sa \in S_{\stored}$. Hence, by the closedness of the PL* observation table, the L* table is clearly also closed, as required. Now, assume that $(S_{\pre} \cdot A) \setminus S_{\stored}$ is non-empty. In order to show that $(S_{\pre}, S_{\suff}, T_{\text{L*}})$ is closed, we only need to verify that, for each $s \in S_{\pre}$ and $a \in A$ such that $sa \not\in S_{\stored}$, $\row_{sa} = \row_{s'}$ for some $s' \in S_{\pre}$. Due to the dead prefix property  (Property~\ref{property-dead-pref}), we know that any $sa \in (S_{\pre} \cdot A) \setminus S_{\stored}$ is a dead prefix, and so $\row_{sa}$ would just map every $e \in S_{\suff}$ to 0. Furthermore, the very fact that $(S_{\pre} \cdot A) \setminus S_{\stored}$ is non-empty means that some $sa \in S_{\pre} \cdot A$ was either never added to $S_{\stored}$ or was added and later removed from it. If $sa$ was never added to $S_{\stored}$, this means that, after $s$ was added to $S_{\pre}$, Algorithm~\ref{update-s-pre-algo} when run with input $sa$ would've ended up in line 13 -- implying that some $x \in S_{\stored}$ was a dead prefix. (Even if this $x$ had later been removed from $S_{\stored}$, Algorithm~\ref{update-s-pre-algo} would have \emph{immediately} afterwards added some other dead prefix to $S_{\stored}$ in its place.) If $sa$ had been added to $S_{\stored}$ but then later removed, Algorithm~\ref{update-s-pre-algo} would've added some other dead prefix to $S_{\stored}$ \emph{immediately} after removing $sa$. Thus, in either case, we have that there is some $x \in S_{\stored}$ which is a dead prefix. Since $x$ is a dead prefix, $\row_x$ would map every $e \in S_{\suff}$ to 0. If $x \in S_{\stored} \cap S_{\pre}$ then $\row_{sa} = \row_x$, and we are done. Otherwise, since the PL* observation table is closed, there would be some $x' \in S_{\pre} \cap S_{\stored}$ such that $\row_x = \row_{x'}$, and so $\row_{sa} = \row_{x'}$.

    Now we prove the other direction. Assume now that the L* observation table $(S_{\pre}, S_{\suff}, T_{\text{L*}})$ is closed. If $S_{\pre} \setminus S_{\stored}$ is empty, then $S_{\pre} \cap S_{\stored} = S_{\pre}$. Thus, by the closedness of the L* observation table, the PL* table is clearly also closed, as required. Now, we assume that $S_{\pre} \setminus S_{\stored}$ is non-empty. In order to show that $(S_{\pre}, S_{\suff}, S_{\stored}, S_{\deadpref}, T_{\text{PL*}})$ is closed, we only need to verify that, for every $s \in S_{\pre} \setminus S_{\stored}$, $\row_s = \row_{s'}$ for some $s' \in S_{\stored} \cap S_{\pre}$. Due to the dead prefix property, any $s \in S_{\pre} \setminus S_{\stored}$ is a dead prefix -- and so $\row_s$ would map every $e \in S_{\suff}$ to 0. The fact that $S_{\pre} \setminus S_{\stored}$ is non-empty (i.e., some string in $S_{\pre}$ was not added to $S_{\stored}$) implies that there is some $x \in S_{\pre} \cap S_{\stored}$ which is a dead prefix. Thus, $\row_s = \row_x$ -- and so we are done.
\end{proof}

Now, we turn our attention to the notion of consistency.

\begin{definition}[Consistency]\label{definition-consistency}
    Let $s \in (S_{\pre} \cup S_{\pre} \cdot A)$ be arbitrary. We define $r_{s}$ as follows:

    \begin{align*}
        r_s &= \begin{cases}
            \row_s, &\text{if $s \in S_{\stored}$}\\
            \row_x, &\text{otherwise,}\\
            &\text{where $x$ is a dead prefix which is also an element}\\
            &\text{of $S_{\stored}$}\\
        \end{cases}
    \end{align*}
    
    We say that a PL* observation table is \emph{consistent} if, for all $s_1, s_2 \in S_{\pre} \cap S_{\stored}$ such that $\row_{s_1} = \row_{s_2}$, $r_{s_1a} = r_{s_2a}$ for all $a \in A$.  
\end{definition}

First, let us verify that $r_{s}$ is well-defined. If $s \in S_{\stored}$, then $\row_s$ is indeed a function mapping $S_{\suff}$ to $\{0, 1\}$, where for each $e \in S_{\suff}$, $\row_s(e) = T_{\text{PL*}}(se)$. Now assume that $s \not\in S_{\stored}$. How do we know that there exists an $x \in S_{\stored}$ which is also a dead prefix? If $s \not\in S_{\stored}$, this means that it was either not added in the first place, or it was added but then later removed. In Algorithm~\ref{update-s-pre-algo}, we see that a string $s$ is not added to $S_{\stored}$ only if it is a dead prefix, and some $x \in S_{\stored}$ is a dead prefix (lines 5 and 13). In the case that $s$ was added but then later removed, the removal would only have occurred if there was some $s' \in S_{\stored}$ which was a dead prefix (line 7). So, $r_{s}$ is indeed well-defined.

We now prove the following lemma:

\begin{lemma}\label{lemma-consistency}
    The PL* table $(S_{\pre}, S_{\suff}, S_{\stored}, S_{\deadpref}, T_{\text{PL*}})$ is consistent if and only if the equivalent L* table $(S_{\pre}, S_{\suff}, T_{\text{L*}})$ is consistent.
\end{lemma}

\begin{proof}
    Assume that the PL* observation table \\
    $(S_{\pre}, S_{\suff}, S_{\stored}, S_{\deadpref}, T_{\text{PL*}})$ is consistent. We will now show that the equivalent L* observation table $(S_{\pre}, S_{\suff}, T_{\text{L*}})$ also is consistent. We take cases, and in each case prove that $\row_{s_1a} = \row_{s_2a}$:

    \begin{description}
        \item[Case 1] $s_1, s_2 \in S_{\pre} \cap S_{\stored}$, $a \in A$, $\row_{s_1} = \row_{s_2}$, and both $s_1a$ and $s_2a$ are in $S_{\stored}$
    
             We have that $\row_{s_1a} = \row_{s_2a}$ from the consistency of the PL* table.
    
        \item[Case 2] $s_1, s_2 \in S_{\pre}$ but they aren't both in $S_{\stored}$, $a \in A$, $\row_{s_1} = \row_{s_2}$, and both $s_1a$ and $s_2a$ are in $S_{\stored}$
    
            This case cannot actually occur in practice. Assume, without loss of generality, that $s_2 \not\in S_{\stored}$. Because of the dead prefix property (Property~\ref{property-dead-pref}), $s_2$ is a dead prefix -- and so $s_2a$ must be too. Based on how we have defined our observation table and learner (see Section~\ref{learner}), PL* considers adding $s_2$ to $S_{\stored}$ before it considers adding $s_2a$. $s_2$ not being added to $S_{\stored}$ means that there is some $s \in S_{\stored} \cap S_{\pre}$ which is a dead prefix (line 5 in Algorithm~\ref{update-s-pre-algo}). So, according to Algorithm~\ref{update-s-pre-algo}, $s_2a$ wouldn't have been added to $S_{\stored}$ either! We have a contradiction -- and so this case cannot occur in practice.  

        \item[Case 3] $s_1, s_2 \in S_{\pre} \cap S_{\stored}, a \in A$, $\row_{s_1} = \row_{s_2}$, and exactly one of $s_1a$ and $s_2a$ is in $S_{\stored}$
    
            Assume, without loss of generality, that $s_1a \in S_{\stored}, s_2a \not\in S_{\stored}$. Since the PL* observation table is consistent, we know that $r_{s_1a} = r_{s_2a}$ i.e., $\row_{s_1a} = \row_x$ for $x$ a dead prefix which is in $S_{\stored}$. Since $s_2a \not\in S_{\stored}$, from the dead prefix property we know that it's a dead prefix -- and so $\row_{s_2a}$ would be a function mapping each $e \in S_{\suff}$ to 0. This is equivalent to $\row_x$, and thus $\row_{s_1a}$, as required.

        \item[Case 4] $s_1, s_2 \in S_{\pre}$ but they aren't both in $S_{\stored}$, $a \in A$, $\row_{s_1} = \row_{s_2}$, and exactly one of $s_1a$ and $s_2a$ is in $S_{\stored}$

            Firstly, assume that $s_1 \not\in S_{\stored}$ and $s_2 \not\in S_{\stored}$. As we argued above, this would also mean that neither $s_1a$ nor $s_2a$ are in $S_{\stored}$ -- which contradicts our case definition. So, it must be the case that exactly one of $s_1$ and $s_2$ is in $S_{\stored}$.
    
            Assume, without loss of generality, that $s_2 \not\in S_{\stored}$ (and so $s_1 \in S_{\stored}$). As we argued above, it must also be the case that $s_2a \not\in S_{\stored}$ (and so $s_1a \in S_{\stored}$). $s_2$ not being added to $S_{\stored}$ means that $s_2$ is a dead prefix (Property~\ref{property-dead-pref}), and also that there is some $x \in S_{\stored} \cap S_{\pre}$ which is a dead prefix (line 5 in Algorithm~\ref{update-s-pre-algo}). Thus, $\row_x = \row_{s_2}$, and so $\row_x = \row_{s_1}$. By the consistency of the PL* table, we have that $r_{xa} = \row_{s_1a}$. Since $x$ is a dead prefix, $xa$ would also be a dead prefix, and so $r_{xa}$ would be a function mapping each $e \in S_{\suff}$ to $0$. Thus, $\row_{s_1a}$ would also be a function mapping each $e \in S_{\suff}$ to $0$. Since $s_2a \not\in S_{\stored}$, $s_2a$ is a dead prefix (Property~\ref{property-dead-pref}), and so $\row_{s_2a}$ would be a function mapping each $e \in S_{\suff}$ to $0$. Thus, $\row_{s_1a} = \row_{s_2a}$, as required.
    
        \item[Case 5] $s_1, s_2 \in S_{\pre}, a \in A$, $\row_{s_1} = \row_{s_2}$, and neither $s_1a$ nor $s_2a$ are in $S_{\stored}$
    
            If neither $s_1a$ nor $s_2a$ are in $S_{\stored}$, we have by the dead prefix property that $s_1a$ and $s_2a$ are both dead prefixes. Thus, both $\row_{s_1a}$ and $\row_{s_2a}$ would be functions mapping each $e \in S_{\suff}$ to 0, and hence would be equivalent as required.
    \end{description}
    
    Now we prove the other direction. Assume that the L* observation table $(S_{\pre}, S_{\suff}, T_{\text{L*}})$ is consistent. We will now show that the equivalent PL* observation table $(S_{\pre}, S_{\suff}, S_{\stored}, S_{\deadpref}, T_{\text{PL*}})$ also is consistent. We take cases, and in each case prove that $r_{s_1a} = r_{s_2a}$:
    
    \begin{description}
        \item[Case 1] $s_1, s_2 \in S_{\stored} \cap S_{\pre}$, $a \in A$, $\row_{s_1} = \row_{s_2}$ and $s_1a, s_2a \in S_{\stored}$.

            From the consistency of the L* table, we know that $\row_{s_1a} = \row_{s_2a}$. Since $s_1a, s_2a \in S_{\stored}$, $r_{s_1a} = \row_{s_1a}$ and $r_{s_2a} = \row_{s_2a}$. Hence, $r_{s_1a} = r_{s_2a}$ as required.
    
        \item[Case 2] $s_1, s_2 \in S_{\stored} \cap S_{\pre}, a \in A, \row_{s_1} = \row_{s_2}$ and exactly one of $s_1a$ and $s_2a$ is in $S_{\stored}$.
    
            Assume, without loss of generality, that $s_1a \in S_{\stored}, s_2a \not\in S_{\stored}$. Since the L* table is consistent, we know that $\row_{s_1a} = \row_{s_2a}$. Since $s_1a \in S_{\stored}$, we know that $r_{s_1a} = \row_{s_1a}$ Since $s_2a \not\in S_{\stored}$, we know due to the dead prefix property that $s_2a$ is a dead prefix -- and so $\row_{s_2a}$ would be a function mapping each $e \in S_{\suff}$ to 0. This is equivalent to $\row_x$, for $x$ a dead prefix which is also in $S_{\stored}$, and thus $r_{s_2a}$. So, we have that $\row_{s_1a} = r_{s_1a}$ and $\row_{s_2a} = \row_x = r_{s_2a}$. Since $\row_{s_1a} = \row_{s_2a}$, we have that $r_{s_1a} = r_{s_2a}$, as required.
    
        \item[Case 3] $s_1, s_2 \in S_{\stored} \cap S_{\pre}, a \in A, \row_{s_1} = \row_{s_2}$ and neither $s_1a$ nor $s_2a$ are in $S_{\stored}$.
    
            Since the L* table is consistent, we know that $\row_{s_1a} = \row_{s_2a}$. Since neither $s_1a$ nor $s_2a$ are in $S_{\stored}$, we know by the dead prefix property that $s_1a$ and $s_2a$ are both dead prefixes. Thus, $\row_{s_1a}$ and $\row_{s_2a}$ would both be equivalent to $\row_x$, for $x$ a dead prefix which is also in $S_{\stored}$, and thus $r_{s_1a}$ and $r_{s_2a}$. Hence, we have that $r_{s_1a} = r_{s_2a}$, as required.
    \end{description}
\end{proof}

\noindent\textbf{Constructing an Automaton from a Closed and Consistent Observation Table}\label{obs-table-automata-construction}

If $(S_{\pre}, S_{\suff}, S_{\stored}, S_{\deadpref}, T)$ is a closed and consistent observation table, we define a corresponding DFA \\
$M(S_{\pre}, S_{\suff}, S_{\stored}, S_{\deadpref}, T)$ over the alphabet $A$, with state set $Q$, initial state $q_0$, final states $F$, and transition function $\delta$:

\begin{align*}
    Q &= \{\row_s \mid s \in S_{\pre} \cap S_{\stored}\} \\
    q_0 &= \row_{\varepsilon} \\
    F &= \{\row_s \mid s \in S_{\pre} \cap S_{\stored} \text{ and } T(s) = 1\} \\
    \delta(\row_s, a) &= 
        \begin{cases}
            \row_{sa}, &\text{if $sa \in S_{\stored}$}\\
            \row_x, &\text{otherwise,}\\
            &\text{where $x$ is a dead prefix which is also an element}\\
            &\text{of $S_{\stored}$}\\
        \end{cases}
\end{align*}

We now show that this DFA is well-defined, using an approach similar to what is used by \citet{RN6}. 

We first note that the initial value of $S_{\pre}$ is $\{\varepsilon\}$ -- and so $S_{\stored}$ must contain $\varepsilon$ (in the trivial case where the target language is $\varnothing$, and so $\varepsilon$ is a dead prefix, the first, and only, string added to $S_{\stored}$ would be $\varepsilon$). Thus, $q_0$ is well-defined. 

To see that $F$ is well-defined, we note that, if $s_1, s_2 \in S_{\pre} \cap S_{\stored}$ such that $\row_{s_1} = \row_{s_2}$, this means that, for every $e \in S_{\suff}$, $T(s_1e) = T(s_2e)$. Since $\varepsilon \in S_{\suff}$ (it is the first element added to $S_{\suff}$, and never removed), we have that $T(s_1) = T(s_2)$, and so $F$ is well-defined. 

Finally, we show that $\delta$ is well-defined. In other words, we must show that, for every $s_1, s_2 \in S_{\pre} \cap S_{\stored}$ such that $\row_{s_1} = \row_{s_2}$, and every $a \in A$, $\delta(\row_{s_1}, a) = \delta(\row_{s_2}, a) = \row_s$ for some $s \in S_{\pre} \cap S_{\stored}$. 

Let $a \in A$ and $s_1, s_2 \in S_{\pre} \cap S_{\stored}$ such that $\row_{s_1} = \row_{s_2}$ be arbitrary.

We now take cases:

\begin{description}
    \item[Case 1] Both $s_1a$ and $s_2a$ are in $S_{\stored}$.

        Since $s_1a$ and $s_2a$ are both in $S_{\stored}$, we have that $\delta(\row_{s_1}, a) = \row_{s_1a}$ and $\delta(\row_{s_2}, a) = \row_{s_2a}$. Because the observation table is consistent, we have that $\row_{s_1a} = \row_{s_2a}$. Furthermore, since the table is closed, we have that $\row_{s_1a} = \row_{s_2a} = \row_s$ for some $s \in S_{\pre} \cap S_{\stored}$.

    \item[Case 2] Exactly one of $s_1a$ and $s_2a$ is in $S_{\stored}$.

        Assume, without loss of generality, that $s_1a$ is in $S_{\stored}$, but $s_2a$ is not. Thus, according to our definition of $\delta$, $\delta(\row_{s_1}, a) = \row_{s_1a}$, whilst $\delta(\row_{s_2}, a) = \row_x$, where $x \in S_{\stored}$ is a dead prefix. How do we know that such a $\row_x$ exists? $s_2a$ not being in $S_{\stored}$ means that either it was never added in the first place, or it was added but then later removed. Either way, this indicates that some $x \in S_{\stored}$ is a dead prefix.
        
        Since the observation table is consistent, $r_{s_1a} = r_{s_2a}$ i.e., $\row_{s_1a} = \row_x$ for $x \in S_{\stored}$ a dead prefix. Thus, $\delta(\row_{s_1}, a) = \delta(\row_{s_2}, a) = \row_x$. If $x \in S_{\pre}$ we are done; otherwise, if $x \in ((S_{\pre} \cdot A) \setminus S_{\pre})$, then, because the table is closed, $\row_x = \row_s$ for some $s \in S_{\pre} \cap S_{\stored}$.

    \item[Case 3] Neither $s_1a$ nor $s_2a$ are in $S_{\stored}$.

        Since neither $s_1a$ nor $s_2a$ are in $S_{\stored}$, we have that $\delta(\row_{s_1}, a) = \delta(\row_{s_2}, a) = \row_x$, where $x \in S_{\stored}$ is a dead prefix. Like we did in the previous case, we can deduce that such a $\row_x$ indeed exists. If $x \in S_{\pre}$ we are done; otherwise, if $x \in ((S_{\pre} \cdot A) \setminus S_{\pre})$, then, because the table is closed, $\row_x = \row_s$ for some $s \in S_{\pre} \cap S_{\stored}$.
\end{description}

Since $a, s_1, s_2$ were all arbitrary, we have that, for every $s_1, s_2 \in S_{\pre}$ such that $\row_{s_1} = \row_{s_2}$, and every $a \in A$, $\delta(\row_{s_1}, a) = \delta(\row_{s_2}, a) = \row_s$ for some $s \in S_{\pre} \cap S_{\stored}$, as required.

We now explain how $M\obstable$ is equivalent to the automaton that would be constructed from the equivalent L* observation table. In order to do so, we need to address the following:

\begin{enumerate}
    \item Is the set of states in $M\obstable$ equivalent to the set of states in the automaton constructed from the equivalent L* table?

    \item Why is it correct to let $\delta(\row_s, a) = \row_x$ (for $x$ a dead prefix which is in $S_{\stored}$) if $sa \not\in S_{\stored}$?
\end{enumerate}

We begin by addressing the first point. The set of states in the automaton constructed from the PL* table is $\{\row_s \mid s \in S_{\pre} \cap S_{\stored}\}$, whereas the set of states in the automaton constructed from the equivalent L* table would be $\{\row_s \mid s \in S_{\pre}\}$. If $S_{\pre} \setminus S_{\stored}$ is empty, $S_{\pre} = S_{\pre} \cap S_{\stored}$, and so we are done. Thus, assume $S_{\pre} \setminus S_{\stored}$ is non-empty. Consider an arbitrary $s \in S_{\pre} \setminus S_{\stored}$. Due to the dead prefix property (Property~\ref{property-dead-pref}), $s$ would be a dead prefix -- and the reason why it wouldn't have been added to $S_{\stored}$ is that there had already been some $x \in S_{\pre} \cap S_{\stored}$ which was a dead prefix (line 5 of Algorithm~\ref{update-s-pre-algo}). Thus, $\row_x = \row_s$. Thus, for every $s \in S_{\pre} \setminus S_{\stored}$, $\row_s = \row_x$ for $x \in S_{\pre} \cap S_{\stored}$ a dead prefix -- and so $\{\row_s \mid s \in S_{\pre} \cap S_{\stored}\} = \{\row_s \mid s \in S_{\pre}\}$. Thus, the set of states in $M\obstable$ is indeed equivalent to the set of states in the automaton constructed from the equivalent L* table, as required.

Now, we discuss the second point. We know that, when L* constructs its hypothesis automaton, it would set $\delta(\row_s, a)$ to $\row_{sa}$. If $sa \not\in S_{\stored}$, we know from the dead prefix property that $sa$ is a dead prefix. Thus, $\row_{sa} = \row_x$ for $x$ a dead prefix which is in $S_{\stored}$. Thus, in setting $\delta(\row_s, a)$ to $\row_x$ (for $x$ a dead prefix which is in $S_{\stored}$), PL* indeed correctly mimics the behaviour of L*.

\noindent\textbf{Size of an Observation Table}\label{obs-table-size}

One of the metrics we use to compare L* and PL* in Section~\ref{chap:empirical-evaluation} (Empirical Evaluation) is the \emph{size} of each observation table. We now define the notion of size for both L* and PL* observation tables.

\begin{definition}[Size of an L* Observation Table]\label{definition-l*-table-size}
    Let $(S_{\pre}, S_{\suff}, T_{\text{L*}})$ be an L* observation table. The size of this table is equal to
    
    $$|S_{\pre}| + |S_{\suff}| + (|S_{\pre} \cup (S_{\pre} \cdot A)| \times |S_{\suff}|)$$ 
\end{definition}

\begin{definition}[Size of a PL* Observation Table]\label{definition-pl*-table-size}
    Let \\
    $(S_{\pre}, S_{\suff}, S_{\stored}, S_{\deadpref}, T_{\text{PL*}})$ be a PL* observation table. The size of this table is equal to
    
    \begin{align*}
        |S_{\pre}| + |S_{\suff}| + |S_{\stored}| + |S_{\deadpref}| + (|S_{\stored}| \times |S_{\suff}|)
    \end{align*}
\end{definition}

\subsection{Correctness and Termination of PL*}\label{correctness-termination}

\subsubsection{Correctness}\label{correctness}

PL* only terminates if the teacher responds \y to a conjecture (line 22 of Algorithm~\ref{pl*}). Hence, if the teacher can accurately answer prefix and equivalence queries, then, if PL* ever terminates, its output is a DFA which recognises exactly the target language.

\subsubsection{Termination}\label{termination}

In describing the notion of a MAT, \citet{RN6} does not state that it is deterministic. Let's consider what a non-deterministic MAT would look like. Let's say that we present the incorrect conjecture $\mathcal{A}$ to a non-deterministic MAT for unknown regular target language $X \subseteq A^*$. As the conjecture was incorrect, the MAT would return some $x \in X \oplus L(\mathcal{A})$. After receiving $x$, suppose that we present $\mathcal{A}$ to the MAT a second time. Then, since the MAT is non-deterministic, it may on this occasion return a counterexample $y \in X \oplus L(\mathcal{A})$, where $y \neq x$. 

Though the teacher which PL* interacts with is not a MAT, it answers equivalence queries in exactly the same way that a MAT does. Thus, if the teacher used by PL* was non-deterministic, then, like the non-deterministic MAT, it may on different occasions present \emph{different} counterexamples in response to the \emph{same} incorrect conjecture. However, our proof of the termination of PL* requires us to assume that the teacher is deterministic:

\begin{assumption}\label{assumption-teacher}
    Consider the setting where the unknown language is $X \subseteq A^*$. Let $\mathcal{A}$ be an arbitrary incorrect conjecture (i.e., a conjecture satisfying $L(\mathcal{A}) \neq X$). Then, each time the teacher is presented with $\mathcal{A}$, it responds with the \textup{same} counterexample $x \in X \oplus L(\mathcal{A})$.
\end{assumption} 

Making this assumption allows us to eliminate the non-determinism due to choice of counterexample -- thus allowing us to prove concrete theoretical results.

Each execution of L* or PL* consists of one or more \emph{rounds}, where each round involves the learner making the table closed and consistent, before constructing a hypothesis and presenting it to the teacher. If the teacher provides a counterexample, this is added to the observation table, and then the next round begins.

We prove the termination of PL* by arguing that it ``simulates'' L*. This is stated more formally in the following theorem:

\begin{theorem}\label{theorem-termination}
    Let $X \subseteq A^*$ be the unknown target language. Let $(S_{\pre}, S_{\suff}, S_{\stored}, S_{\deadpref}, T_{\text{PL*}})$ denote the observation table maintained by PL*, and $(S_{\pre}, S_{\suff}, T_{\text{L*}})$ the table maintained by L*, when each is tasked with learning $X$. Then, after the $n^{\text{th}}$ round (for $n \geq 0$), $(S_{\pre}, S_{\suff}, T_{\text{L*}})$ is the L* table equivalent to $(S_{\pre}, S_{\suff}, S_{\stored}, S_{\deadpref}, T_{\text{PL*}})$ (where Definition~\ref{definition-table-eq} defines the L* table equivalent to a given PL* table).
\end{theorem}

\begin{proof}
    We prove this theorem via induction on $n$.

    \begin{description}
        \item[Base case] $n = 0$
        
            We need to show that, before the first learning round begins, $(S_{\pre}, S_{\suff}, T_{\text{L*}})$ is the L* table equivalent to\\
            $(S_{\pre}, S_{\suff}, S_{\stored}, S_{\deadpref}, T_{\text{PL*}})$.
            
            When constructing their initial observation tables, both L* and PL* set $S_{\pre}$ and $S_{\suff}$ to $\{\varepsilon\}$. PL* adds a subset of $(S_{\pre} \cup S_{\pre} \cdot A)$ to $S_{\stored}$ -- where the strings \emph{not} added to $S_{\stored}$ are all dead prefixes (due to the dead prefix property i.e., Property~\ref{property-dead-pref}). All minimal dead prefixes are added to $S_{\deadpref}$. Let $(S_{\pre}, S_{\suff}, T'_{\text{L*}})$ denote the L* table equivalent to the initial PL* table. According to Definition~\ref{definition-table-eq}, $T'_{\text{L*}}(s \cdot \varepsilon) = T_{\text{PL*}}(s \cdot \varepsilon)$ for all $s \in S_{\stored} \cap (S_{\pre} \cup S_{\pre} \cdot A)$ (where $T_{\text{PL*}}(s \cdot \varepsilon)$ is 1 if $s \cdot \varepsilon \in X$, and 0 otherwise). For all $s \in (S_{\pre} \cup S_{\pre} \cdot A) \setminus S_{\stored}$ (which are dead prefixes), $T'_{\text{L*}}(s \cdot \varepsilon) = 0$, and indeed $s \cdot \varepsilon \not\in X$ by definition of ``dead prefix''. Thus, for all $s \in (S_{\pre} \cup S_{\pre} \cdot A)$, it is clearly the case that $T'_{\text{L*}}(s \cdot \varepsilon)$ is set to $1$ if $s \in X$, and $0$ otherwise -- which is exactly how L* would set $T_{\text{L*}}(s \cdot \varepsilon)$. Thus, $T_{\text{L*}} = T'_{\text{L*}}$, and so $(S_{\pre}, S_{\suff}, T_{\text{L*}}) = (S_{\pre}, S_{\suff}, T'_{\text{L*}})$, as required.
            
        \item[Inductive hypothesis] $n = k$

            We assume that, after the $k^{\text{th}}$ learning round, $(S_{\pre}, S_{\suff}, T_{\text{L*}})$ is the L* table equivalent to $(S_{\pre}, S_{\suff}, S_{\stored}, S_{\deadpref}, T_{\text{PL*}})$.

        \item[Inductive step] $n = k + 1$ 

            We know that, at the \emph{start} of the $(k+1)^{\text{st}}$ round, $(S_{\pre}, S_{\suff}, T_{\text{L*}})$ is the L* table equivalent to $(S_{\pre}, S_{\suff}, S_{\stored}, S_{\deadpref}, T_{\text{PL*}})$ (from the inductive hypothesis). We now need to show that, following all the modifications made by L* and PL* to their respective tables during round $(k+1)$ (including making them closed and consistent, and potentially adding a counterexample), it is still the case that $(S_{\pre}, S_{\suff}, T_{\text{L*}})$ is the L* table equivalent to $(S_{\pre}, S_{\suff}, S_{\stored}, S_{\deadpref}, T_{\text{PL*}})$.

            We first show that, after L* and PL* make their respective tables closed and consistent, $(S_{\pre}, S_{\suff}, T_{\text{L*}})$ remains the L* table equivalent to $(S_{\pre}, S_{\suff}, S_{\stored}, S_{\deadpref}, T_{\text{PL*}})$. We showed in Section~\ref{obs-table-closed-consistent} that a PL* table is closed (consistent) if and only if the equivalent L* table is closed (consistent). If at the beginning of the round, either table is both closed and consistent, then \emph{both} tables must be both closed and consistent -- and so we are trivially done. So, assume that neither table is both closed and consistent. In order to make its table closed and consistent, PL* would augment it by making a series of additions to $S_{\pre}$ and/or $S_{\suff}$. If L* also makes these same additions, its table will then become the L* table equivalent to PL*'s augmented table. Furthermore, because PL*'s augmented table is closed and consistent, L*'s augmented table would also be both closed and consistent. Thus, after L* and PL* make their own tables closed and consistent, $(S_{\pre}, S_{\suff}, T_{\text{L*}})$ remains the L* table equivalent to $(S_{\pre}, S_{\suff}, S_{\stored}, S_{\deadpref}, T_{\text{PL*}})$, as required. (Note that there is some non-determinism possible in both L* and PL* in terms of which strings they choose to add to $S_{\pre}$ and/or $S_{\suff}$ -- for the sake of this proof we assume that, at every stage, they make the same choices.)

            After making their observation tables closed and consistent, L* and PL* then construct their hypotheses. In Section~\ref{obs-table-automata-construction}, we argued that, given equivalent observation tables, L* and PL* construct identical hypothesis automata. In response, the teacher would either confirm that both hypotheses are correct -- in which case both learners would terminate (with equivalent observation tables) -- or otherwise would provide both learners with the same counterexample (given Assumption~\ref{assumption-teacher}). Both L* and PL* would add this counterexample and all its prefixes to $S_{\pre}$ -- and so $(S_{\pre}, S_{\suff}, T_{\text{L*}})$ would remain the L* table equivalent to $(S_{\pre}, S_{\suff}, S_{\stored}, S_{\deadpref}, T_{\text{PL*}})$.

            Thus, at the end of the $(k+1)^{\text{st}}$ round, $(S_{\pre}, S_{\suff}, T_{\text{L*}})$ is still the L* table equivalent to $(S_{\pre}, S_{\suff}, S_{\stored}, S_{\deadpref}, T_{\text{PL*}})$, as required.
    \end{description}
\end{proof}

Given the termination of L* \citep[p.~94-96]{RN6}, we have that, during some round $n$ (for $n \geq 1$), L* makes its table closed and consistent, constructs a hypothesis which the teacher deems correct, and then terminates returning this hypothesis. We have from Theorem~\ref{theorem-termination} that L*'s table at the end of round $n$ is the L* table equivalent to PL*'s table at the end of its round $n$. Thus, at the end of round $n$, PL* constructs the same correct hypothesis and returns it; i.e., PL* also terminates at the end of round $n$.

\subsection{Comparing PL* with L*}\label{comparison-with-lstar}

PL* makes several changes to the standard L* algorithm in order to exploit the additional information given by prefix queries over membership queries. However, these changes only apply when there are in fact strings $x \in A^*$ such that $x \not\in L, x \not\in \PP(L)$ i.e., strings which cause a \deadpref response to be returned from a prefix query. If $L$ is a dense language (Definition~\ref{definition-dense}), then there are no strings which give a \deadpref response, and so PL* would run exactly like standard L* i.e., it would have the same runtime.

For non-dense languages, however, PL* will always give some improvement over L*:

\begin{theorem}\label{theorem-non-dense}
    Let $L \subseteq A^*$ be an arbitrary non-dense regular language.  

    Let $(S_{\pre}, S_{\suff}, T_{\text{L*}})$ be L*'s final observation table, and\\
    $\obstable$ PL*'s final observation table, when both learners are tasked with learning $L$. Let $q$ be the number of membership queries performed by L* during its execution, and $q'$ the number of prefix queries performed by PL* during its execution.
    
    Then, 
    
    $$q - q' \geq |S_{\suff}| \times (1 + |A|) - 1$$
\end{theorem}

\begin{proof}    
    We begin by proving two lemmas.

    \citet[p.~91]{RN6} proves that, for the DFA $M(S_{\pre}, S_{\suff}, T_{\text{L*}})$ constructed from the L* observation table $(S_{\pre}, S_{\suff}, T_{\text{L*}})$, whose state set is $Q'$ and initial state is $q_0'$, and for every $s \in S_{\pre} \cup S_{\pre} \cdot A$, $\delta^*(q_0', s) = \row_s$. We now prove an analogous result for PL*.

    \begin{lemma}\label{lemma-angluin-lemma-2-pl*}
        Assume that $\obstable$ is a closed, consistent PL* observation table. Let $\row_s$ be a function mapping from $S_{\suff}$ to $\{0, 1\}$, where $\row_s(e) = T(se)$. For the DFA $M\obstable$ (with state set $Q$ and initial state $q_0$) and for every $s \in S_{\stored}$, $\delta^*(q_0, s) = \row_s$.
    \end{lemma}

    \begin{proof}
        Like Angluin, we prove this lemma by induction on the length of $s$ (denoted $|s|$).

        \begin{description}
            \item[Base case] $|s| = 0$

                By definition, $q_0 = \row_{\varepsilon}$, and clearly $\delta^*(q_0, \varepsilon) = q_0$. Thus, $\delta^*(q_0, \varepsilon) = \row_{\varepsilon}$, as required.

            \item[Inductive hypothesis] $|s| \leq k$

                We assume that, for all $s \in S_{\stored}$ such that $|s| \leq k$, $\delta^*(q_0, s) = \row_s$.

            \item[Inductive step] $|s| = k+1$

                Let $s \in S_{\stored}$, where $|s| = k+1$, be arbitrary. Hence, $s = s'a$, for $a \in A$ and some string $s'$ such that $|s'| = k$. We know that $s' \in S_{\pre}$, for either $s \in S_{\pre} \cdot A$, in which case $s'$ is in $S_{\pre}$, or $s \in S_{\pre}$, in which case $s' \in S_{\pre}$ as well since $S_{\pre}$ is prefix-closed. We also can deduce that $s' \in S_{\stored}$. Why is this? Assume the opposite i.e., that $s' \not\in S_{\stored}$. Based on how strings are added to $S_{\stored}$ (see Section~\ref{obs-table-update-s-pre-standard} for details), PL* would have considered adding $s'$ to $S_{\stored}$ before it considered adding $s$. PL* would not have added $s'$ to $S_{\stored}$ only if there was some $x \in S_{\stored} \cap S_{\pre}$ which was a dead prefix. Thus, $s$ would certainly have not been added to $S_{\stored}$ either -- and even if $s$ was itself later added to $S_{\pre}$, still it would not be added to $S_{\stored}$ due to the existence of this $x$. But, we originally defined $s$ as being an element of $S_{\stored}$! Contradiction -- and so we conclude that $s' \in S_{\stored}$. 

                Now, knowing that $s' \in S_{\stored} \cap S_{\pre}$, we have the following:

                \begin{align*}
                    \delta^*(q_0, s) &= \delta(\delta^*(q_0, s'), a)\\
                    &= \delta(\row_s', a), &\text{by the inductive hypothesis}\\
                    &= \row_{s'a}, &\text{by the definition of $\delta$}\\
                    &= \row_{s}
                \end{align*}
        \end{description}
    \end{proof}

    We now prove a second lemma.
    
    \begin{lemma}\label{lemma-some-bad-prefix-row}
        Let $\obstable$ be PL*'s final observation table when learning the non-dense language $L$. Let $M\obstable$ have state set $Q$, initial state $q_0$, final states $F$ and transition function $\delta$. Then, there must be some $s \in S_{\pre} \cap S_{\stored}$ which is a dead prefix.
    \end{lemma}

    \begin{proof}
        Assume the opposite i.e., there is \emph{no} $s \in S_{\pre} \cap S_{\stored}$ which is a dead prefix. Thus, for any dead prefix $s$ of $L$, $\row_s \not\in Q$. Let $x$ be some dead prefix of $L$ (since $L$ is non-dense, there must be at least one dead prefix of $L$). After $M\obstable$ consumes $x$, it must end up in some state $\row_{x'} \in Q$. Since $x'$ is not a dead prefix of $L$, and $\delta^*(q_0, x') = \row_{x'}$ (from Lemma~\ref{lemma-angluin-lemma-2-pl*}), the state language of $x'$ must be non-empty. Thus, \\
        $M\obstable$ will accept some extension of $x$ -- but this contradicts that $x$ is a dead prefix! So we conclude that there is some $s \in S_{\pre} \cap S_{\stored}$ which is a dead prefix, as required.
    \end{proof}
    
    Now, we turn our attention to proving Theorem~\ref{theorem-non-dense}. Let $s$ be an element of $S_{\pre} \cap S_{\stored}$ which is a dead prefix (we know from Lemma~\ref{lemma-some-bad-prefix-row} that such an $s$ must exist). Though PL* might've needed to perform a prefix query to identify that $s$ was a dead prefix, upon finding that it was, it wouldn't have needed to perform prefix queries for any $se$ such that $e \in S_{\suff} \setminus \{\varepsilon\}$ (in order to determine $T(se)$). However, L*, upon adding $s$ to $S_{\pre}$, would've needed to perform a membership query for \emph{every} $se$ such that $e \in S_{\suff}$. Furthermore, for each $a \in A$ and $e \in S_{\suff}$, PL* would not have needed to perform a prefix query for $sae$ (in order to determine $T(sae)$); L*  however would've needed to perform a membership query for every such $sae$. Hence, for $q$ the number of membership queries asked by L*, and $q'$ the number of prefix queries asked by PL*, we have that

    $$q - q' \geq |S_{\suff}| \times (1 + |A|) - 1$$
    
    as required.
\end{proof}

Generally, in practice, the improvement given by PL* in terms of the number of queries made will be greater than what is stated in Theorem~\ref{theorem-non-dense} (as we shall see in Section~\ref{empirical-evaluation-exact}). However, for any dense language, the improvement will be \emph{at least as much} as what is stated in Theorem~\ref{theorem-non-dense}.

Note that the improvement stated in Theorem~\ref{theorem-non-dense} is dependent upon both the alphabet size $|A|$ and $|S_{\suff}|$ i.e., it becomes more significant as $|A|$ and $|S_{\suff}|$ increase.

\subsection{PL* without an Equivalence Oracle}\label{random-sampling-oracle}

The PL* algorithm we have described so far relies on two oracles -- a prefix oracle and an equivalence oracle. Though in Section~\ref{pref-query} we argued that many parsers common in the field of software engineering are naturally capable of behaving as prefix oracles, we do not know of any that are capable of behaving as (exact) equivalence oracles. Because our aim was to design an algorithm that could be applied to the problem of modelling a parser, we now turn our attention to a more practical setting, where an equivalence oracle is unavailable. This change of setting does not prevent us from using PL*; rather, instead of using an equivalence oracle, PL* can now be run with a \emph{random sampling oracle} within the probably-approximately correct (PAC) framework. When a random sampling oracle is called, it returns some $x \in A^*$ drawn according to a fixed probability distribution $D$ over $A^*$. When $q_i$ calls are made to the sampling oracle in place of the $i^{\text{th}}$ equivalence query, where

$$q_i = \left\lceil{\frac{1}{\varepsilon} \left(\ln{\frac{1}{\delta}} + i \ln{2}\right)}\right\rceil$$

there is a theoretical guarantee on the classification accuracy of the hypothesis learned, relative to $D$ \citep{RN1}. Specifically, if $X' \subseteq A^*$ denotes the language of the inferred hypothesis and $X \subseteq A^*$ the target language, then $X'$ is guaranteed to satisfy the following:

$$\Pr(d(X', X) \leq \varepsilon) \geq 1 - \delta$$

where

$$d(L, L') = \sum_{x \in L \oplus L'} \Pr(x)$$

In words: the probability that the hypothesis returned by the learner is a ``good'' (quantified by $\varepsilon$) approximation of $X$ should be ``high'' (quantified by $\delta$).

(See Section~\ref{learning-models-active-pac} for a more detailed description of the PAC learning framework.)

From a theoretical perspective, the PAC guarantee is strong, powerful and tunable (due to being parameterised by $\varepsilon$ and $\delta$). However, a key limitation of this guarantee is that it is dependent upon not just $X'$ and $X$, but also the sampling distribution $D$. In practice, we generally do not have a particular distribution in mind that we would like our hypothesis to perform well against, and so we want an evaluation metric that is dependent only on $X'$ and $X$ (and thus is independent of any particular sampling distribution). One such metric is \emph{F1-score}, which has been used in several works to evaluate language learning algorithms \citep{RN51, RN52, RN76}. We now provide a definition of this metric. 

\subsubsection{F1-score}\label{random-sampling-oracle-f1-defn}

The metric of \emph{F1-score} is attributed to \citet{RN90}, who is viewed as a founder of the field of information retrieval, but nowadays it is widely-used in the field of machine learning. Specifically, it is used to evaluate the performance of a \emph{binary classifier}; that is, a machine learning algorithm which classifies input data into \emph{two} possible categories. In the context of language learning, the binary classifier we are evaluating is the hypothesis inferred by the language learning algorithm (which is, in the case of PL*, a DFA), which classifies strings in $A^*$ into two possible categories -- ``in the language'' and ``not in the language''. The hypothesis should classify strings according to their membership in the target language $X$; specifically, it should classify as ``in the language'' strings which are in $X$, and classify as ``not in the language'' strings which aren't in $X$. F1-score is a metric which evaluates the ability of the hypothesis to do exactly that.

F1-score combines two other metrics, \emph{precision} and \emph{recall}, defined as follows:

\begin{definition}[Precision and Recall]\label{definition-precision-recall}
    In the setting of language learning, given that the language accepted by the inferred hypothesis is $X' \subseteq A^*$ and the unknown target language is $X \subseteq A^*$, \textup{precision} and \textup{recall} are given by

    \begin{align*}
        \text{precision} &= \frac{|X' \cap X|}{|X'|} \\
        \text{recall} &= \frac{|X' \cap X|}{|X|}
    \end{align*}
\end{definition}

In words: \emph{precision} is the proportion of strings in $X'$ which are also in $X$, whilst \emph{recall} is the proportion of strings in $X$ which are also in $X'$ (note: if the inferred hypothesis accepts $\varnothing$, we set its precision to 1). High precision indicates that the inferred hypothesis is capable of rejecting strings not in $X$, whilst high recall indicates that the inferred hypothesis is capable of accepting strings which are in $X$. Figure~\ref{fig:prec-rec} is a visualisation of these two metrics.

\begin{figure*}[h]
\centering
    \includegraphics[width=\textwidth*3/4]{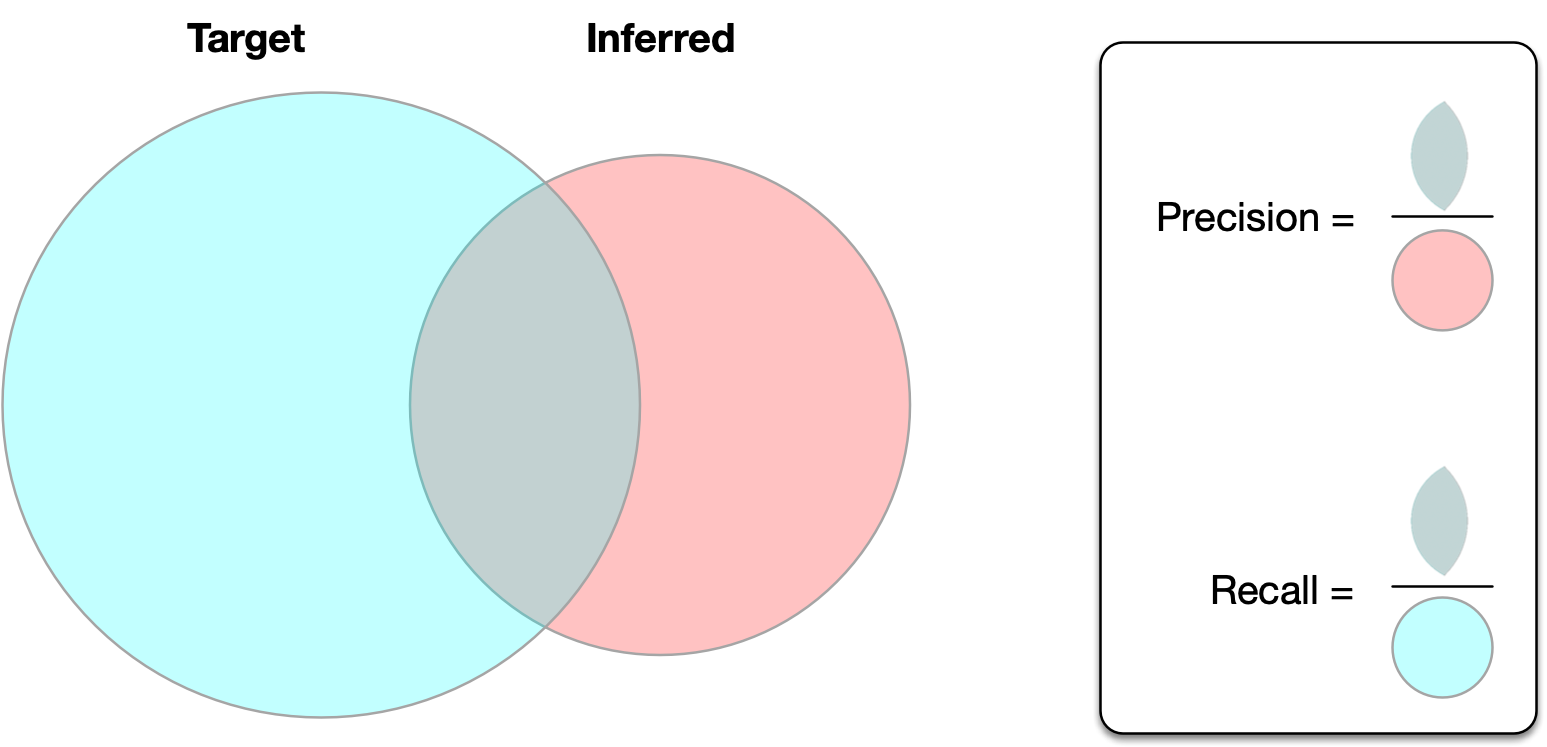}
\caption{Visualisation of precision and recall}\label{fig:prec-rec}
\end{figure*}

It is trivial to achieve high precision and low recall -- just return a hypothesis whose language is $\varnothing$. It is likewise trivial to achieve low precision and high recall -- just return a hypothesis whose language is $A^*$. Ideally, our inferred hypothesis achieves both high precision \emph{and} high recall -- and its ability to do so is captured by F1-score.

\begin{definition}[F1-score]\label{definition-f1}
    \textup{F1-score} is defined as the harmonic mean of precision and recall:

    \begin{align*}
        \text{F1-score} &= \frac{2 \times \text{precision} \times \text{recall}}{\text{precision} + \text{recall}}
    \end{align*}
\end{definition}

Compared to the other types of mean (namely arithmetic and geometric), harmonic mean gives greater weight to lower input values. Specifically, it tends to be closer in value to the \emph{minimum} of the input values than the arithmetic or geometric mean would be. For this reason, F1-score is high only if \emph{both} precision and recall are high; if either are low, then it also is low.

\subsubsection{Finding a Sampling Distribution that Optimises F1-score}\label{sampling-distribution-optimise-f1}

To summarise: we aim to run PL* with a random sampling oracle within the PAC framework; however, our end goal will not be to minimise $d(X', X)$, as is promised by the PAC guarantee, but will instead be to maximise F1-score. Though F1-score itself is distribution-independent, we still need to find a sampling distribution $D$ over $A^*$ if we are to run PL* within the PAC framework. In particular, we will need to find a sampling distribution which allows PL* to achieve ``high'' F1-scores, where, based on \citet{RN51} and \citet{RN52}, we consider a ``high'' F1-score as being greater than or equal to 0.9. What are some potential sampling distributions we could use? In the following sections, we will consider some possibilities.

\noindent\textbf{Pseudo-Uniform Distribution over $A^*$}\label{sampling-distribution-uniform}

As a baseline, we will consider the pseudo-uniform distribution given by the following procedure, based on the PAC equivalence oracle implemented by \citet{RN54}:

\begin{enumerate}
    \item Sample string length $l$ from a certain range $[l_{\min}, l_{\max}]$.
    \item Randomly-generate a string of length $l$ by repeatedly sampling characters from $A$ (assuming a uniform distribution over $A$).
\end{enumerate}

We refer to the distribution given by the above algorithm as a \emph{pseudo}-uniform distribution because a ``uniform distribution'' would imply that every string in $A^*$ has equal probability of being chosen. This is however not the distribution arising from this sampling algorithm. Firstly, the algorithm only generates strings whose length is in the range $[l_{\min}, l_{\max}]$, and so, for any string $x$ whose length is either less than $l_{\min}$ or greater than $l_{\max}$, the probability of it being generated is 0. Furthermore, this distribution is also not a true uniform distribution over all strings $x$ such that $l_{\min} \leq |x| \leq l_{\max}$. For $l_{\min} \leq n \leq l_{\max}$, the probability of generating a \emph{particular} $x \in A^n$ (where $A^n$ is the set of all strings over $A$ of length $n$) would be 

$$\frac{1}{(l_{\max} - l_{\min} + 1) \cdot |A^n|}$$

Let $l_{\min} \leq i < j \leq l_{\max}$, $x \in A^i$ and $y \in A^j$ be arbitrary. Since $i < j$, $|A^i| < |A^j|$, and so the probability of generating $x$ would be \emph{higher} than the probability of generating $y$.

Despite not being a true uniform distribution, we still may consider the distribution given by the above algorithm as \emph{pseudo}-uniform since, in Step (1), it chooses a string \emph{length} $l$ uniformly at random, and in Step (2), it chooses a string in $A^l$ uniformly at random.

A key drawback of this pseudo-uniform distribution is that, for many practical languages, the majority of strings drawn from it will be negative examples (i.e., strings \emph{not} in the target language). This is because most practical languages are \emph{sparse}, where, informally, a sparse language is a very small subset of $A^*$. Thus, if PL* is run with a random sampling oracle that draws from this distribution, it will most likely not learn what kinds of strings should be \emph{accepted}, and so, as we shall see in Section~\ref{empirical-evaluation-pac-f1}, will often return a hypothesis which has a very poor F1-score. 

\noindent\textbf{Distribution Given by a Prefix Query-Based Sampler}\label{sampling-distribution-pref-sampler}

An alternative to the approach described in Section~\ref{sampling-distribution-uniform} is Algorithm~\ref{pref-sampler}, adapted from the bFuzzer algorithm \citep{RN89}\footnote{The bFuzzer algorithm was originally designed to fuzz-test programs (i.e., to automatically generate test cases for a blackbox program) using prefix queries.}, which uses prefix queries to sample strings from $A^*$. Intuitively, the probability of generating a particular $x \in A^*$ is the same every time Algorithm~\ref{pref-sampler} is run. Thus, Algorithm~\ref{pref-sampler} does indeed give us a fixed sampling distribution over $A^*$, and so it can be used within the PAC framework.

\begin{algorithm}
\caption{Prefix query-based sampler}\label{pref-sampler}
\begin{flushleft}
    \hspace*{\algorithmicindent} \textbf{Input:} A finite alphabet $A$, probability of generating a positive example $p$, maximum number of attempts $m$, minimum length (for positive example) $l_{\min}$, maximum length $l_{\max}$, prefix oracle \\
    \hspace*{\algorithmicindent} \textbf{Output:} A string $x \in A^*$
\end{flushleft}
\begin{algorithmic}[1]
    \STATE $r \gets \varepsilon$ \algorithmiccomment{Current prefix}
    \STATE $S \gets \varnothing$ \algorithmiccomment{Set of symbols seen so far}
    \STATE $C \gets A$ \algorithmiccomment{Set of choices for the next symbol}
    \STATE $t \gets$ random([\xspace\positive, \negative], $w=[p, 1-p]$) \algorithmiccomment{Type of example to be generated}
    \STATE $a \gets 0$ \algorithmiccomment{Number of attempts made so far}
    \WHILE{$a < m$}
        \STATE $s \gets$ random($C$)
        \STATE $n \gets rs$
        \STATE $r \gets$ response returned by prefix oracle when given $n$ as input

        \IF{$r = \member$}
            \IF{$t = \positive$ and $|n| \geq l_{\min}$}
                \RETURN $n$
            \ELSE
                \STATE $c \gets$ random([ \extend, \restart]) \algorithmiccomment{Determine whether sampler will extend $n$}
                \IF{$c = \extend$}
                    \STATE $S \gets \varnothing$
                    \STATE $r \gets n$
                \ELSE
                    \STATE add $s$ to $S$
                \ENDIF
            \ENDIF
        \ELSIF{$r = \livepref$}
            \IF{$t = \positive$}
                \STATE $S \gets \varnothing$
                \STATE $r \gets n$
            \ELSE
                \RETURN $n$
            \ENDIF
        \ELSE
            \IF{$t = \positive$}
                \STATE add $s$ to $S$
            \ELSE
                \RETURN $n$
            \ENDIF
        \ENDIF

        \STATE $C \gets A \setminus S$

        \IF{$|n| > l_{\max}$}
            \STATE $a \gets a + 1$
            \STATE $r \gets \varepsilon$ \algorithmiccomment{Reset all variables}
            \STATE $S \gets \varnothing$
            \STATE $C \gets A$
            \STATE \algorithmicbreak
        \ENDIF
    \ENDWHILE
    \RETURN $\varnothing$ \algorithmiccomment{Failed to generate an example within $m$ attempts}
\end{algorithmic}
\end{algorithm}

We now describe Algorithm~\ref{pref-sampler}. It is parameterised by $l_{\max}$ (as well as other parameters), and produces a distribution over all strings $x$ such that $|x| \leq l_{\max}$ (i.e., a word whose length is greater than $l_{\max}$ has probability zero). We first assign $A$ to the set of choices $C$ (line 3). We then set $t$ to either \positive (with probability $p$) or \negative (with probability $1 - p$) (line 4). If $t$ is set to \positive, then Algorithm~\ref{pref-sampler} will attempt to generate a positive example (i.e., a string in the target language $X \subseteq A^*$); otherwise, it will attempt to generate a negative example (i.e., a string not in $X$). We then enter the main loop (lines 7-43), which is executed until we have exhausted our permitted number of attempts ($m$). During each iteration of the main loop, we do the following: we randomly-select a symbol $s$ from our set of choices $C$ (line 7), and then produce a new prefix $n$ (line 8) which is passed to the prefix oracle (line 9). Based on the value returned by the prefix oracle, we proceed as follows:

\begin{description}
    \item[\member] 
        If we'd chosen to generate a positive example, and furthermore $n$ meets the minimum length requirement, we return $n$ (line 12). If either condition is not met (i.e., we'd chosen to generate a negative example, or $n$ does not meet the minimum length requirement, or both), we then randomly-choose whether or not to extend $n$ (line 14). If we had chosen to extend $n$, we clear all the symbols that have been seen so far (line 16), and update the prefix $r$ to the value of the new prefix $n$ (line 17). Otherwise, if we had chosen \emph{not} to extend $n$, we note the symbol $s$ as having been seen, and continue (line 19).
    \item[\livepref]
        If we'd chosen to generate a positive example, then we should continue extending $n$. We clear all the symbols that have been seen so far (line 24), and and update the prefix $r$ by setting it to $n$ (line 25). Otherwise, if we'd chosen to generate a negative example, then we simply return $n$ (line 27).
    \item[\deadpref]
        If we'd chosen to generate a positive example, then, by definition of a dead prefix, we should not extend $n$. Hence, we record that the symbol $s$ had been seen, and continue (line 31). Otherwise, if we'd chosen to generate a negative example, then we simply return $n$ (line 33).
\end{description}

After performing these steps, we update our set of choices $C$ to exclude the last seen symbol (line 36). The next symbol $s$ will then be chosen from this updated set of choices (line 7). 

Unlike the pseudo-uniform distribution described in Section~\ref{sampling-distribution-uniform}, we hypothesise that the sampling distribution given by Algorithm~\ref{pref-sampler} will allow PL* to achieve high F1-scores, due to the algorithm's ability to sample both positive and negative examples (and furthermore control the probability of generating a positive example). In Section~\ref{empirical-evaluation-pac-f1}, we shall see whether this is indeed the case.

\section{Empirical Evaluation} \label{chap:empirical-evaluation}

In this section, we empirically evaluate our novel learning algorithm, PL*. We begin by introducing the regular target languages on which we will evaluate PL* as well as L*, the classical learning algorithm on which PL* is based (Section~\ref{empirical-evaluation-langs}). We then compare PL* and L* on a range of metrics when both are run with exact equivalence oracles (Section~\ref{empirical-evaluation-exact}). Finally, we evaluate PL* when run with a random sampling oracle, investigating how its parameters can be set to enable it to achieve high F1-scores (Section~\ref{empirical-evaluation-pac}).

\subsection{Regular Languages used for Evaluation}\label{empirical-evaluation-langs}

For all the experiments performed in this section, $A$ (the alphabet) is a (fixed) 256-symbol subset of the Unicode character set.

We have performed our empirical evaluation on four regular languages (over $A$) of practical interest, specifically the languages of

\begin{enumerate}
    \item restricted-nesting arithmetic expressions (\arith),
    \item restricted-nesting JSON strings (\json),
    \item restricted-nesting Dyck(2) words (\dyck), and
    \item dates in DD/MM/YY format (\dat).
\end{enumerate}

Below, we provide definitions of these four languages. Note that, for \arith, \json and \dyck, we present context-free grammars (CFGs); however, evaluation was performed on \emph{regular} equivalents of these context-free languages i.e., with a maximum nesting depth imposed. (The regular expressions representing the regular equivalents would be too long to include here.) \dat however is naturally a regular language, and so is represented by a regular expression below.

\begin{description}
    \item[\arith]

        \begin{align*}
            S &\rightarrow E \\
            E &\rightarrow F * E \mid F / E \mid F \\
            F &\rightarrow T + F \mid T - F \mid T \\
            T &\rightarrow (E) \mid D \\
            D &\rightarrow 0 \mid 1 \mid 2 \mid 3 \mid 4 \mid 5 \mid 6 \mid 7 \mid 8 \mid 9
        \end{align*}

    \item[\json]

        \begin{align*}
            S &\rightarrow E \\
            E &\rightarrow \{ C \} \\
            C &\rightarrow P, C \mid P \\
            P &\rightarrow T : V \\
            V &\rightarrow E \mid T \mid N \\
            T &\rightarrow \text{``$H$''} \mid \text{``''} \\
            H &\rightarrow \text{any character in the alphabet} \\
            N &\rightarrow 0 \mid 1 \mid 2 \mid 3 \mid 4 \mid 5 \mid 6 \mid 7 \mid 8 \mid 9
        \end{align*}

    \item[\dyck]

        \begin{align*}
            S &\rightarrow (S) S \mid [S] S \mid \varepsilon
        \end{align*}

    \item[\dat]
        \begin{align*}
        \end{align*}
        \vspace{-1.5cm}
        \begin{lstlisting}[xleftmargin=0.05\textwidth, xrightmargin=0.05\textwidth]
^(0[1-9]|[1,2][0-9]|3[0,1])\/(01|03|05|07|08|10|12)\/([0-9][0-9])$|^((0[1-9]|[1,2][0-9]|30)\/(04|06|09|11)\/([0-9][0-9]))$|^((0[1-9]|[1,2][0-9])\/(02)\/([0-9][0-9]))$
        \end{lstlisting}
    
\end{description}

\subsection{Comparing PL* and L* with an Exact Equivalence Oracle}\label{empirical-evaluation-exact}

We first compare the performance of L*, standard PL* and optimised PL* (as described in ~\Cref{chap:method}) when each is run with an \emph{exact} equivalence oracle; i.e., an oracle which, given a hypothesis $\mathcal{A'}$, responds with \y if $L(\mathcal{A'}) = X$ (for $X \subseteq A^*$ the target language), and \n otherwise. Specifically, L* will be run with an exact equivalence oracle and a membership oracle, and standard PL* and optimised PL* will each be run with an exact equivalence oracle and a \emph{prefix} oracle.

We will compare the performance of L*, standard PL* and optimised PL*, when each is run with an \emph{exact} equivalence oracle, in terms of the following four metrics:

\begin{enumerate}
    \item Number of cell comparisons made during closedness and consistency checks (where ``cell'' refers to a cell in an observation table, and a ``cell comparison'' is a comparison between any two cells in the table).
    \item Number of membership queries (MQs)/prefix queries (PQs) made (MQs in the case of L*, PQs in the case of standard/optimised PL*).
    \item Number of equivalence queries (EQs) made.
    \item Size of the learner's final observation table.
\end{enumerate}

For all these metrics, smaller values are more desirable than larger values. Smaller values of metrics (1), (2) and (3) are typically associated with lower runtime. Smaller values of metric (4) typically correspond to lower runtime \emph{and} space usage.

Figure~\ref{fig:exact-plots} compares the performance of L*, standard PL* and optimised PL* when each is run with an exact equivalence oracle and tasked with learning each of the following four languages (as described in Section~\ref{empirical-evaluation-langs}):

\begin{enumerate}
    \item \arith with maximum nesting depth 3,
    \item \json with maximum nesting depth 2,
    \item \dyck with maximum nesting depth 4, and
    \item \dat.
\end{enumerate}

\begin{figure*}[ht]
    \centering
    \begin{subfigure}[t]{0.45\textwidth}
        \centering
        \includegraphics[width=\textwidth]{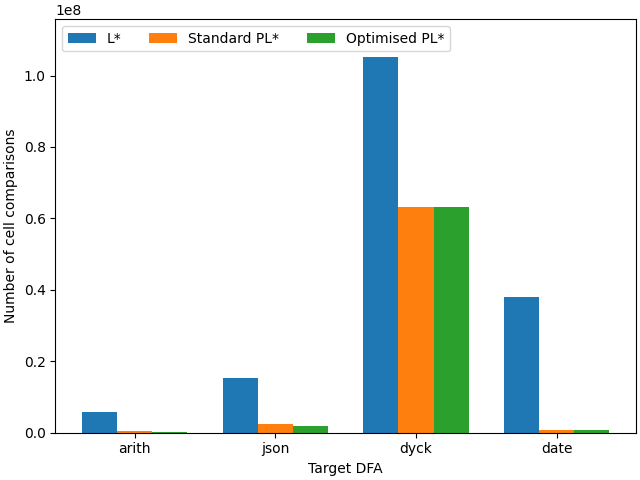}
        \caption{Number of cell comparisons}
        \label{fig:exact-cell-comparisons}
    \end{subfigure}
    \hspace{0.05\textwidth}
    \begin{subfigure}[t]{0.45\textwidth}
        \centering
        \includegraphics[width=\textwidth]{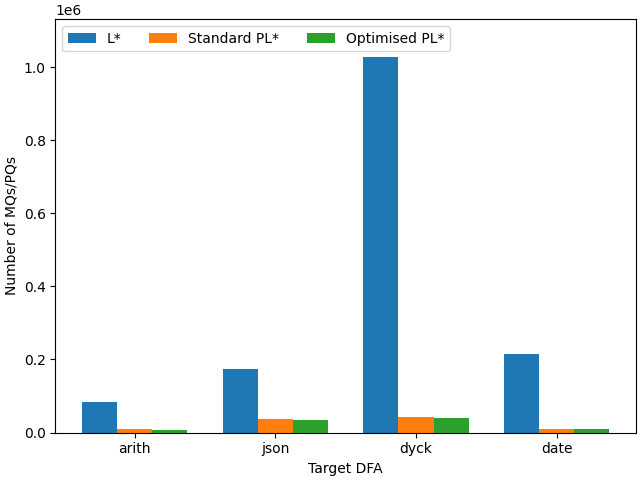}
        \caption{Number of MQs/PQs}
        \label{fig:exact-mqs-pqs}
    \end{subfigure}
    
    \vspace{0.5em}
    
    \begin{subfigure}[t]{0.45\textwidth}
        \centering
        \includegraphics[width=\textwidth]{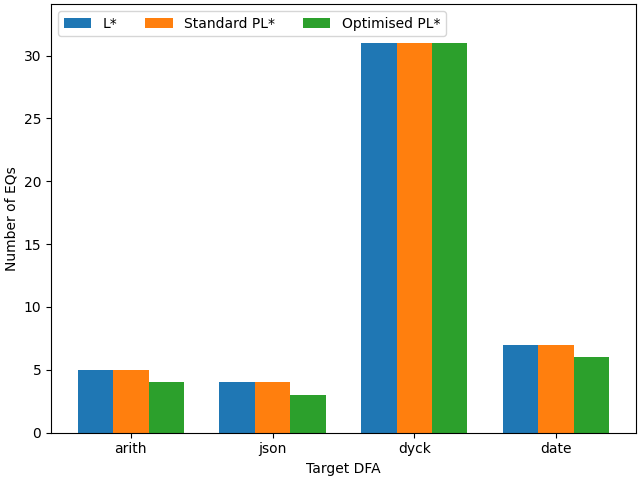}
        \caption{Number of EQs}
        \label{fig:exact-eqs}
    \end{subfigure}
    \hspace{0.05\textwidth}
    \begin{subfigure}[t]{0.45\textwidth}
        \centering
        \includegraphics[width=\textwidth]{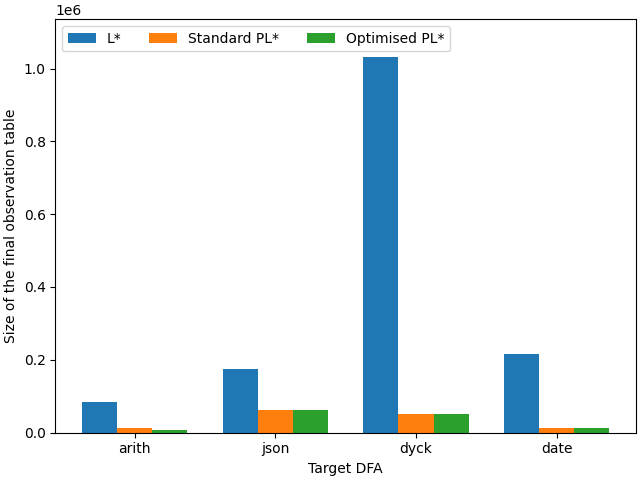}
        \caption{Size of the final observation table}
        \label{fig:exact-table-cells}
    \end{subfigure}
    
    \caption{Comparing L*, standard PL*, and optimised PL* when each is run with an exact equivalence oracle}
    \label{fig:exact-plots}
\end{figure*}

Figure~\ref{fig:exact-cell-comparisons} compares the three learning algorithms in terms of the number of cell comparisons performed. Across all four target languages, standard PL* and optimised PL* have similar performance, though in some cases optimised PL* performs slightly better. Both versions of PL* significantly outperform L* on all four languages. 

Figure~\ref{fig:exact-mqs-pqs} compares the three learners in terms of the number of membership queries (MQs) or prefix queries (PQs) performed (MQs in the case of L*, PQs in the case of standard/optimised PL*). Once again, standard PL* and optimised PL* have similar performance, with optimised PL* occasionally having slightly better performance. Again, both versions of PL* significantly outperform L*. Notably, for the \dyck language, the improvement given by PL* in terms of number of MQs/PQs performed is more significant than the improvement in number of cell comparisons (Figure~\ref{fig:exact-cell-comparisons}). This implies that, whilst PL* is able to significantly reduce the number of queries performed, it is perhaps less effective in speeding up closedness and consistency checks.

In Figure~\ref{fig:exact-eqs}, L*, standard PL* and optimised PL* are compared in terms of the number of equivalence queries (EQs) performed. As expected, L* and standard PL* have identical performance in terms of this metric -- from Theorem~\ref{theorem-termination}, we know that L* and standard PL* have equivalent observation tables, and thus construct identical hypotheses, at the end of each learning round. Hence, they both terminate after the same number of rounds (after outputting the correct hypothesis). As every round ends with a single equivalence query, L* and standard PL* thus perform exactly the same number of equivalence queries during their execution (when tasked with learning the same target language). Optimised PL* however occasionally performs less equivalence queries than the other two learners.

Figure~\ref{fig:exact-table-cells} compares the three learning algorithms in terms of observation table size. It is similar in appearance to Figure~\ref{fig:exact-mqs-pqs} i.e., standard and optimised PL* have similar performance, but both significantly outperform L*.

Overall, these results demonstrate that modifying a language learning algorithm to exploit the additional information given by prefix queries over membership queries can indeed lead to significant improvements both in terms of runtime and space usage.

\subsection{Evaluating PL* without an Equivalence Oracle}\label{empirical-evaluation-pac}

We now turn our attention to the more practical setting described in Section~\ref{random-sampling-oracle}; specifically, one in which only membership and prefix oracles are available (and an exact equivalence oracle is not). In this setting, we evaluate only optimised PL*, since in Section~\ref{empirical-evaluation-exact} it was revealed to be the best-performing learner across a range of metrics. Recall from Section~\ref{random-sampling-oracle} that, in the absence of an exact equivalence oracle, we instead run PL* with a prefix oracle and a \emph{random sampling oracle} within the probably approximately correct (PAC) learning framework. When PL* is run with a random sampling oracle that samples from $A^*$ according to distribution $D$, using the approach proposed by \citet{RN1}, there is a theoretical guarantee on the accuracy of the hypothesis learned. Specifically, if $X' \subseteq A$ denotes the language accepted by the inferred hypothesis and $X \subseteq A^*$ the target language, then $X'$ is guaranteed to satisfy the following:

$$\Pr(d(X', X) \leq \varepsilon) \geq 1 - \delta$$

where

$$d(L, L') = \sum_{x \in L \oplus L'} \Pr(x)$$

quantifies the ``difference'' between two languages $L$ and $L'$ over $A$.

In words: the probability that the hypothesis returned by the learner is a ``good'' (quantified by $\varepsilon$) approximation of $X$ should be ``high'' (quantified by $\delta$).

In Section~\ref{empirical-evaluation-pac-verif-guarantee}, we verify that the PAC guarantee stated above indeed holds in practice. Then, in Section~\ref{empirical-evaluation-pac-f1}, we proceed to investigate how the sampling distribution $D$, as well as the PAC parameters $\varepsilon$ and $\delta$, can be set to allow PL* to achieve high F1-scores (Definition~\ref{definition-f1}), which is the metric we are ultimately seeking to maximise.

\subsubsection{Empirically Verifying the PAC Guarantee}\label{empirical-evaluation-pac-verif-guarantee}

Let us now investigate whether the theoretical PAC guarantee indeed holds in practice. For each of the sampling distributions $D$ described in Section~\ref{random-sampling-oracle} (namely the pseudo-uniform distribution and the distribution given by the prefix query-based sampler), we run PL* with a prefix oracle and a random sampling oracle that draws from $D$, across all four target languages (namely \arith with maximum nesting depth 3, \json with maximum nesting depth 2, \dyck with maximum nesting depth 4, and \dat). For all target languages, we set the parameters $p$ and $m$ of the prefix query-based sampler to 0.5 and 200, respectively. However, we set $l_{\min}$ and $l_{\max}$ differently for each target language:

\begin{itemize}
    \item for \arith, we set $l_{\min} = 100, l_{\max} = 200$,
    \item for \json, we set $l_{\min} = 25, l_{\max} = 40$,
    \item for \dyck, we set $l_{\min} = 10, l_{\max} = 20$, and
    \item for \dat, we set $l_{\min} = 8, l_{\max} = 40$.
\end{itemize}

(For each target, we select values of $l_{\min}$ and $l_{\max}$ that allow reasonably complex strings to be sampled, for example strings that include some parentheses nesting, various operators, and so forth.)

For each target-distribution combination, we run PL* 100 times and then compute $d(X', X)$ each time (by sampling 1000 strings from $D$). We then use these 100 $d(X', X)$ values to compute $\Pr(d(X', X) \leq \varepsilon)$, as shown in Table~\ref{tab:empirical-evaluation-pac-verif-guarantee}. With $\Pr(d(X', X) \leq \varepsilon)$ consistently being 1.0, and $(1 - \delta)$ only being 0.95, clearly $\Pr(d(X', X) \leq \varepsilon) \geq 1 - \delta$ holds true in all cases.

\begin{table*}[ht]
    \centering
    \begin{tabularx}{\textwidth}{|X|X|r|r|r|}
\hline
        \textbf{Target Language $X$} & \textbf{Sampling Distribution} & \textbf{$\varepsilon$} & \textbf{$\delta$} & \textbf{$\Pr(d(X', X) \leq \varepsilon)$} \\
        \hline
        \multirow{2}{*}{\arith} & \multicolumn{1}{l|}{Pseudo-Uniform} & \multicolumn{1}{r|}{0.05} & \multicolumn{1}{r|}{0.05} & \multicolumn{1}{r|}{1.0} \\\cline{2-5}
                                & \multicolumn{1}{l|}{Prefix Query-Based} & \multicolumn{1}{r|}{0.05} & \multicolumn{1}{r|}{0.05} & \multicolumn{1}{r|}{1.0} \\\cline{2-5}
        \hline
        \multirow{2}{*}{\json} & \multicolumn{1}{l|}{Pseudo-Uniform} & \multicolumn{1}{r|}{0.05} & \multicolumn{1}{r|}{0.05} & \multicolumn{1}{r|}{1.0} \\\cline{2-5}
                                & \multicolumn{1}{l|}{Prefix Query-Based} & \multicolumn{1}{r|}{0.05} & \multicolumn{1}{r|}{0.05} & \multicolumn{1}{r|}{1.0} \\\cline{2-5}
        \hline
        \multirow{2}{*}{\dyck} & \multicolumn{1}{l|}{Pseudo-Uniform} & \multicolumn{1}{r|}{0.05} & \multicolumn{1}{r|}{0.05} & \multicolumn{1}{r|}{1.0} \\\cline{2-5}
                                & \multicolumn{1}{l|}{Prefix Query-Based} & \multicolumn{1}{r|}{0.05} & \multicolumn{1}{r|}{0.05} & \multicolumn{1}{r|}{1.0} \\\cline{2-5}
        \hline
        \multirow{2}{*}{\dat} & \multicolumn{1}{l|}{Pseudo-Uniform} & \multicolumn{1}{r|}{0.05} & \multicolumn{1}{r|}{0.05} & \multicolumn{1}{r|}{1.0} \\\cline{2-5}
                                & \multicolumn{1}{l|}{Prefix Query-Based} & \multicolumn{1}{r|}{0.05} & \multicolumn{1}{r|}{0.05} & \multicolumn{1}{r|}{1.0} \\\cline{2-5}
        \hline
    \end{tabularx}
    \caption{Empirically verifying the PAC guarantee}
    \label{tab:empirical-evaluation-pac-verif-guarantee}
\end{table*}

\subsubsection{Optimising F1-score}\label{empirical-evaluation-pac-f1}

Recall the metric of \emph{F1-score} (Definition~\ref{definition-f1}). Intuitively, F1-score quantifies how well the language $X' \subseteq A^*$ of the inferred hypothesis ``matches'' the target language $X \subseteq A^*$. Specifically, F1-score is defined as the harmonic mean of two other metrics, \emph{precision} and \emph{recall}:

\begin{align*}
    \text{precision} &= \frac{|X' \cap X|}{|X'|} \\
    \text{recall} &= \frac{|X' \cap X|}{|X|} \\
    \text{F1-score} &= \frac{2 \times \text{precision} \times \text{recall}}{\text{precision} + \text{recall}}
\end{align*}

High precision indicates that the inferred hypothesis is capable of \emph{rejecting} strings \emph{not} in $X$, whilst high recall indicates that the inferred hypothesis is capable of \emph{accepting} strings which \emph{are} in $X$.

In Section~\ref{random-sampling-oracle}, we explained why, despite the strength of the PAC guarantee, in practice F1-score is a more appropriate metric to optimise. In order to bridge the mismatch between the metric optimised by PAC (namely $d(X', X)$) and F1-score, we now investigate whether there is perhaps a relationship between the parameters of the PAC framework (namely $\varepsilon$ and $\delta$) and the F1-score of the hypothesis returned by PL*. Specifically, for each of our four target languages (\arith, \json, \dyck and \dat), we determine the relationship between $\varepsilon$, $\delta$ and F1-score when PL* is run with four kinds of random sampling oracles:

\begin{enumerate}
    \item the pseudo-uniform random sampler described in Section~\ref{sampling-distribution-uniform} (our baseline approach),
    \item the prefix query-based sampler described in Section~\ref{sampling-distribution-pref-sampler}, with $p$ set to 0.05,
    \item the prefix query-based sampler with $p$ set to 0.5, and
    \item the prefix query-based sampler with $p$ set to the maximum value of 1.
\end{enumerate}

Note that, compared to the previous experiments, here we increase the maximum nesting depth of the target languages \arith, \json and \dyck to mimic learning the corresponding context-free grammar (which would allow \emph{unrestricted} nesting):

\begin{itemize}
    \item for \arith, we learn the regular equivalent enforcing a maximum nesting depth of 100,
    \item for \json, we learn the regular equivalent enforcing a maximum nesting depth of 100, and
    \item for \dyck, we learn the regular equivalent enforcing a maximum nesting depth of 5.
\end{itemize}

Note also that, whilst the parameter $m$ of the prefix query-based sampler is set to 200 for all experiments, the parameters $l_{\min}$ and $l_{\max}$ are set differently based on the language being learned (as in Section~\ref{empirical-evaluation-pac-verif-guarantee}):

\begin{itemize}
    \item for \arith, we set $l_{\min} = 100, l_{\max} = 200$,
    \item for \json, we set $l_{\min} = 25, l_{\max} = 40$,
    \item for \dyck, we set $l_{\min} = 10, l_{\max} = 20$, and
    \item for \dat, we set $l_{\min} = 8, l_{\max} = 40$.
\end{itemize}

Furthermore, note that we time out each execution of PL* after 2 minutes.

Let us now explain how we estimate precision and recall for a given target-sampler-$\varepsilon$-$\delta$ combination. Let $X \subseteq A$ denote the target language, and $X' \subseteq A^*$ the language of the hypothesis inferred by PL*. Let $M_X$ denote our handwritten DFA recognising $X$ (for each target language, we handwrote a DFA recognising it). Let $M_{X'}$ denote the DFA returned by PL* (which recognises $X'$). We first take $M_X$ and $M_{X'}$ and convert them into regular grammars $G_X$ and $G_{X'}$ respectively, using the procedure described in Section~\ref{grammars}. Then, to estimate precision, we sample (uniformly at random, using the procedure described by \citet{RN78}) 1000 strings from $G_{X'}$, within the length range $[l_{\min}, l_{\max}]$, and determine the proportion of strings in this sample that are also in $X$. Similarly, to estimate recall, we sample (again uniformly at random) 1000 strings from $G_{X}$, within the length range $[l_{\min}, l_{\max}]$, and determine the proportion of strings in this sample that are also in $X'$. F1-score is then the harmonic mean of these two values. For every such target-sampler-$\varepsilon$-$\delta$ combination, we run PL* 5 times, and compute an average of the F1-scores produced.

Figures~\ref{fig:empirical-evaluation-pac-f1-arith}, \ref{fig:empirical-evaluation-pac-f1-json}, \ref{fig:empirical-evaluation-pac-f1-dyck} and \ref{fig:empirical-evaluation-pac-f1-date} illustrate the relationship between $\varepsilon$, $\delta$ and F1-score when PL* is tasked with learning each of our four target languages and run with each of our four random sampling oracles. (In each of these graphs, the variable on the horizontal axis is $\varepsilon$, the variable on the vertical axis is $\delta$, and colour indicates F1-score.) For each target-sampler combination, we also identify the $\varepsilon$-$\delta$ combination that produces the highest F1-score.

In Figure~\ref{fig:empirical-evaluation-pac-f1-arith}, we see that running PL* with the pseudo-uniform random sampler always yields F1-scores of 0. As described in Section~\ref{sampling-distribution-uniform}, this is as expected because a pseudo-uniform random sampler yields only very few positive examples, and so the language of the inferred hypothesis is generally a very small subset of the target language. When PL* is run with the prefix query-based sampler with $p = 0.05$, F1-score is high only when $\varepsilon$ is very small. When $p = 0.5$, F1-score is fairly high for most $\varepsilon$-$\delta$ combinations, only dropping when $\varepsilon$ and $\delta$ are both large. Finally, when $p = 1$, F1-score is high for all $\varepsilon$-$\delta$ combinations, though it tends to be at its highest when $\varepsilon$ and $\delta$ are both small. Furthermore, the highest F1-scores achieved when the prefix query-based sampler is used are all very close to or equal to 1.

\begin{figure*}[h]
\includegraphics[width=\columnwidth]{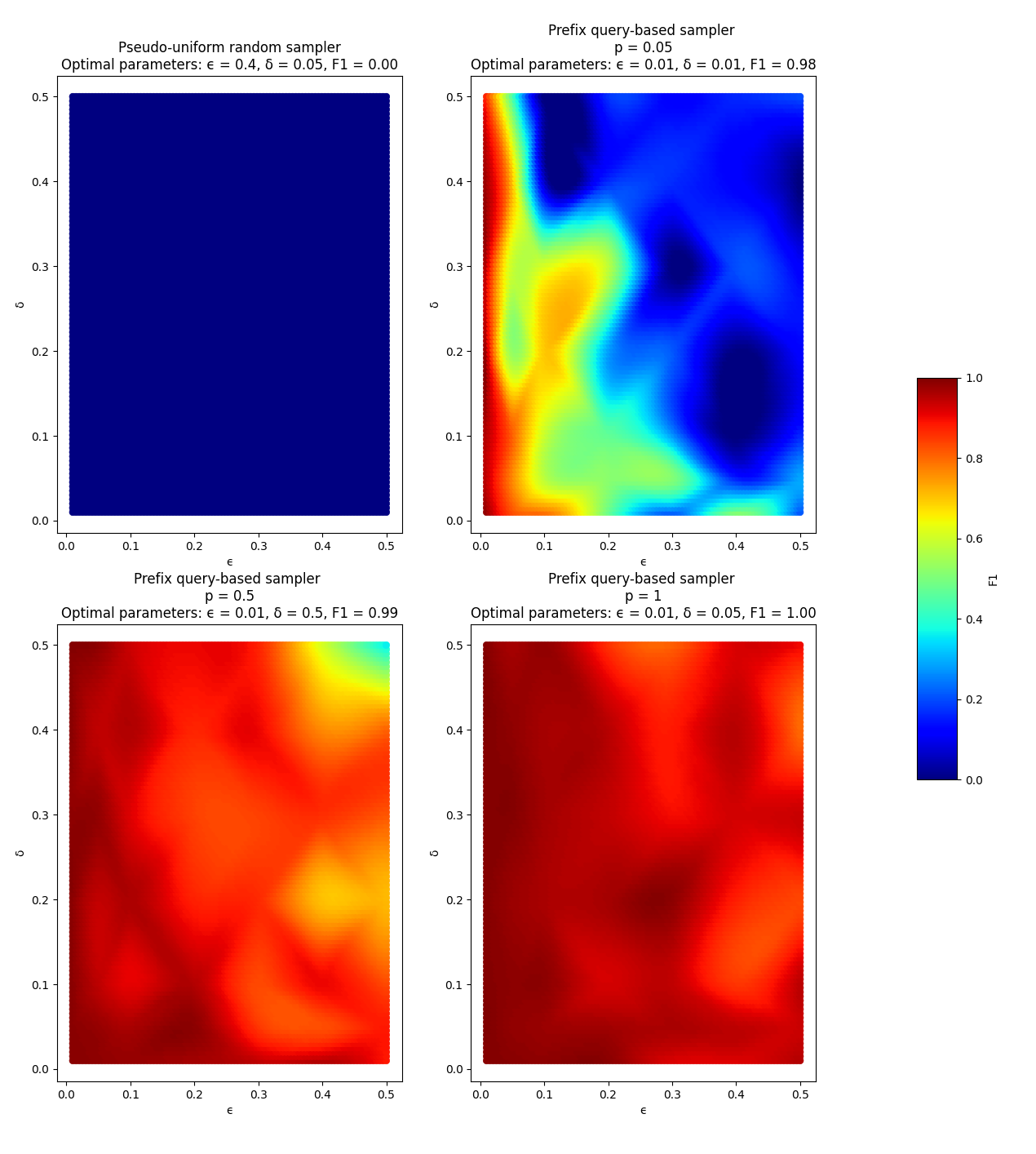}
\caption{Relationship between $p$, $\varepsilon$, $\delta$ and F1 for \arith with maximum nesting depth 100}\label{fig:empirical-evaluation-pac-f1-arith}
\end{figure*}

In Figure~\ref{fig:empirical-evaluation-pac-f1-json}, we see that running PL* with the pseudo-uniform random sampler once again consistently yields F1-scores of 0. When PL* is run with the prefix query-based sampler with $p = 0.05$, F1-score is once again high only when $\varepsilon$ is very small. When $p = 0.5$, F1-score is \emph{consistently} high when both $\varepsilon$ and $\delta$ are small; elsewhere, some variability is evident. Finally, when $p = 1$, F1-score is high for most $\varepsilon$-$\delta$ combinations, with slight drops when either (or both) of $\varepsilon$ or $\delta$ are large. Furthermore, the highest F1-scores achieved when the prefix query-based sampler is used are once again all very close to or equal to 1.

\begin{figure*}[h]
\includegraphics[width=\textwidth]{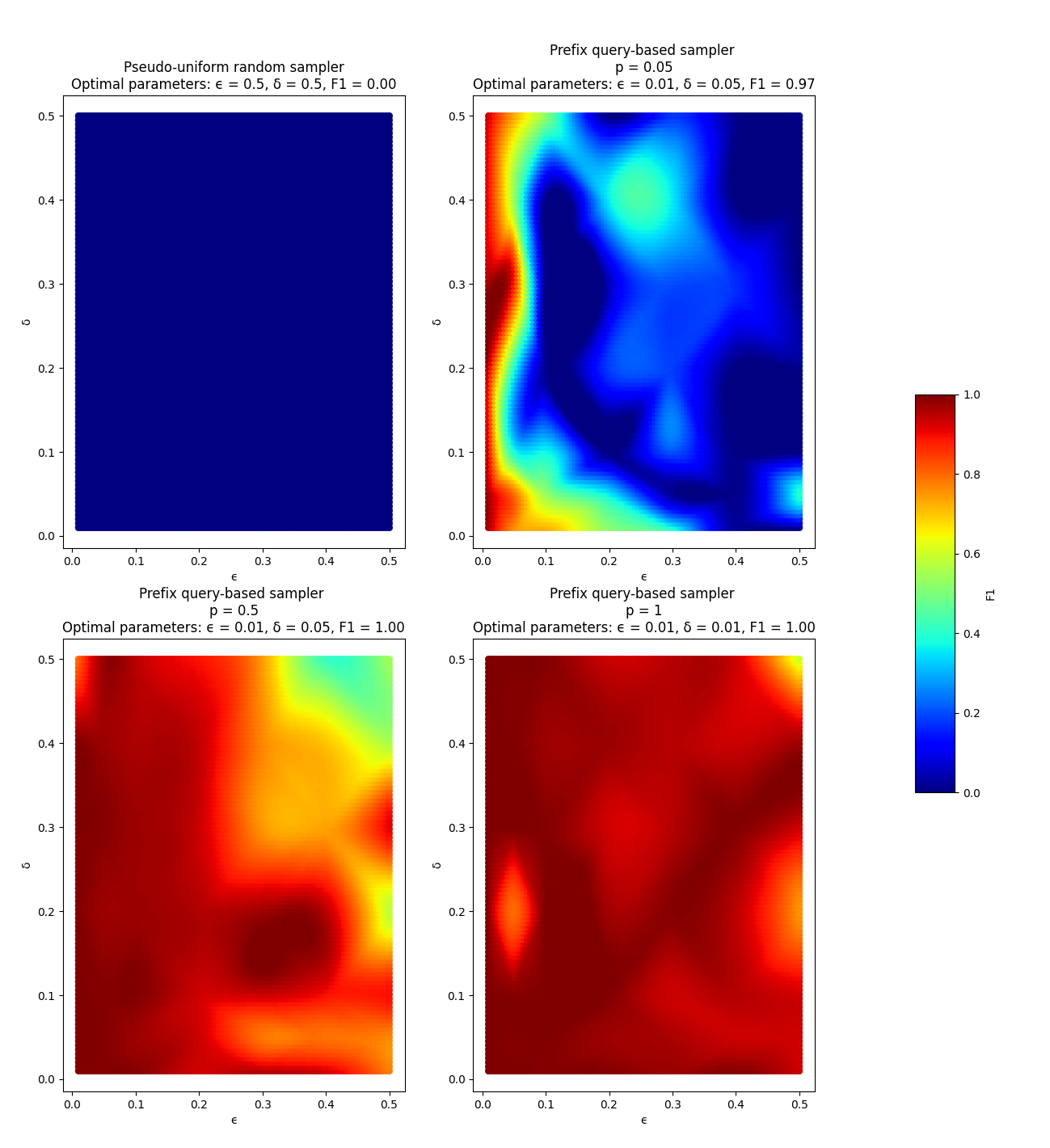}
\caption{Relationship between $p$, $\varepsilon$, $\delta$ and F1 for \json with maximum nesting depth 100}\label{fig:empirical-evaluation-pac-f1-json}
\end{figure*}

In Figure~\ref{fig:empirical-evaluation-pac-f1-dyck}, we see that running PL* with the pseudo-uniform random sampler again yields F1-scores of 0 for all $\varepsilon$-$\delta$ combinations. When PL* is run with the prefix query-based sampler with $p$ set to $0.05$, F1-score is once again high only when $\varepsilon$ is very small. When $p = 0.5$, F1-score is high for most $\varepsilon$-$\delta$ combinations, dropping only when $\varepsilon$ is large. Finally, when $p = 1$, F1-score is consistently 1, the highest-possible value. Furthermore, the highest F1-scores achieved when the prefix query-based sampler is used are yet again all very close to or equal to 1.

\begin{figure*}[h]
\includegraphics[width=\textwidth]{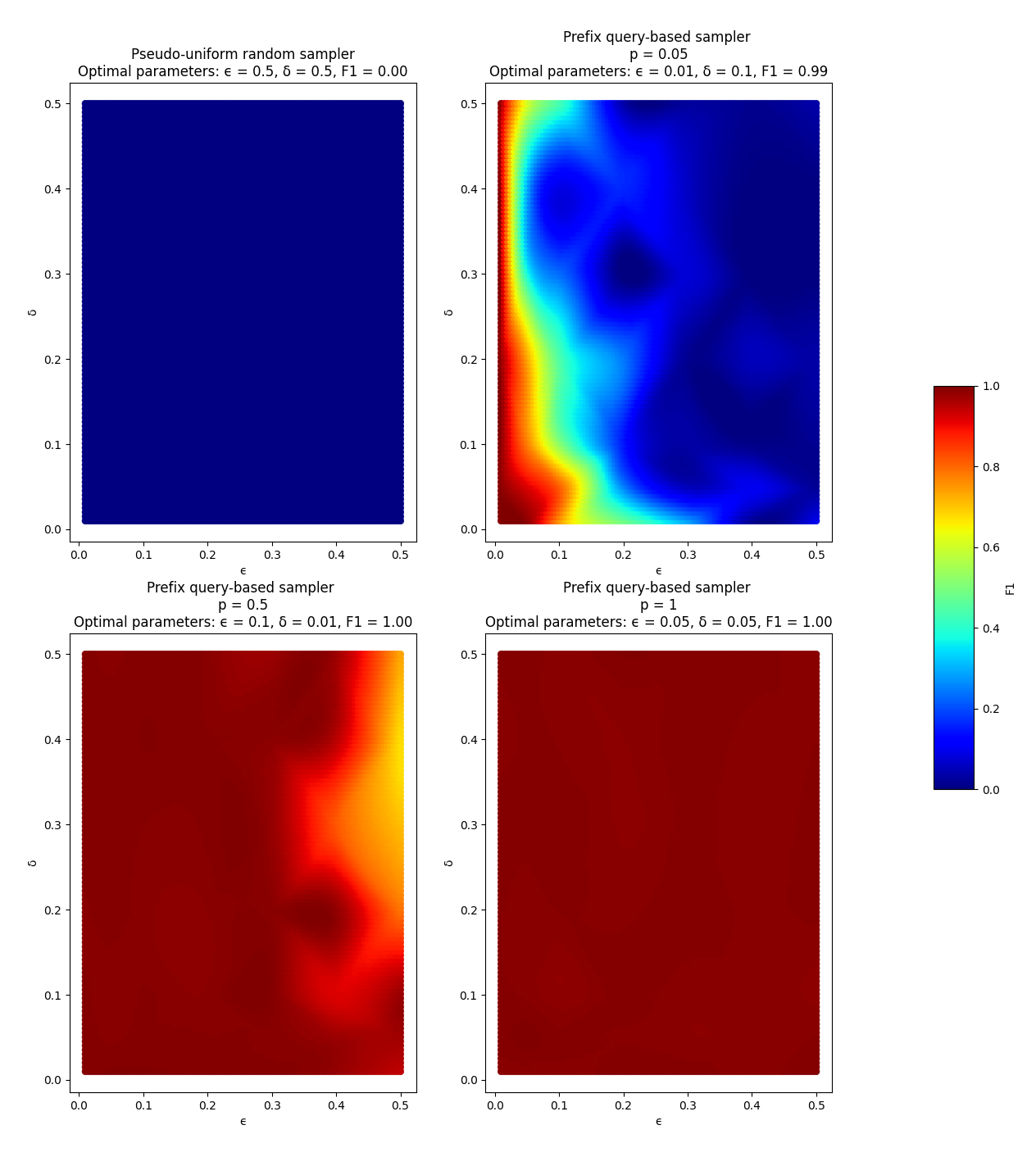}
\caption{Relationship between $p$, $\varepsilon$, $\delta$ and F1 for \dyck with maximum nesting depth 5}\label{fig:empirical-evaluation-pac-f1-dyck}
\end{figure*}

In Figure~\ref{fig:empirical-evaluation-pac-f1-date}, we once again see that running PL* with the pseudo-uniform random sampler consistently yields F1-scores of 0. When PL* is run with the prefix query-based sampler with $p = 0.05$, F1-score is once again high only when $\varepsilon$ is very small. When $p = 0.5$, F1-score is high when both $\varepsilon$ and $\delta$ are small, though surprisingly there is a slight drop as $\varepsilon$ and $\delta$ both approach 0. This is very likely a consequence of the variability inherent in the learning process, which, like on this occasion, may sometimes lead to unexpected results. Finally, when $p = 1$, F1-score is always fairly high, though there is some variability throughout. Furthermore, the highest F1-scores achieved when the prefix query-based sampler is used are all very close to 1.

\begin{figure*}[h]
\includegraphics[width=\textwidth]{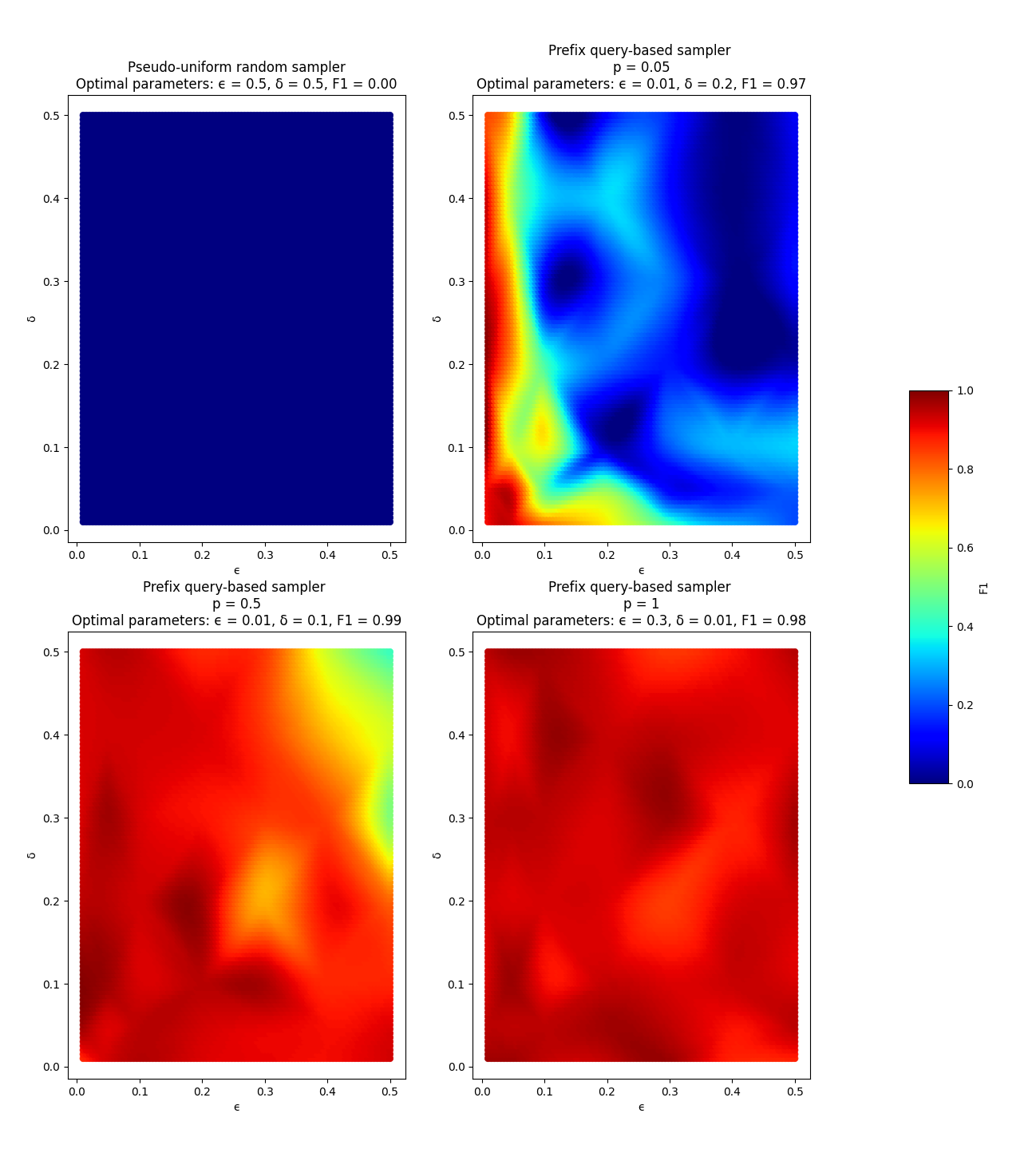}
\caption{Relationship between $p$, $\varepsilon$, $\delta$ and F1 for \dat}\label{fig:empirical-evaluation-pac-f1-date}
\end{figure*}

We now describe some general trends evident in Figures~\ref{fig:empirical-evaluation-pac-f1-arith}, \ref{fig:empirical-evaluation-pac-f1-json}, \ref{fig:empirical-evaluation-pac-f1-dyck} and \ref{fig:empirical-evaluation-pac-f1-date}:

\begin{description}
    \item [The pseudo-uniform random sampler] Using the pseudo-uniform random sampler as a random sampling oracle always leads to F1-scores of 0 -- the lowest possible value. As explained previously, this result is expected, because the pseudo-uniform random sampler produces hardly any positive examples. Thus, PL* does not adequately learn what kinds of strings it should \emph{accept} -- and so the hypothesis produced generally recognises a very small subset of the target language.

    \item[The parameters $\varepsilon$ and $\delta$] When PL* is run with the prefix query-based sampler, \emph{smaller} values of $\varepsilon$ and $\delta$ generally lead to higher F1-scores (for fixed $p$). Within the PAC framework, smaller values of $\varepsilon$ and $\delta$ intuitively lead to a ``better'' hypothesis being produced, in terms of classification accuracy. Our results suggest that, if $\varepsilon$ and $\delta$ are reduced (and thus classification accuracy is optimised), F1-score also increases, thus demonstrating a clear connection between the PAC guarantee and F1-score.
    
    \item[The parameter $p$] When PL* is run with the prefix query-based sampler, \emph{larger} values of $p$ generally lead to higher F1-scores (for fixed $\varepsilon$ and $\delta$). This result is somewhat surprising because, intuitively, one might expect that using a prefix query-based sampler which produces a \emph{balanced} sample of positive and negative examples (as is the case when $p = 0.5$) would lead to higher F1-scores than using one which produces an imbalanced sample; that is, a sample containing mostly negative examples ($p = 0.05$), or a sample containing solely positive examples ($p = 1$). The trend we observe in practice suggests that positive examples are in some way more useful to PL* than negative examples as it refines its hypothesis.
\end{description}

Importantly, when PL* is run with the prefix query-based sampler, with $\varepsilon$ and $\delta$ small and $p$ large, it is able to produce a hypothesis with ``high'' F1-score (i.e., 0.9 or greater), for every target language.

Overall, the results presented in this section reveal how the parameters of the PAC framework ($\varepsilon$ and $\delta$) as well as the sampling distribution used relate to the metric of F1-score. Several clear trends are evident, revealing that the PAC framework is of use even when the metric we are interested in differs from the metric it directly optimises. It is also encouraging to see that PL*, when run within the PAC framework, does have the ability to achieve high F1-scores, across various languages of practical interest. In other words, PL*, when adapted to use only queries that can be answered by a parser, can still efficiently and accurately learn the target language.

\section{Discussion} \label{chap:discussion}

In this section, we provide a final discussion of the results achieved thus far (Section~\ref{discussion-plstar}), followed by some potential directions for future work (Section~\ref{discussion-future-work}).

\subsection{The PL* Algorithm}\label{discussion-plstar}

Recall from ~\Cref{chap:problem-statement} that the aim of this work was to understand how existing language learning algorithms could be modified to exploit prefix queries, to either learn more efficiently or learn a broader class of languages. Furthermore, upon designing such a modification, we aimed to show how it could be used to learn the language of a parser. We were able to successfully address this aim in two phases, as summarised below.

We began by designing the PL* algorithm, the first known language learning algorithm to use prefix queries and a novel modification of the classical L* algorithm proposed by \citet{RN6}. Whilst L* uses equivalence and membership queries, PL* uses equivalence and \emph{prefix} queries. We proved that PL* always performs less queries than L*, when both are tasked with learning the same non-dense language. Furthermore, we showed empirically that PL* significantly outperforms L* in terms of both runtime and space usage. This therefore demonstrated how exploiting the additional information given by prefix queries over membership queries -- which is precisely what PL* does -- can lead to substantial improvements in efficiency.

Subsequently, because prefix queries naturally arose from parsers (that is, many parsers are naturally capable of answering prefix queries), we aimed to show how our proposed algorithm could be applied to the problem of learning the language of a parser. In order to apply PL* to this problem, we needed to adapt it to a more practical but constrained setting in which an equivalence oracle was unavailable, since parsers are generally incapable of answering equivalence queries. Specifically, instead of providing PL* with a prefix oracle and an equivalence oracle, we gave it access to a prefix oracle and a \emph{random sampling oracle} which sampled from a fixed distribution $D$ over $A^*$, within the probably approximately correct (PAC) framework. Though running PL* within the PAC framework (in the manner proposed by \citet{RN1}) gave rise to a strong theoretical guarantee, we observed that the metric it optimised was not so useful in practice since it was dependent upon the sampling distribution $D$. We identified F1-score as a more suitable alternative, and, despite the apparent mismatch between this metric and the metric optimised by PAC, we were able to demonstrate a clear relationship between the distribution $D$, the parameters of the PAC framework (namely $\varepsilon$ and $\delta$) and F1-score. Specifically, we were able to identify how $D$, $\varepsilon$ and $\delta$ could be chosen to allow PL* to achieve high F1-scores across a range of target languages common in software engineering. Thus, we observed that PL*, when applied to the problem of learning the language of a parser, was still able to accurately learn the target language, as quantified by F1-score.

\subsection{Future Work}\label{discussion-future-work}

There are five main directions for future work in this area. Firstly, in Section~\ref{comparison-with-lstar}, we proved that PL* always performs less queries than L*, when both are tasked with learning the same non-dense language (Theorem~\ref{theorem-non-dense}). However, we do not prove a similar result for space usage; doing so would be a possible direction for future work. Furthermore, in Section~\ref{comparison-with-lstar}, we note that, in practice, the improvements given by PL* over L* may be greater than those stated in Theorem~\ref{theorem-non-dense}. Hence, another possible direction for future work would be providing a \emph{tighter} theoretical analysis of the improvements given by PL* over L*.

Whilst we have theoretically and empirically compared PL* with L* when each was run with an \emph{exact} equivalence oracle (Sections~\ref{comparison-with-lstar} and \ref{empirical-evaluation-exact}), we have not yet compared PL* run with a \emph{random sampling} oracle with other learning algorithms designed to model a parser, for example GLADE \citep{RN51} and ARVADA \citep{RN52}. Comparing PL* with such learning algorithms in terms of F1-score would be useful in improving our understanding of where our approach sits in the context of other algorithms designed to model a parser. Note however that, though we haven't yet performed an empirical comparison with GLADE and ARVADA, we can still \emph{qualitatively} compare PL* with these algorithms. One aspect of PL* which differentiates it from both GLADE and ARVADA is that PL* does not require input examples, whilst GLADE and ARVADA both do. Specifically, PL* only requires blackbox access to the parser whose language it is modelling, whilst GLADE and ARVADA both require blackbox access to the parser \emph{and} a set of positive input examples. Thus, PL* can be applied even when input examples are unavailable; for example, if we only have access to the program binary.

With PL* achieving significant performance improvements over the classical algorithm it is based on (namely L*), it is conceivable that modifying a more recent language learning algorithm to use prefix queries in place of membership queries may yield similar improvements. Though L* is an important contribution to the field of language learning, more efficient algorithms for learning regular languages with membership and equivalence queries have been proposed \citep{RN71, RN38}. Hence, it would be worthwhile to consider modifying those more recent algorithms in a similar way to how we have modified L*.

PL*, like many language learning algorithms, learns a deterministic finite automata (DFA) representation of the target language. However, in practice, some languages cannot be adequately modelled by a DFA, and so it would be desirable to learn a more expressive context-free grammar (CFG) representation. Hence, a possible area of future work would be extending this work to the learning of CFGs with prefix queries. Perhaps one could consider an existing learning algorithm which learns CFGs \citep{RN48, RN49, RN63, RN51, RN52}, and modify it to use prefix queries (in a similar way to how we modified L*). Alternatively, it may be possible to somehow generalise the restricted-nesting DFA representation returned by PL* into a CFG which allows \emph{unrestricted} nesting, using various heuristics or otherwise.

\section{Conclusion}
\label{sec:conclusion}

This work sought to investigate how the prefix query could be applied in the field of language learning. We identified that prefix queries arise naturally from the field of software engineering, in that many practical parsers are capable of answering them. Our key contribution was the PL* algorithm, the first known language learning algorithm to make use of the prefix query and a novel modification of the classical L* algorithm proposed by \citet{RN6}. We proved that PL* always makes less queries than L* when both learners are tasked with learning the same non-dense language (many languages of practical interest are in fact non-dense). Furthermore, we demonstrated empirically that PL* significantly outperforms L* both in terms of runtime and space usage. We then described how the PL* algorithm could be used to model the language of a parser by adapting it to a more practical setting, in which an equivalence oracle was unavailable. Notably, in this new setting, PL* only required access to a prefix oracle; it did not require labelled examples or any other information about the target language. We showed that, even in this more restricted setting, PL* could still efficiently and accurately learn a range of languages of practical interest. With various possible directions for future work, PL* lays the foundation for the application of prefix queries to the field of language learning, taking one step closer towards bridging the gap between formal language theory and software engineering.
\bibliographystyle{ACM-Reference-Format}
\bibliography{arxiv2025-lstar}

\end{document}